\documentclass[review]{elsarticle}
\usepackage{array}
\usepackage{amsthm,amsmath,amssymb}
\usepackage{graphicx}
\usepackage{listings}
\usepackage{epstopdf}

\usepackage{a4wide}
\RequirePackage{natbib}
\RequirePackage[colorlinks,citecolor=blue,urlcolor=blue]{hyperref}

\usepackage[usenames,dvipsnames,svgnames,table]{xcolor}

\journal{.}

\newtheorem{lemma}{Lemma}
\newtheorem{theorem}{Theorem}

\newtheorem{proposition}{Proposition}

\biboptions{authoryear}
\bibpunct{(}{)}{;}{a}{}{;}

\makeatletter
\def\ps@pprintTitle{%
  \let\@oddhead\@empty
  \let\@evenhead\@empty
  \let\@oddfoot\@empty
  \let\@evenfoot\@oddfoot
}
\makeatother

\begin{document}

\newcommand\gai[1]{{#1}}\newcommand\ggai{}

\renewcommand{\thefigure}{\arabic{figure}}
\renewcommand{\thetable}{\arabic{table}}

\renewcommand{\tilde}{\widetilde}
\renewcommand{\hat}{\widehat}
\newcommand{\F}{{\mathcal F}}
\newcommand{\I}{{\mathbf I}}
\newcommand{\X}{{\mathbf X}}
\newcommand{\x}{{\mathbf x}}
\newcommand{\Y}{{\mathbf Y}}
\newcommand{\y}{{\mathbf y}}
\newcommand{\Z}{{\mathbf Z}}
\newcommand{\z}{{\mathbf z}}
\newcommand{\U}{{\mathbf U}}
\newcommand{\Q}{{\mathbf Q}}
\newcommand{\bu}{{\mathbf u}}
\newcommand{\V}{{\mathbf V}}
\newcommand{\bv}{{\mathbf v}}
\newcommand{\W}{{\mathbf W}}
\newcommand{\w}{{\mathbf w}}
\newcommand{\bi}{{\mathbf i}}
\newcommand{\bj}{{\mathbf j}}
\newcommand{\E}{{\rm E}}
\newcommand{\Var}{{\rm Var}}
\newcommand{\Cov}{{\rm Cov}}
\newcommand{\bmu}{{\boldsymbol \mu}}
\newcommand{\um}{\underline{m}}
\newcommand{\R}{{\mathbf R}}
\newcommand{\tr}{{\text{\rm tr}}}
\newcommand{\bbu}{{\mathbf u}}
\newcommand{\bbx}{{\mathbf x}}
\newcommand{\al}{{\alpha}}
\newcommand{\te}{{\theta}}
\newcommand{\tetan}{{\hat\theta_n}}
\newcommand\tetanmk{{\hat{\theta}_{n}^{(k)}}}
\newcommand{\CC}{\mathbb{C}}
\newcommand{\cC}{\mathcal{C}}
\newcommand{\la}{\lambda}\newcommand{\si}{\sigma}\newcommand{\Si}{\Sigma}
\newcommand{\us}{\underline{s}}
\newcommand{\de}{\delta}

\newcommand{\ep}{\varepsilon}

\newcommand\bbT{\mathbf{T}}
\newcommand\tx{\widetilde{\x}}

\newcommand{\bbal}{{\mbox{\boldmath$\alpha$}}}
\newcommand{\bbbe}{{\mbox{\boldmath$  \beta$}}}
\newcommand{\bbde}{{\mbox{\boldmath$\delta$}}}
\newcommand{\bbDe}{{\mbox{\boldmath$\Delta$}}}
\newcommand{\bbeps}{{\mbox{\boldmath$\varepsilon$}}}
\newcommand{\bbgam}{{\mbox{\boldmath$  \gamma$}}}
\newcommand{\bbGam}{{\mbox{\boldmath$  \gamma$}}}
\newcommand{\bbla}{{\mbox{\boldmath$\lambda$}}}
\newcommand{\bbLa}{{\mbox{\boldmath$\Lambda$}}}
\newcommand{\bbmu}{{\mbox{\boldmath$\mu$}}}
\newcommand{\bbnu}{{\mbox{\boldmath$\nu$}}}
\newcommand{\bbom}{{\mbox{\boldmath$\omega$}}}
\newcommand{\bbPsi}{{\mbox{\boldmath$\Psi$}}}
\newcommand{\bbpsi}{{\mbox{\boldmath$\psi$}}}
\newcommand{\bbphi}{{\mbox{\boldmath$\phi$}}}
\newcommand{\bbPhi}{{\mbox{\boldmath$\Phi$}}}
\newcommand{\bbsi}{{\mbox{\boldmath$\sigma$}}}
\newcommand{\bbSi}{{\mbox{\boldmath$\Sigma$}}}
\newcommand{\bbtau}{{\mbox{\boldmath$\tau$}}}
\newcommand{\bbth}{{\mbox{\boldmath$\theta$}}}
\newcommand{\bbTh}{{\mbox{\boldmath$\Theta$}}}
\newcommand{\bbxi}{{\mbox{\boldmath$\xi$}}}
\newcommand{\bbXi}{{\mbox{\boldmath$\Xi$}}}
\newcommand{\bbze}{{\mbox{\boldmath$\zeta$}}}
\newcommand{\bbeta}{{\mbox{\boldmath$\eta$}}}

\begin{frontmatter}

\title{\large\bf On structure testing for component covariance
  matrices of a high-dimensional mixture}




\author{Weiming Li\fnref{label1}}
\ead{li.weiming@shufe.edu.cn}
\fntext[label1]{Weiming Li's research is partly supported by National
  Natural Science Fundation of China, No.\ 11401037 and Program of
  IRTSHUFE.
}
 \address{School of Statistics and Management, Shanghai University of Finance and Economics, Shanghai, China}

 \author{Jianfeng Yao\fnref{label2}}
 \ead{jeffyao@hku.hk}
 \fntext[label2]{Jianfeng Yao's research is partly  supported by HKSAR   Research Grants Council grant No. 17332416.}
 
 \address{Department of Statistics and Actuarial Sciences, The
   University of Hong Kong, Hong Kong SAR}
 
\date{}
\begin{abstract}

By studying the family of $p$-dimensional scale mixtures, this paper
shows  for the first time a non trivial example where the
eigenvalue distribution  of the corresponding sample covariance matrix 
{\em does not converge} to the celebrated  Mar\v{c}enko-Pastur law. 
A different and new limit is found and characterized.
The reasons of failure of the Mar\v{c}enko-Pastur limit in this
situation are found to be a strong dependence between the $p$-coordinates of the mixture.
Next, we address the problem of testing whether the mixture has a
spherical covariance matrix.
{\ggai 
To analize the traditional John's type test we establish a novel and general CLT for linear statistics of eigenvalues of the sample covariance matrix. 
It is shown that the John's test and its recent high-dimensional extensions both fail for high-dimensional mixtures, precisely due to the different spectral limit above.  As a remedy, a new test procedure is constructed  afterwards for the sphericity hypothesis.
This test is then applied to identify the covariance structure in
model-based clustering. It is shown that the test has much higher
power than the widely used ICL and BIC criteria in detecting non
spherical component covariance matrices of a high-dimensional mixture. 
} 
\end{abstract}


%

\end{frontmatter}
 

%
\section{Introduction}\label{sec:intro}

Let $\phi({\scriptscriptstyle\bullet};\,\bbmu, \bbSi)$ be the density function of a $p$-dimensional
normal  distribution with mean $\bbmu$ and covariance matrix $\bbSi$.
A $p$-dimensional vector $\bbx\in\mathbb{R}^p$ is a 
{\em multivariate normal  mixture}  (MNM) if its density function has the form
\begin{equation}
  \label{eq:fmn}
  f(\bbx)= \sum_{j=1}^K \alpha_j \phi(\bbx;\,\bbmu_j, \bbSi_j). 
\end{equation}
Here the $(\alpha_j)$ are the $K$ mixing weights and
$(\bbmu_j,\bbSi_j)$ are the parameters of the $j$th normal component.  
Such finite mixture models have a long history; yet they continue to
attract considerable attention  in recent years due to their wide
usage  in  high-dimensional data analysis such as 
in
pattern recognition,  signal and image
processing, machine learning in bioinformatics, to name a few. 
The popularity of an MNM  is largely due to the fact that by
construction the distribution can be interpreted  as a  mixture of $K$
sub-populations (or groups, clusters)  with respective parameters
$(\bbmu_j,\bbSi_j)$ and this interpretation is particularly relevant
for clustering or classifying heterogeneous data. 
For detailed account on these models, we refer to the 
monographs \citet{McLachlan00} and \citet{Fruhwirth06}.

When the number of features $p$ in $\bbx$ is large compared to the
number $n$ of available samples from an MNM, 
the inference of a general  MNM  becomes intricate.
The reason is that   the number of free 
parameters of an MNM model is $ K (p+2)(p+1)/2-1$ which  explodes
quadratically with the dimension $p$.
In order to have a concrete picture of this inflation,
the numbers of parameters in 
four  particular
MNMs are detailed in Table~\ref{tab:mnm} below
\citep{Bouveyron07}.
We see from the table that the full MNM will require as many as  5303
parameters when  50 variables of interest and 4 clusters are involved although 50 is
a quite  small number in today's big data era.  
Even for a homogeneous MNM, 1478 parameters are still needed which almost excludes 
any standard procedure like the maximum likelihood
estimation. This highlights that 
 inference of a high-dimensional  MNM remains an open and challenging problem
even in the homogeneous case.  

\begin{table}[htp]
  \centering
  \caption{Four standard covariance structures  in an MNM with their
    number of parameters. Here
    $a=Kp+K-1$ denotes the 
    number of parameters in  $(\alpha_j)$ and $(\bbmu_j)$. 
    \label{tab:mnm}}
  {\small
  \begin{tabular}[hb]{llll}
    \hline 
    Model     &  $\bbSi_j$'s    & Number of parameters  &   
    [case of  $(K,p)=(4,50)$]
    \\  \hline 
  Full MNM    &   Unrestricted  &  $a+Kp(p+1)/2$    &  [ 5303 ]           \\

  Scale MNM  &  Proportional:  $\bbSi_j=\sigma_j^2\bbSi$  &
  $a+p(p+1)/2+ K-1$ & [ 1481 ]\\   
  Homogeneous  MNM &   Identical:   $\bbSi_j\equiv \bbSi$  
  &  $a+p(p+1)/2$  & [ 1478 ] \\
  Spherical MNM  &    Spherical:  $\bbSi_j=\sigma_j^2  \I_p $  &
  $a+K$  & [ 207 ]
  \\
  \hline\hline 
  \end{tabular}
  }
\end{table}

Meanwhile, such difficulty for inference  is not that surprising in lights of recent
developments of high-dimensional statistics.  
Consider either the case there was no mixture
at all, that is $K=1$, $\bbmu_j\equiv\bbmu$ and $\bbSi_j\equiv\bbSi$,
or the case of homogeneous MNM, 
$K>1$ and $\bbSi_j\equiv\bbSi$.
The inference of both models  contains the estimation of
a high-dimensional covariance
matrix $\bbSi$.  This  estimation problem has been widely studied
recently and it is
well-known that typically no  {\em consistent} estimation exists
for such  a large covariance matrix  $\bbSi$ without further drastic
constrains on its structure 
\citep{BickLev08b}. Therefore, high-dimensional MNM cannot be
consistently identified in general when the dimension is large compared to the
sample size.

Notice that the literature contains extensive proposals for 
reduction of the model dimension by using 
some parsimonious MNM models where the 
$K$ component covariance matrices $(\bbSi_j)$ are restricted to
certain structure.   
The common approach introduces such restricted structure on
the eigenvalues and the eigenvectors of these component matrices 
\citep{Banfield93,Fraley98,Fraley02}.  For example, \citet{Bensmail96}
and \citet{Bouveyron07} proposed 14 and 28 such restricted models,
respectively. These restricted models also 
include  the so-called   {\em mixtures of factor analyzers}
\citep[Chapter 8]{McLachlan00} which are particularly popular in
handling high-dimensional data. These mixtures specify that 
$\bbSi_i=\bbLa_i\bbLa_i' + \bbPsi_i$ where 
$\bbLa_i$ is a $p\times d_i$ loading matrix with $d_i\ll p$ and
$\bbPsi_i$ a diagonal matrix representing the base component of
$\bbSi_i$. 

The other lesson learnt from recent developments in high-dimensional
statistics is that although the estimation and identification of a
high-dimensional covariance matrix are generically unfeasible, testing
hypotheses on their structure is indeed possible.  
Such structure testing includes equality to the unit (identity matrix), 
proportional to the unit (sphericity test), equality to a diagonal
matrix for the one-sample case, or equality between several
high-dimensional covariance matrices in the case of a multiple-sample
problem. To mention a few on this literature,
we refer to 
\cite{LW02},
\cite{BD05},
\cite{BJYZ09},
\cite{ChenZhangZhong10},
\cite{WangYao13},
\cite{Tian15}
and the review  
\cite{PaulAue14}.

In this paper we investigate the structure testing problem for the component
covariance matrices $(\bbSi_j)$ in a high-dimensional MNM. 
Precisely, we assume that the $K$ group means $(\bbmu_j)$ have been
satisfactorily identified  so that all our attention will be devoted at
the study of the $K$ component  covariance matrices $(\bbSi_j)$ and
at their structure testing.  We thus hereafter assume $\bbmu_j\equiv
0$. 
The $p$-variate population $\x$ is assumed to be a 
{\em scale  mixture} 
of  the form 
\begin{equation}\label{mix-model}
  \x=w\bbT_p\z,
\end{equation}
where {\ggai  $w$ is a scalar mixing random variable, $\bbT_p \in \mathbb R^{p\times p}$ is a positive definite matrix, assuming $\tr(\bbT_p^2)=p$ so that
$w$ and $\bbT_p$ can be identified in the model, and 
$\z=(z_1,\ldots,z_p)'$ is a set of i.i.d.\ random variables, independent of $w$, having zero mean and unit variance.
The mean and the covariance matrix of the scale mixture
\eqref{mix-model} are
\begin{equation}
\label{eq:mom}
\E(\x)=0\quad\text{and}\quad \Cov(\x)=E(w^2) \bbT_p^2:=\bbSi_p,
\end{equation}
respectively.}
Notice  that here $w$  is a  latent label variable, and if $w$ takes
values in a finite set of $K$ values, 
say $\{\sigma_1,\ldots,\sigma_K\} $ 
with respective probability $\{\alpha_j\}$, the mixture $\x$
becomes a  
{\em finite scale mixture}. If moreover the $z_i$'s are i.i.d.\ standard normal, then $\x$
reduces to the scale MNM in Table~\ref{tab:mnm} 
with mixing weight $(\alpha_j)$ and 
components covariance matrices $\bbSi_j=\sigma_j^2 \bbT_p^2$
($1\le j\le K$).  

{\ggai 
The scale mixture \eqref{mix-model} can also be regarded as an
extension of the standard elliptical model \citep{FZ90} where the
vector $\z$ is assumed to be uniformly distributed on the unit
sphere in $\mathbb R^p$. This extension allows the population to
possess a heavier or lighter tail by controlling the fourth moment of
$\z$. In \cite{E10}, the scale mixture plus a non-zero mean vector was
studied in the context of portfolio optimization. A major difference
here is  that Karoui's model makes Gaussian assumption on $\z$, while
our model allows non-Gaussian distributions for $\z$. Recently,
\cite{X17} proposed a similar model in the study of high-dimensional
integrated covolatility matrices. Their sample data can be modeled as
$\x_i=w_i\bbT_p\z_i$ which has the same form as the scale mixture
\eqref{mix-model}, but $w_i$ in their model is a non-random function
of the index $i$ and $(\z_i)$ is again a sequence of i.i.d.\ Gaussian
vectors. 
}

Let  $\x_1,\ldots,\x_n$ be  a sample from the mixture \eqref{mix-model}. 
Our approach is based on the spectral properties of the sample
covariance matrix    
\begin{equation}\label{sample-cov}
  B_n=\frac{1}{n}\sum_{i=1}^n\x_i\x_i'.
\end{equation}
Let $(\lambda_j)_{1\le j\le p}$ be its eigenvalues, referred as {\em sample eigenvalues}. 
The {\em empirical spectral distribution} (ESD) of $B_n$ is by
definition 
$
F_n=p^{-1}\sum_{j=1}^p\delta_{\lambda_j}, 
$
{\ggai  where and throughout the paper  $\delta_b$ denotes  the Dirac measure at the point $b$.  }

Properties of eigenvalues of large sample covariance matrices have
been extensively studied in random matrix theory (RMT).
Consider a $p$-variate population $\widetilde\x$  of the form 
\begin{equation}\label{linear-trans}
  \tx=\sigma\bbT_p \z, 
\end{equation}
where $\bbT_p$ and $\z$ are  as before but $\sigma$ is now a constant 
unlike the random mixing variable $w$ in \eqref{mix-model}.
Let $\tx_1,\ldots,\tx_n$ be  a sample from the population and 
denote  the corresponding sample covariance matrix by
$\widetilde B_n=n^{-1}\sum_{i=1}^n \tx_i\tx_i'$.
To further simplify the discussion, let us assume that $\bbT_p=\I_p$
so here $\tx=\sigma\z$ and $\Cov(\tx)=\sigma^2 \I_p$.
It has been known since \citep{MP67,Silv95} 
that 
when $p$ and $n$ grow to infinity proportionally such that $c=\lim p/n
>0$,
the ESD of the sample covariance matrix $\widetilde B_n$ 
will converge to the celebrated 
Mar\v{c}enko-Pastur  law  (MP law) with parameter $(c,\sigma^2)$, i.e.,  
{\ggai $ \nu_{(c,\sigma^2)}(dx)=f(x) dx + (1-1/c)\delta_0(dx) 1_{\{c> 1\}}$} with the density function
$
f(x)=(2\pi c\sigma^2x)^{-1} \sqrt{(b-x)(x-a)}1_{[a,b]}(x),
$
where 
$a=\sigma^2(1-\sqrt c)^2$ and  $b=\sigma^2(1+\sqrt c)^2$.
Consider next the scale mixture \eqref{mix-model} where we let 
also $\bbT_p=\I_p$, that is,  
\begin{equation}
  \label{eq:mix2}
  \x = w \z.
\end{equation} 
Then $\Cov(\x)=\sigma^2\I_p$ is spherical as before 
with here $\sigma^2=\E(w^2)$: in particular the   $p$ coordinates 
of $\x$ are {\em uncorrelated}.
A {\em striking finding} from this paper is that despite a same
spherical covariance matrix $\sigma^2\I_p$, the sample covariance 
matrix $B_n$ from the mixture \eqref{eq:mix2} is very different 
of the sample covariance matrix  $\tilde B_n$ from the linear transformation model
\eqref{linear-trans}. In particular, the ESD of $B_n$ will converge to
a distribution which is not the MP law  $\nu_{(c,\sigma^2)}$. 
An illustration of this difference is given in 
Figure~\ref{mp:mix}.

\begin{figure}[h]
\begin{minipage}[t]{0.33\linewidth}
\includegraphics[width=2in]{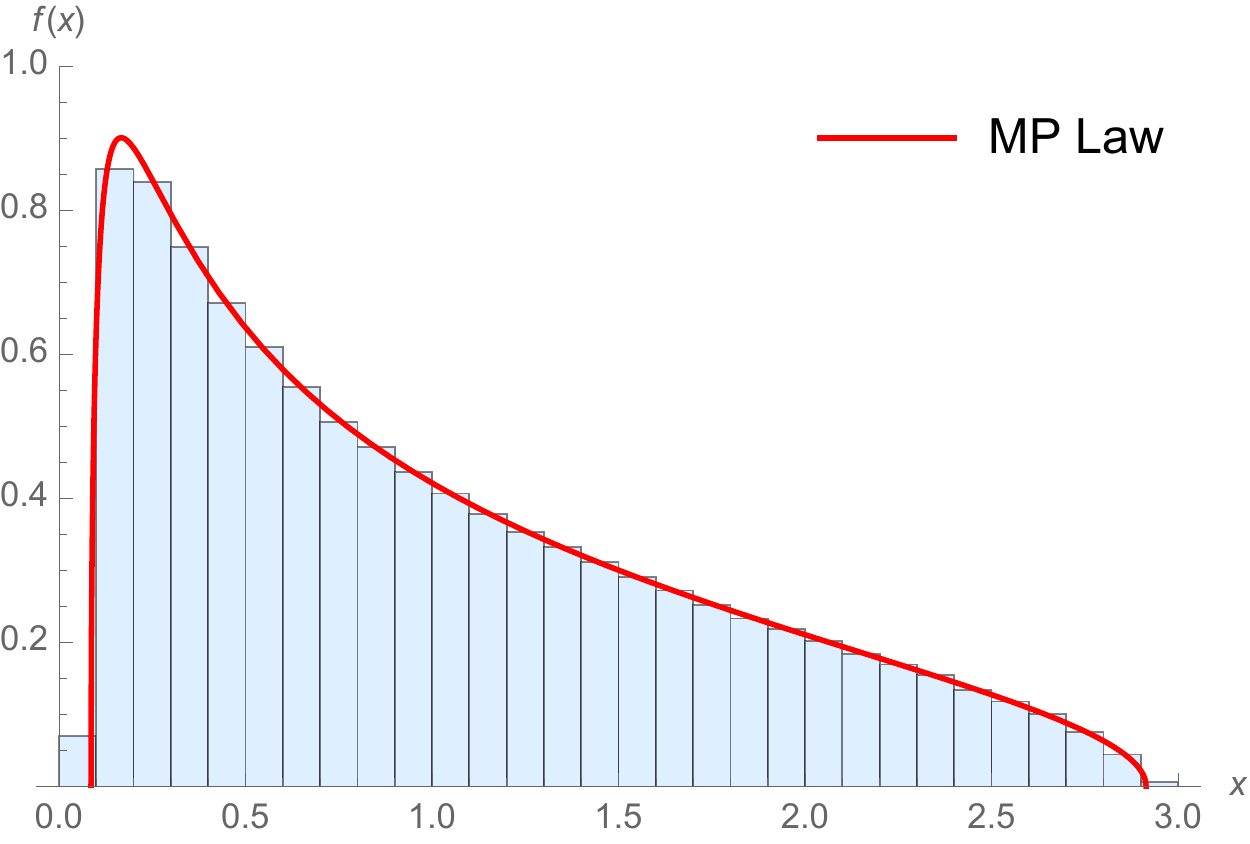} 
\end{minipage}%
\begin{minipage}[t]{0.33\linewidth}
\includegraphics[width=2in]{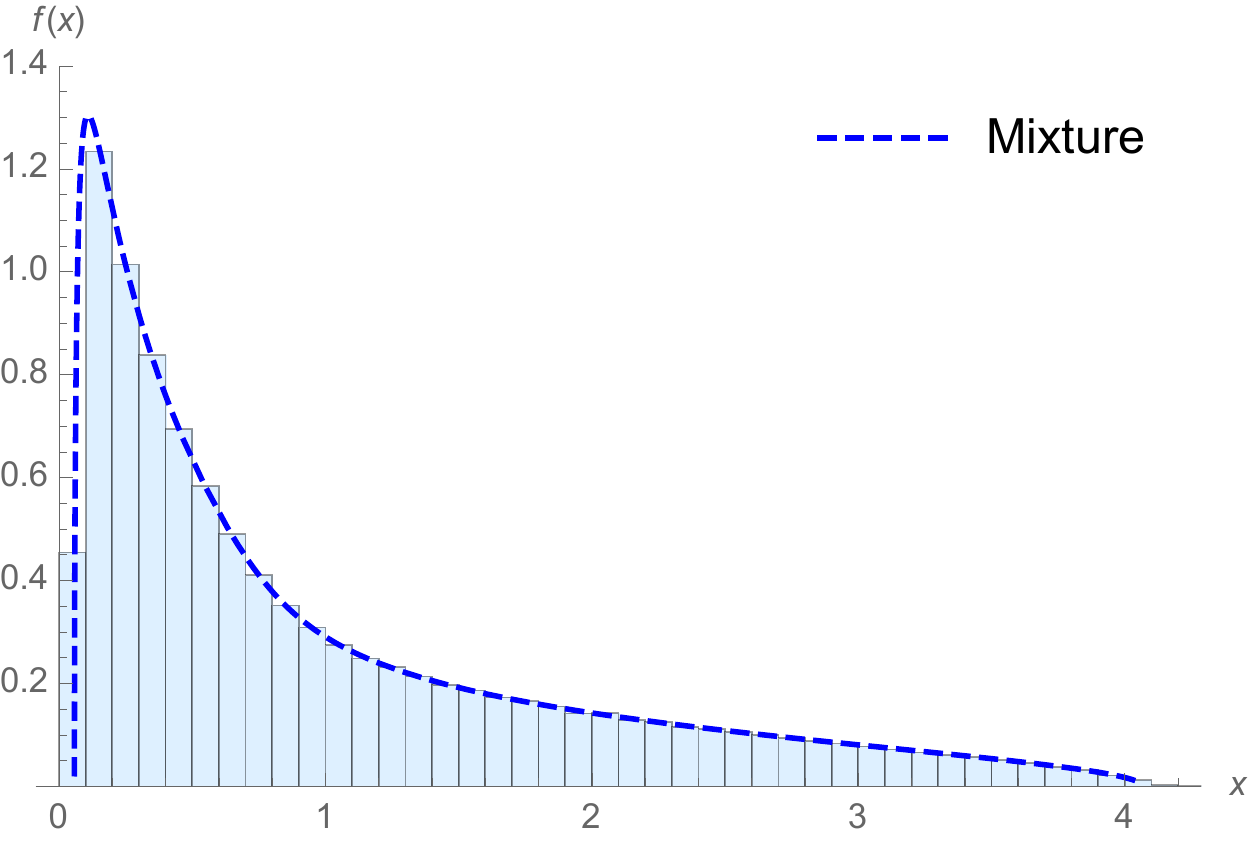}
\end{minipage} 
\begin{minipage}[t]{0.33\linewidth}
\includegraphics[width=2in]{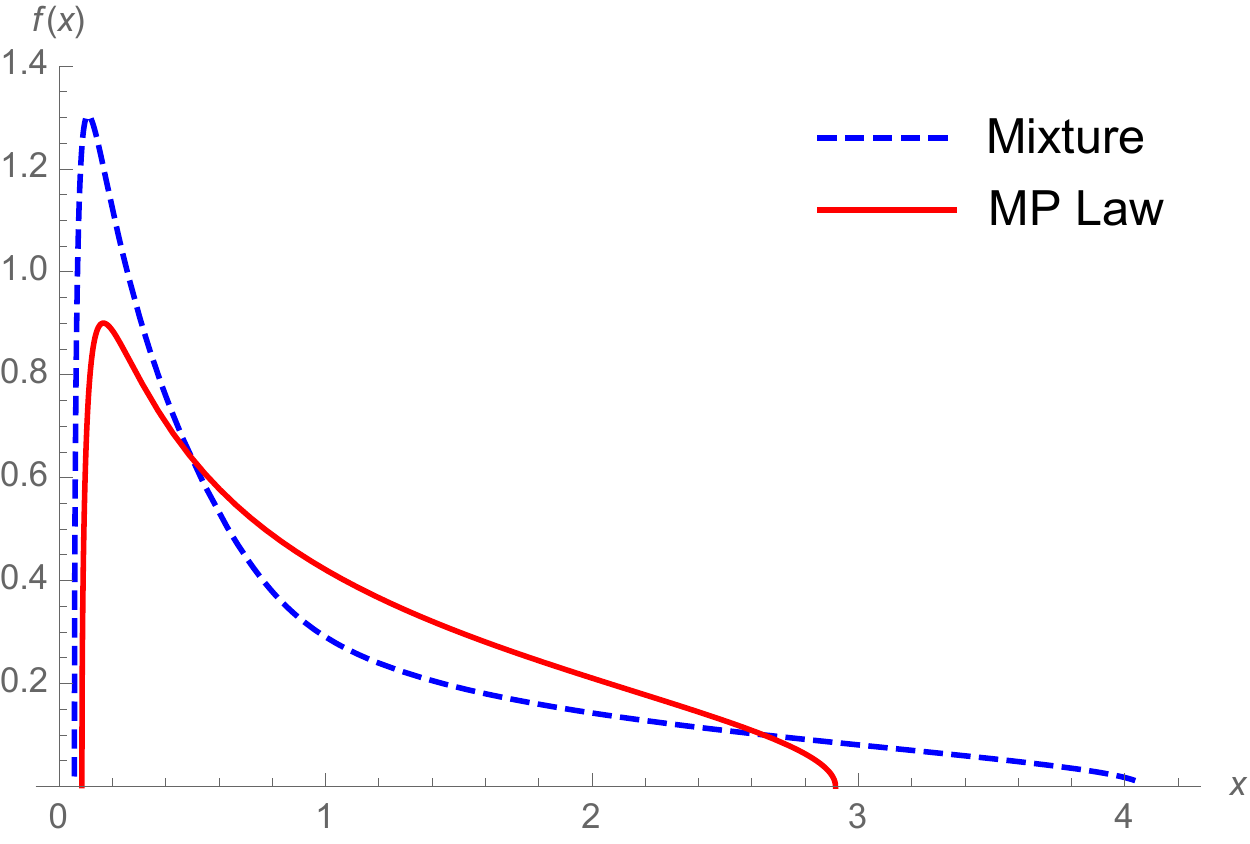}
\end{minipage} 
\caption[]{Left panel: 
  histogram of eigenvalues of the sample covariance matrix
  $\tilde B_n$ from a population $  {\bf x}={\bf z}$ with
  i.i.d. standardized coordinates.
  The LSD is the Mar\v{c}enko-Pastur  law (solid line) with support 
  [0.0858,2.9142].~~
  Middle panel: 
  histogram of eigenvalues of  sample covariance matrix
  $B_n$ from a scale MNM $ {\bf x}=w{\bf z}$ with
  density function $f(\bbx)= 0.25\phi(\bbx;\,0, 2.5I_p)+0.75\phi(\bbx;\,0,
  0.5I_p)$ whose covariance matrix is also unit. 
  The LSD (dashed line) is not the Mar\v{c}enko-Pastur  law and has
  support $[0.0576,4.0674]$.
  ~~ Right panel compares  the two LSDs.
  Both histograms  used dimensions $(p,n)=(500,1000) $ 
  with eigenvalues collected from 100 independent replications.
  \label{mp:mix}} 
\end{figure}

The failure of the Mar\v{c}enko-Pasture law for the scale mixture
\eqref{mix-model} can be explained by the strong dependence 
between the $p$  uncorrelated  coordinates of the mixture. Indeed,  
 \cite{BaiZhou08}  proved that the MP law always  holds for a
 population ${\y}$ with 
 {\em weakly dependent} coordinates in the following sense: 
for any sequence of symmetric matrices 
$\{A_p\}$ bounded in spectral norm,   
\begin{equation}\label{BZ-condition}
  \Var(\y' A_p\y)=o(p^2).
\end{equation}
In particular, the linear transformation model $\tx$ in \eqref{linear-trans}
has weakly dependent coordinates: 
 indeed one can  easily show that 
$
\Var(\tx' A_p\tx)\leq \kappa p||\bbT_p'A_p\bbT_p||^2
$
where $\kappa$ is a constant (function of $\E(z_i^4)$),
and this bound  is of order $O(p)$ since the
sequence $(A_p)$ is bounded.  Therefore the MP law
applies for the sample covariance matrix 
$\widetilde B_n$. Now we show that the scale mixture $\x$ of
\eqref{mix-model} has {\em strongly dependent} coordinates. 
Indeed, 
for $A_p=(\bbT_p^2)^{-1}$  one easily finds (by conditioning on $w$)
that  
$
  \Var(\x'A_p\x)
  =p\E(w^4)\Var(z_1^2)+p^2\Var(w^2),
$
which is at least of order $p^2$ (unless  the mixing variable $w$ is
degenerated).
Therefore, it does not satisfy Bai-Zhou's weak dependence condition 
\eqref{BZ-condition}.  Notice that other weak dependence condition
guaranteeing a limiting  MP law is also available as in 
\cite{Banna15}, but again this does not apply to the scale mixture
\eqref{mix-model}.

To summarize, we have reached  the following conclusions.
(i)
 Structure testing on the component covariance matrices
  $(\bbSi_j)$ of a high-dimensional mixture will involve ultimately
  the study of the eigenvalues of the sample covariance matrix $B_n$
  in \eqref{sample-cov};
~(ii). 
Very unfortunately, existing results on high-dimensional covariance
matrices from the existing random matrix theory do not apply to $B_n$.

The main contributions of this paper are presented   as follows. First in
Section~\ref{sec:main}, by using tools of random matrix theory, 
we develop new asymptotic results on the eigenvalues of the sample
covariance matrix $B_n$. This includes (i) the characterisation of the
limits of the ESD $F_n$ of $B_n$ under fairly general moment
conditions and (ii), a central limit theorem  for linear spectral statistics of
the form $\int f(x) dF_n(x)$ for a class of smooth test function $f$.
{\ggai Then in  Section~\ref{sec:sph}, 
we apply this general theory to analyze the failure of the John's test for the hypothesis that the population $\x$  is  
a spherical mixture. As a byproduct, we find that the John's statistic can test whether a spherical population is a mixture or not. In the light of this study, a new test procedure is then put forward for general spherical hypothesis. In Section \ref{sec:clust}, the two tests are numerically examined in the identification of the covariance structure in model-based clustering.
Section \ref{sec:data} present a microarray data analysis on their covariance structure.} 
All the technical proofs of the results of the paper
are gathered in Section~\ref{sec:proofs}.
The paper has also an on-line supplementary file which includes the following material:
(i) a consistent estimator for the parameters of 
a centered spherical mixture (which is an exceptional case where the
estimation can be carried out completely);
(ii)
procedures for numerical evaluation of the
density function and the support set of the LSD 
of the sample covariance matrix found in Section~\ref{sec:main}.
\gai{
Finally, computing codes for reproduction of the  numerical results of
the paper and the related
data sets are availabe  at {\small\texttt{http://web.hku.hk/\~{}jeffyao/papersInfo.html}}.
}


%
\section{High-dimensional theory for eigenvalues of $B_n$}
\label{sec:main}

\subsection{Non standard limit of the sample eigenvalue distribution}

Our interest is to study the convergence of the ESD sequence $(F_n)$ in high-dimensional frameworks,  as defined in the following assumptions.
Throughout the paper, the distribution of the squared mixing variable
$w^2$ is denoted as $ G$ and 
referred as {\em Mixing Distribution} (MD).

\medskip 
\noindent{\em Assumption}   (a). \quad
  The sample and population sizes $n,p$ both tend to infinity with their ratio $c_n=p/n\to c \in(0,\infty)$.

\medskip
\noindent{\em Assumption}   (b). \quad
There are two independent arrays of i.i.d.\ random variables
  $(z_{ij})_{i,j\geq 1}$ and $(w_{i})_{i\geq 1}$, satisfying 
\begin{equation}\label{mom-condition}
\E(z_{11})=0,\quad \E(z^2_{11})=1,\quad \E(z^4_{11})<\infty,
\end{equation}
such that  for each $p$ and $n$ the observation vectors can be represented as $\x_i= w_i \bbT_p\z_i$  with $\z_i=(z_{i1},\ldots,z_{ip})'$, $i=1,\ldots,n.$

\medskip
\noindent{\em Assumption}   (c). \quad
The spectral distribution $H_p$ of the matrix $\bbT_p^2$
weakly converges to a probability distribution $H$, as $p\rightarrow\infty$, referred as {\em Population Spectral Distribution} (PSD).

\medskip
\noindent{\em Assumption} (d).\quad
The support set $S_{ G}$ of the MD $ G$ is bounded above and
from below, that is $S_{ G}\subset [a,b]$ for some  $0<a<b<\infty.$ 


The LSD of $B_n$ will be derived under Assumptions (a)-(b)-(c) while
Assumption (d) is  required when establishing the CLT for linear spectral statistics.
Recall that the Stieltjes transform of a probability measure $\mathcal P$, supported on $S_{\mathcal P}\subset \mathbb R$, is defined as 
$$m_{\mathcal P}(z)=\int\frac{1}{x-z}d\mathcal P(x),\quad z\in\mathbb C\setminus S_{\mathcal P}.$$ 

\begin{theorem}\label{lsd}
Suppose that Assumptions (a)-(c) hold. Then, almost surely, the empirical spectral distribution $F_n$ of $B_n$ converges in distribution to a probability distribution $F^{c, G,H}$
whose Stieltjes transform $m=m_{F^{c, G, H}}(z)$ is a solution
to the following system of equations, defined on the upper complex plane $\CC^+$,
\begin{eqnarray}  \label{lsd-eqs}
\begin{cases}
 zm(z)=-1+\int \frac{p(z)t}{1+cp(z)t}d G(t),\\
 zm(z)=-\int \frac{1}{1+q(z)t}dH(t),\\
 zm(z)=-1-zp(z)q(z),
\end{cases}
\end{eqnarray}
where $p(z)$ and $q(z)$ are two auxiliary analytic functions.
The solution is also unique in the set 
$$
\{m(z): -(1-c)/z+cm(z)\in \CC^+, ~zp(z)\in \CC^+, ~q(z)\in \CC^+,~z\in \CC^+\}.
$$
\end{theorem}

The proof is given in Section~\ref{ssec:lsd}. {\ggai To clarify the
  role of the two auxiliary functions in \eqref{lsd-eqs}, we express the sample covariance matrix as $B_n=\bbT_pZ_n\Sigma_GZ_n'\bbT_p/n$ and denote its companion matrix as $\underline{B}_n=\Sigma_G^{1/2}Z_n'\bbT_p^2Z_n\Sigma_G^{1/2}/n$, where $Z_n=(\z_1,\ldots,\z_n)$ and $\Sigma_G=diag(w^2_1,\ldots,w^2_n)$ is a diagonal matrix. Then $p(z)$ and $q(z)$ are actually the limits of 
$\tr[\bbT_n^2(B_n-zI)^{-1}]/p$ and $\tr[\Sigma_G(\underline{B}_n-zI)^{-1}]/n$, respectively, see Section 4.3.2 in \cite{Z06}.
}

Important special cases include the following.
If $H= G=\delta_1$ or just $ G=\delta_1$,  the system  \eqref{lsd-eqs} reduces to 
a single equation which characterizes the standard MP law
$\nu_{(c,1)}$ or the generalized MP law \citep{Silv95}.
A case of particular interest is 
for $H=\delta_1$ where the equations reduce to 
\begin{equation}  \label{mp3}
  z  =  - \frac1 {m}  +  \int\!\frac{t}{1+ctm} d G(t)~,
\end{equation}
{\ggai  with  $p(z)=m(z)$ and $q(z)=-(1+zm(z))/(zm(z))$.
Equation \eqref{mp3}}  defines a new type of LSD corresponding to a scale-mixture population with spherical covariance matrix. 

We run a small simulation experiment to illustrate the LSD from a
spherical mixture whose LSD is given in \eqref{mp3}.  Notice that the
density function of the LSD as well as its support set can be
determined using standard tools from random matrix theory; they are
detailed in Section B of the supplementary file.
The MD $ G$ is set to be $ G=0.5\delta_1+0.5\delta_9$ and
the dimensional ratio is $c=0.5$ or 2. Samples of $(z_{ij})_{p\times n}$ are drawn from standard normal $N(0, 1)$ with $(p,n)=(500, 1000)$ and $(1000, 500)$, respectively.

This mixture is made up of two normal distributions $N(0,I_p)$ and
$N(0,9I_p)$ with equal weights.  The sample eigenvalues may form one
or two clusters depending on  the value of $c$.
Theoretically, the critical value for the spectrum separation is $c=1.1808$ under this MD. Therefore the support $S_F$ is a unique interval for $c=0.5$ and consists of two separate intervals for $c=2$. 
The results are shown in Figure \ref{esd-lsd}
where we see that the empirical histograms match perfectly with their
 limiting density curves predicted by Theorem~\ref{lsd}.
\begin{figure}[h]
\begin{minipage}[t]{0.5\linewidth}
\includegraphics[width=2.3in,height=1.5in]{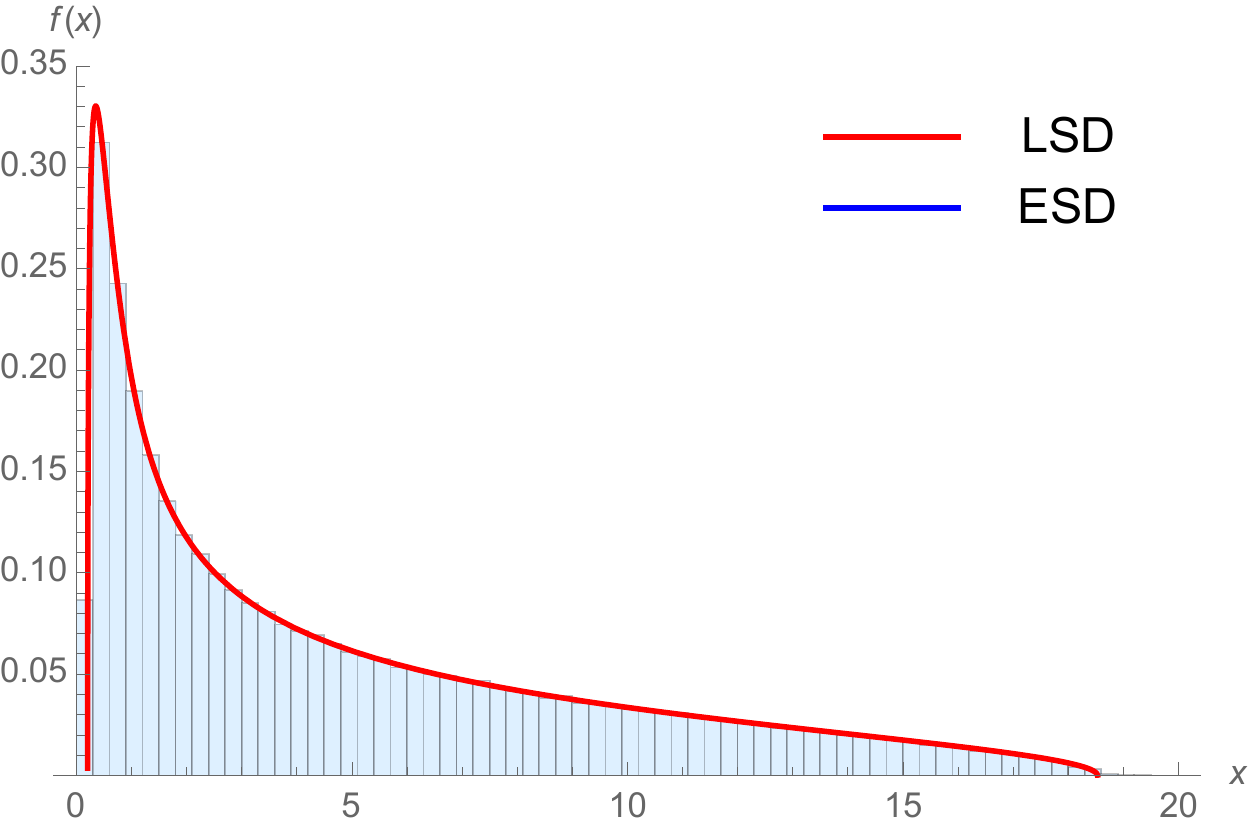} 
\end{minipage}%
\begin{minipage}[t]{0.5\linewidth}
\includegraphics[width=2.3in,height=1.5in]{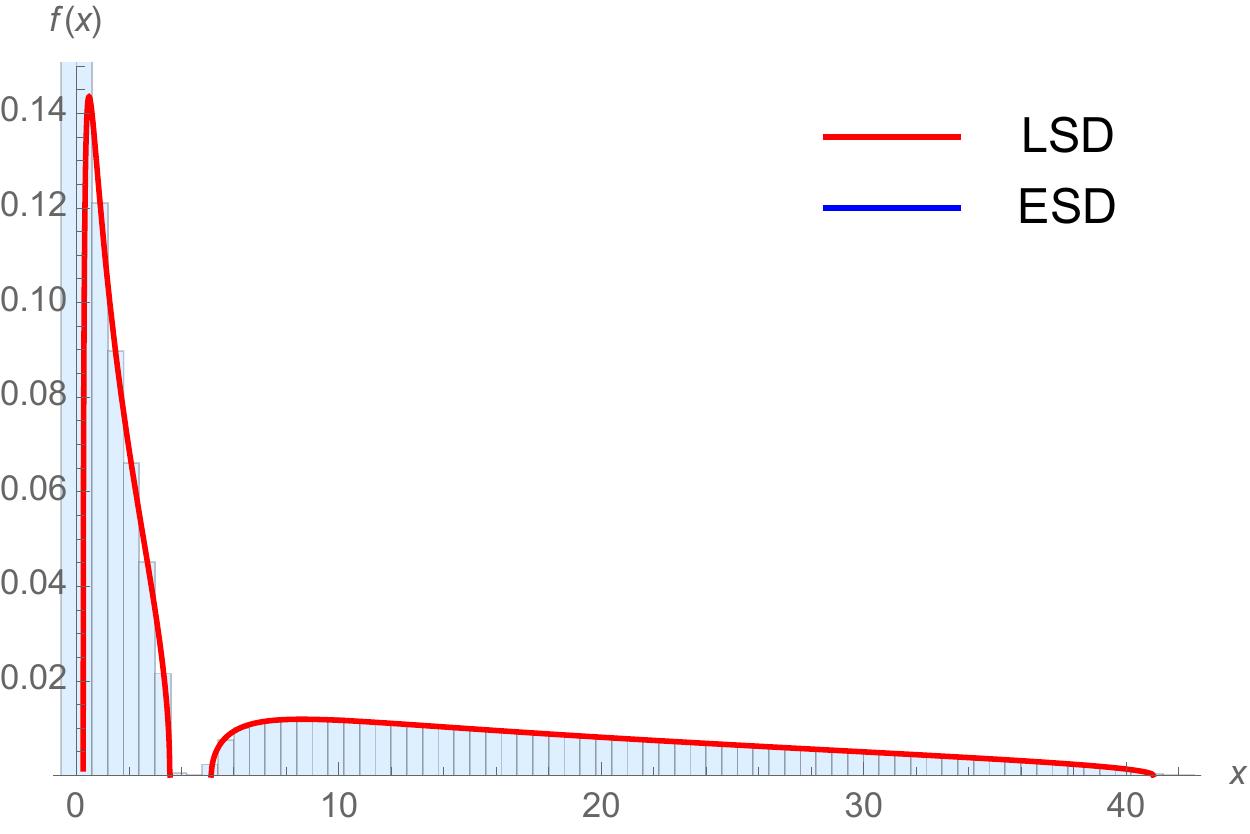}
\end{minipage} 
\caption[]{Comparison between sample eigenvalues (histogram)  and their limit
  density (solid curve). Left panel: $(p,n,c)=(500, 1000, 0.5)$ with a unique
  support interval [0.2,18.5].  Right panel: $(p,n,c)=(1000, 500, 2)$
  and the support is $\{0\}\cup[0.26,3.56]\cup[5.14, 41.04]$.}
  \label{esd-lsd}
\end{figure}

\subsection{CLT for linear spectral statistics of $B_n$}

In this section, we study the fluctuation of linear spectral statistics (LSS) of the sample covariance matrix $B_n$ under the mixture model with a spherical covariance matrix. The LSS are quantities of the form
$$\frac{1}{p}\sum_{j=1}^pf(\lambda_j)=\int f(x)dF_n(x)$$
where  $f$ is a function on $[0, \infty)$. In \cite{BS04} and
\cite{PZ08}, the LSS under their settings are proved to be
asymptotically normal distributions. As said in Introduction, the
central limit theorem studied in these papers all assume the linear
transformation form in \eqref{linear-trans}, and thus is not applicable to the present model of scale mixtures. 

Let $ G_n$ be the empirical distribution generated by $w_1^2,\ldots,w_n^2$ which correspond to the sample data $\x_1,\ldots,\x_n$.
Let also $F^{c_n,  G_n}$ and $F^{c_n,  G}$ be the LSDs as defined in \eqref{mp3} for $F^{c, G}$ but with the parameters $(c, G)$ replaced by $(c_n,  G_n)$ and $(c_n,  G)$, respectively. Notice that $F^{c_n,  G_n}$ is a random measure while $F^{c_n,  G}$ is deterministic. The aim here is to study the fluctuation of 
\begin{eqnarray*}
\int f(x) d\F_n(x):=\int f(x) d(F_n(x)-F^{c_n, G}(x)),
\end{eqnarray*}
which has a decomposition 
\begin{eqnarray}\label{LSS}
\int f(x) d\F_n(x)=\int f(x)d\F_{n1}(x)+\int f(x)d\F_{n2}(x),
\end{eqnarray}
where 
\[ \F_{n1}(x)=F_n(x)-F^{c_n, G_n}(x) \quad \text{ and}\quad  
\F_{n2}(x)=F^{c_n, G_n}(x)-F^{c_n, G}(x).
\]
We show that the first term in \eqref{LSS} converges in distribution to a normal variable at the rate of $1/n$, while the second term converges in distribution to another normal variable at the rate of $1/\sqrt{n}$.

\begin{theorem}\label{clt1}
Suppose that Assumptions (a)-(d) hold. Let $f_1,\ldots, f_k$ be functions on $\mathbb R$ analytic on an open interval containing
$\left[aI_{(0,1)}(1/c)(1-\sqrt{1/c})^2,b(1+\sqrt{1/c})^2\right]$. 
Let $\Delta=E(z_{11}^4)-3$ be the kurtosis coefficient.
Then the random vector 
\begin{equation*}
n\left(\int f_1(x)d\F_{n1}(x),\ldots, \int f_k(x)d\F_{n1}(x)\right)\xrightarrow{D}N_k(\mu, \Gamma_1),
\end{equation*}
where the mean vector $\mu=(\mu_j)$ is
\begin{eqnarray*}\label{muj}
\mu_j&=&-\frac{1}{2\pi{\rm i}}\oint_{\mathcal C_1}\frac{f_j(z) m^3(z)\int t^2(1+ct m(z))^{-3}d G(t)}{(1-c\int  m^2(z)t^2(1+ct m(z))^{-2}d G(t))^2}dz\nonumber\\
&&-\frac{\Delta}{2\pi{\rm i}}\oint_{\mathcal C_1}\frac{f_j(z) m^3(z)\int t^2(1+ct m(z))^{-3}d G(t)}{1-c\int  m^2(z)t^2(1+ct m(z))^{-2}d G(t)}dz
\end{eqnarray*}
and the covariance matrix $\Gamma_1=(\gamma_{1ij})$ has  entries
\begin{eqnarray*}\label{bij}
\gamma_{1ij}&=&-\frac{1}{2\pi^2c^2}\oint_{\mathcal C_2}\oint_{\mathcal C_1}\frac{f_i(z_1)f_j(z_2)}{( m(z_1)- m(z_2))^2} m'(z_1) m'(z_2)dz_1dz_2\nonumber\\
&&-\frac{\Delta}{4\pi^2c}\oint_{\mathcal C_2}\oint_{\mathcal C_1}f_i(z_1)f_j(z_2)
\left(\frac{d^2}{dz_1dz_2}
\int\frac{t^2 m(z_1) m(z_2)d G(t)}{(1+ct m(z_1))(1+ct m(z_2))}\right)dz_1dz_2.
\end{eqnarray*}
Here
the contours $\mathcal C_1$ and $\mathcal C_2$ are simple, closed, non-overlapping, and taken in the positive direction in the complex plane, each enclosing the support of $F^{ c,  G}$.
\end{theorem}

\begin{theorem}\label{clt2}
Under the assumptions of Theorem  \ref{clt1}, the random vector 
\begin{equation*}
\sqrt{n}\left(\int f_1(x)d\F_{n2}(x),\ldots, \int f_k(x)d\F_{n2}(x)\right)\xrightarrow{D}N_k(0, \Gamma_2),
\end{equation*}
where the covariance matrix $ \Gamma_2=(\gamma_{2ij})$ has  entries
\begin{eqnarray*}
  \gamma_{2ij}
&=&\frac{1}{4\pi^2}\oint_{\mathcal C_1}\oint_{\mathcal C_2} \frac{f_i(z_1)f_j(z_2) m'(z_1)m'(z_2)(z_1-z_2)}{c(m(z_1)-m(z_2))}dz_1dz_2\\
&&-\frac{1}{4\pi^2}\oint_{\mathcal C_1}\oint_{\mathcal C_2}\frac{ f_i(z_1)f_j(z_2) m'(z_1)m'(z_2)}{cm(z_1)m(z_2)}dz_1dz_2\\
&&+\frac{1}{4\pi^2}\oint_{\mathcal C_1}\oint_{\mathcal C_2}\frac{ f_i(z_1)f_j(z_2) m'(z_1)m'(z_2)(1+z_1m(z_1))(1+z_2m(z_2))}{m(z_1)m(z_2)}dz_1dz_2.
\end{eqnarray*}
Here  the contours $\mathcal C_1$ and $\mathcal C_2$ are as defined in Theorem \ref{clt1}.
\end{theorem}

\begin{proposition}\label{clt}
Under the assumptions of Theorem \ref{clt1}, the random vector 
\begin{equation}\label{lss}
\sqrt{n}\left(\int f_1(x)d\F_{n}(x),\ldots, \int f_k(x)d\F_{n}(x)\right)\xrightarrow{D}N_k(0,  \Gamma_2),
\end{equation}
where the covariance matrix $\Gamma_2$ is defined in Theorem \ref{clt2}.
\end{proposition}

{\ggai Theorem \ref{clt1} follows from Theorem 1.4 in
  \cite{PZ08}. 
  A brief outline of the main arguments of its proof is given in
    Section \ref{ssec:clt1}.
  The proof of Theorem~\ref{clt2} is given in Sections \ref{ssec:lem}
  and \ref{ssec:Theorem3}.
}
Proposition \ref{clt} is a direct consequence of the two theorems.

These CLTs  demonstrate that the limiting distributions of $\sqrt{n}\int f(x)d\F_{n}(x)$ and $\sqrt{n}\int f(x)d\F_{n2}(x)$ coincide since their difference $\sqrt{n}\int f(x)d\F_{n1}(x)$ is of order $O_p(1/\sqrt{n})$.
Notice that the asymptotic limit in Theorem \ref{clt1} is stochastically independent of the sequence $( G_n)$. Therefore the two components in \eqref{LSS} are asymptotically independent. Consequently,
the CLT in \eqref{lss} always {\em underestimates} the variation and
the absolute mean of corresponding statistics in finite
samples. Fortunately, such differences can be estimated using Theorem
\ref{clt1}, and their incorporation  to the CLT in Proposition \ref{clt} leads to a finite-sample corrected CLT
\begin{equation}\label{clt-corrected}
\sqrt{n}\left(\int f_1(x)d\F_{n}(x),\ldots, \int f_k(x)d\F_{n}(x)\right)\stackrel{\cdot}{\sim}N_k(\mu/\sqrt{n}, \Gamma_1/n+   \Gamma_2).
\end{equation}
This corrected CLT is deemed to provide a  better approximation than the CLT
in Proposition~\ref{clt} in finite-sample situations.

\subsection{Application of the CLTs to moments of sample eigenvalues}

Among all the LSS, the moments of sample eigenvalues are ones of the
most important statistics. They  have been well studied in the
literature again under the linear transformation
model~\eqref{linear-trans}, 
see \cite{PZ08}, \cite{BCY10}, \cite{LY14}, \cite{Tian15}, and the references therein. 
{\ggai In our context, the $j$th moment statistic can be expressed as 
\begin{eqnarray}\label{beta:nk}
\hat \beta_{nj}=\int x^jdF_n(x),\quad j\in \mathbb N,
\end{eqnarray}
and its limit is related to the following four quantities
\begin{eqnarray*}
\beta_{nj}=\int x^jdF^{c_n, G_n}(x),\quad \beta_j=\int x^jdF^{c_n, G}(x),\quad  \gamma_{nj}=\int t^jd G_n(t),\quad   \gamma_j=\int t^jd G(t),
\end{eqnarray*}
which are the $j$th moments of corresponding measures.  From Theorem \ref{lsd} and the convergence of the empirical distribution $G_n$, we may
conclude that $\hat \beta_{nj}- \beta_j\xrightarrow{a.s}0$,
$\beta_{nj}- \beta_j\xrightarrow{a.s.}0$, and $\gamma_{nj}\xrightarrow{a.s.} \gamma_j$, for $j\ge 1$.
Moreover, the deterministic sequence $( \beta_j)$ can be explicitly expressed in terms of $( \gamma_j)$ as  
\begin{eqnarray}\label{beta-gamma}
\beta_j=c_n^j\sum(\gamma_1/c_n)^{i_1}(  \gamma_2/c_n)^{i_2}\cdots(  \gamma_{j}/c_n)^{i_{j}}\phi(i_1,\ldots,i_{j}), ~~j\geq1,
\end{eqnarray}
where $\phi(i_1,\ldots,i_{j})=j!/[i_1!\cdots i_{j}!(j+1-i_1-\cdots -i_{j})!]$
and the sum runs over the following partitions of $j$:
$$
(i_1,\ldots,i_{j}): j=i_1+2i_2+\cdots+ji_{j}, \quad i_l\in\mathbb N.
$$}
These recursive formulae are well known in random matrix theory \citep{BCY10}; they 
can also be  easily derived from the equation \eqref{mp3}. 
The formulae also hold if $( \beta_j,  \gamma_j)$ is
replaced by $( \beta_{nj},  \gamma_{nj})$.

The joint CLT for  the first $k$  moments $(\hat \beta_{nj})_{1\leq j\leq k}$ can be derived by
applying Theorems \ref{clt1} and \ref{clt2} to functions $f_j(x)=x^j,
1\le j\le k$. A major task here is to determine  the
integrals involved in their limiting mean vector and covariance matrix which, as shown below,
can be converted to the calculation of  derivatives of certain
functions. 
These functions are 
\begin{eqnarray*}
P(z)=-1+\int\frac{tzd G(t)}{1+ctz},\quad
Q(z)=\int\frac{t^2d G(t)}{(1+ctz)^3},\quad\text{and}\quad
R(z)=1-c\int\frac{(zt)^2d G(t)}{(1+ctz)^2}.
\end{eqnarray*}

\begin{proposition}\label{clt-moments-1}
Suppose that Assumptions (a)-(d) hold and let $\Delta=E(z_{11}^4)-3$. Then the random vector 
\begin{equation*}
n\left(\hat\beta_{n1}- \beta_{n1},\ldots, \hat\beta_{nk}- \beta_{nk}\right)\xrightarrow{D}N_k(v,\Psi_1),
\end{equation*}
where the mean vector $v=(v_j)$ has coordinates  $v_1=0$ and
\begin{eqnarray*}
v_j=
\frac{1}{(j-2)!}\left[P^j(z)Q(z)\left(\frac{1}{R(z)}+\Delta\right)\right]^{(j-2)}\bigg|_{z=0},&2\leq j\leq k,
\end{eqnarray*}
and the covariance matrix $\Psi_1=(\psi_{1ij})$ has  entries
\begin{eqnarray*}
\psi_{1ij}&=&\frac{2}{c^2}\sum_{l=0}^{i-1}(i-l)u_{i,l}u_{j,i+j-l}\\
&&+\frac{\Delta}{c}\int \frac{t^2}{(i-1)!(j-1)!}
\left[\frac{P^i(z)}{(1+ctz)^2}\right]^{(i-1)}\bigg|_{z=0}\left[\frac{P^j(z)}{(1+ctz)^2}\right]^{(j-1)}\bigg|_{z=0}d G(t),\label{bij-res}
\end{eqnarray*}
where $u_{s,t}$ is the coefficient of $z^t$ in the Taylor expansion of $P^s(z)$.

\end{proposition}

\begin{proposition}\label{clt-moments-2}
Suppose that Assumptions (a)-(d) hold, then the random vector 
\begin{equation*}
\sqrt{n}\left( \beta_{n1}- \beta_1,\ldots, \beta_{nk}- \beta_k\right)\xrightarrow{D}N_k(0,  \Psi_2),
\end{equation*}
where the covariance matrix $ \Psi_2=(\psi_{2ij})$ has  entries
$$
 \psi_{2ij}=\frac{1}{c}\left(\sum_{l=0}^iu_{i+1,l}u_{j,i+j-l}-\sum_{l=0}^{i-1}u_{i,l}u_{j+1,i+j-l}+u_{i,i}u_{j,j}\right)-(u_{i,i}+u_{i+1,i})(u_{j,j}+u_{j+1,j}),
$$
where $u_{s,t}$ is defined in Proposition \ref{clt-moments-1}.
\end{proposition}

\begin{proposition}\label{clt-moments}
Suppose that Assumptions (a)-(d) hold, then the random vector 
\begin{equation}\label{bk}
\sqrt{n}\left(\hat\beta_{n1}- \beta_1,\ldots\hat\beta_{nk}- \beta_k\right)\xrightarrow{D}N_k(0, \Psi_2),
\end{equation}
where the covariance matrix $\Psi_2$ is defined in Proposition \ref{clt-moments-2}.
\end{proposition}

Proposition \ref{clt-moments-1} is a  straightforward application of 
Theorem~\ref{clt2} in this paper in combination with 
Lemma 2 in \cite{Tian15}  and Theorem 1 in
\cite{Q17}. We thus omit its proof here. 
The proof of Proposition \ref{clt-moments-2} is
presented in Section~\ref{ssec:clt-moments}. 
Proposition \ref{clt-moments}
easily follows from Propositions \ref{clt-moments-1} and
\ref{clt-moments-2}. Next,
similarly to the correction  in
\eqref{clt-corrected} for finite samples,  we have the corrected CLT
\begin{equation}\label{clt-moments-corrected}
\sqrt{n}\left(\hat\beta_{n1}- \beta_1,\ldots\hat\beta_{nk}- \beta_k\right)\stackrel{\cdot}{\sim}N_k(v/\sqrt{n}, \Psi_1/n+  \Psi_2).
\end{equation}
Simulation results show that this
corrected CLT indeed provides  a generally  more accurate approximation in finite
sample situations. 

As an example we consider  
the fluctuation of the first two moments.
>From Propositions \ref{clt-moments-1}-\ref{clt-moments}, their finite-sample corrected CLT is
\begin{eqnarray}\label{clt-beta}
  \sqrt{n}\left(
    \begin{matrix}
      \hat \beta_{n1}-  \beta_{1}\\ 
      \hat \beta_{n2}-  \beta_{2}
    \end{matrix}\right)\stackrel{\cdot}{\sim} 
  N\left(\left(
      \begin{matrix}
        0\\
        v_2/\sqrt{n}
      \end{matrix}
    \right),
    \left(\begin{matrix}
        \psi_{111}/n+ \psi_{211}&\psi_{112}/n+ \psi_{212}\\
        \psi_{112}/n+ \psi_{212}& \psi_{122}/n+ \psi_{222}
      \end{matrix}
    \right)\right)
\end{eqnarray}
where the parameters are respectively
\begin{eqnarray*}
&& \beta_1= \gamma_1,\quad  \beta_2=c_n \gamma_2+ \gamma_1^2,\quad v_2=(1+\Delta)  \gamma_2,\quad \psi_{111}=(2+\Delta)  \gamma_2/c,\quad  \psi_{211}=  \gamma_2-  \gamma_1^2,\\
&&\psi_{112}=2 (2 + \Delta) ( \gamma_1 \gamma_2/c+ \gamma_3),\quad  \psi_{212}=c( \gamma_3- \gamma_1 \gamma_2)+2( \gamma_1 \gamma_2- \gamma_1^3),\\
&&\psi_{122}=4 ((2 + \Delta)   \gamma_1^2   \gamma_2/c + 8(2 + \Delta)   \gamma_1   \gamma_2 +    4(  \gamma_2^2 + c (2 + \Delta)   \gamma_4)),\\
&& \psi_{222}=c^2(  \gamma_4-  \gamma_2^2)+4c  \gamma_1  \gamma_3+4(1-c)  \gamma_1^2  \gamma_2-4  \gamma_1^4\label{psi-2}.
\end{eqnarray*}

We have run a simulation experiment for various scale mixtures to
check the finite-sample properties of these two moment estimators 
$\hat \beta_{n1}$ and $\hat \beta_{n2}$. The results are reported in
Appendix C of the supplementary material. Their asymptotic normality
is well confirmed in many tested situations. 
The results also reveal that the correction of the CLT in \eqref{clt-moments-corrected} is significant. For example, under a tested scenario (with  standardized  chi-square distributed
$z_{ij}$'s), 
the limiting distribution of $\sqrt{n}(\hat \beta_{n2}- \beta_2)$
 is  $N(0, 39.32)$, while
the corrected distribution is $N(3.48, 48.88)$ for dimensions
$(p,n)=(200,400)$: the difference is quite significant. 

{\ggai It is worth noting that the CLTs  established in this
  section are based on the scale-mixture model \eqref{mix-model} with
  mean zero. These results can be extended without  difficulties to a
  general population with a non-zero mean $\mu$,
  i.e. $\x=\mu+w\bbT_p\z$.  The required  adjustments are the
  replacements of the 
  covariance matrix $B_n$ and the dimensional ratio $c_n=p/n$ by 
  $B_n^*=\sum_{j=1}^n(\x_j-\bar{\x})(\x_j-\bar{\x})'/(n-1)$ and 
  $c_n^*=p/(n-1)$, respectively. All the CLTs then remain valid 
    following the substitution principle established in 
  \cite{Z15}.} 
\section{Testing the sphericity of  a high-dimensiona  mixture}
\label{sec:sph}

In this section, using the results developed in Section
\ref{sec:main}, we theoretically investigate the reliability of John's
test for the sphericity of a covariance matrix 
\citep{John(1972)} and its  high-dimensional corrected version
\citep{WangYao13} when  the underlying distribution 
is a  high-dimensional mixture. Our findings show  neither 
John's test nor  its corrected version is thus valid any more. 
This motivates  us to propose  a new test procedure.

\subsection{\gai{Failure of the high-dimensional John's test for mixtures}}

In \cite{John(1972)}, the author proposed a locally most powerful
invariant test for the sphericity of a normal population covariance
matrix. Let $\Sigma_p$ be the population covariance matrix. The
sphericity hypothesis to test is 
~$
  H_{0}:\Sigma_p= \sigma^2 I_p$~ 
for some unknown positive constant $\si^2$.
John's  test statistic is
$$
U=\frac{\sum_{i=1}^p(\lambda_i-\sum\lambda_i/p)^2/p}{(\sum_{i=1}^p\lambda_i/p)^2}=\frac{\hat \beta_{n2}}{\hat \beta_{n1}^2}-1,
$$
where $(\lambda_i)$ are the sample eigenvalues and $\hat \beta_{nj}$ is their $j$th empirical moment for $j=1,2$. When the dimension $p$ is assumed fixed, \cite{John(1972)} proved that, under $H_0$,
\begin{equation}\label{clt-john}
nU-p\xrightarrow{D}\frac{2}{p}\chi^2_f-p,
\end{equation}
as $n\rightarrow\infty$, where $\chi^2_f$ denotes the chi-square distribution with degrees of freedom $f=p(p+1)/2-1$. 

{\ggai 
This test has been extended to the high-dimensional framework in a
series of recent works such as 
\cite{LW02,BD05,Srivastavaetal.(2011), WangYao13} and \cite{Tian15}.}  
These extensions have a common assumption that the population follows  the linear transformation model \eqref{linear-trans} with $\sigma\bbT_p=\Sigma_p^{1/2}$ and satisfy the moment conditions in \eqref{mom-condition}. Under the null hypothesis, it is proved that
\begin{equation}\label{clt-wang}
nU-p\xrightarrow{D}N(\Delta+1,4),
\end{equation}
as $(n,p)\rightarrow\infty$, where $\Delta=E(z_{11}^4)-3$ is the kurtosis of $z_{11}$. We can see that the distribution $2\chi^2_f/p-p$ in \eqref{clt-john} tends to the normal distribution $N(1,4)$ if $p\rightarrow\infty$, which is consistent with the CLT in \eqref{clt-wang} in the normal case ($\Delta=0$). 

However, when the population follows the mixture model defined in
\eqref{mix-model}  with a spherical covariance matrix
$\Sigma_p=\sigma^2 I_p$,  
the tests based on \eqref{clt-john} and \eqref{clt-wang} will fail and
reject the sphericity  hypothesis with a probability close to  one for
all large $(n,p)$. This phenomenon can be  intuitively explained by
the point limit of their test statistic. Specifically, for general PSD $H$ and MD $ G$, it can be shown that 
\begin{eqnarray}\label{beta-lim}
\hat \beta_{n1}\xrightarrow{a.s.}   \gamma_1\tilde\gamma_1\quad \text{and}\quad \hat \beta_{n2}\xrightarrow{a.s.} c  \gamma_2\tilde\gamma_1^2+ \gamma_1^2\tilde\gamma_2,
\end{eqnarray}
as $(n,p)\rightarrow\infty$, 
{\ggai  where $\tilde\gamma_1=\int t dH(t)$ and $\tilde\gamma_{2}=\int t^2dH(t)$ are the first and second moments of $H$, respectively. Note that $\tilde\gamma_1\equiv1$ in our settings.} 
Therefore, the statistic 
$$U-c_n=\hat\beta_{n2}/\hat\beta_{n1}^2-1-c_n\xrightarrow{a.s.} c(  \gamma_2/  \gamma_1^2-1)+(\tilde\gamma_2/\tilde\gamma_1^2-1),$$
which is  positive when the population is a mixture.  This implies
that John's test  statistic $nU-p = n(U-c_n)$
{\ggai will tend to infinity for spherical mixture and thus
  entirely lose the control of the type I error.} Analytically, from the corrected CLT in
\eqref{clt-beta} and a standard application of the delta-method, we get
\begin{equation}\label{clt-mix-c}
\sqrt{n}\left(U-c_n  \gamma_2/  \gamma_1^2\right)\stackrel{\cdot}{\sim}N\left(\mu_U/\sqrt{n},\sigma_{1U}^2/n+ \sigma_{2U}^2\right),
\end{equation}
under $H_0$, where $\mu_U=(1+\Delta)  \gamma_2/  \gamma_1^2$ and 
\begin{eqnarray*}\label{sigma-u2}
\sigma_{1U}^2&=&4\left(c \Delta \left(  \gamma_1^2   \gamma_4-2   \gamma_1   \gamma_2   \gamma_3+  \gamma_2^3\right) + \left(2 c   \gamma_1^2   \gamma_4-4 c   \gamma_1   \gamma_2   \gamma_3 
+ 2 c   \gamma_2^3+  \gamma_1^2   \gamma_2^2\right)\right)/  \gamma_1^6,\\
 \sigma_{2U}^2&=&c^2 \left( \gamma_1^2 \left( \gamma_4- \gamma_2^2\right)+4 ( \gamma_2^3- \gamma_1  \gamma_2  \gamma_3)\right)/ \gamma_1^6.\nonumber
\end{eqnarray*}
It follows that for any fixed critical value $z_\alpha$ of the test, the type I error of $T_J^c$ is 
\begin{eqnarray}\label{type-1-p}
&&P\left(\frac{nU-p-\Delta-1}{2}>z_\alpha\right)\nonumber\\
&=&P\left(\sqrt{n}\left(U-c_n  \gamma_2/  \gamma_1^2\right)>\frac{2z_\alpha+\Delta+1+p(1-  \gamma_2/  \gamma_1^2)}{\sqrt{n}}\right)
\rightarrow1,
\end{eqnarray}
as $(n,p)\rightarrow\infty$, which describes the exploded trends of the type I error.

{\ggai 
Despite the invalidation of the corrected John's test for the sphericity hypothesis, one may be surprised to see that the statistic $nU-p$ can be employed to distinguish degenerate spherical mixture (with only one component) from general spherical mixtures. In this situation, the null distribution of the test is provided by the CLT in \eqref{clt-wang}. The power function of the test as well as its consistency are declared by \eqref{type-1-p}.  

We conclude this section by the following observation. Assume  that
the MD $ G$ degenerates to a  Dirac point measure at
$\sigma^2=\E(w^2)$, 
that is the population is not a mixture but the linear
  transformation model \eqref{linear-trans}, 
we have $\sigma_{1U}^2=4$ and $ \sigma_{2U}^2=0$, and thus the CLT in \eqref{clt-mix-c} reduces to
\begin{equation*}
\sqrt{n}\left(U-c_n\right)\stackrel{\cdot}{\sim}N\left((\Delta+1)/\sqrt{n},4/n\right),
\end{equation*}
which coincides with the one in \eqref{clt-wang} as it must be. 
This pleasant coincidence shows also that the smaller order terms 
$\mu_U/\sqrt{n}$ and $\sigma_{1U}^2/n$ appearing in the asymptotic
parameters of \eqref{clt-mix-c} have been precisely evaluated: no
other terms of similar order could be added in.}

\subsection{\gai{A sphericity test for high-dimensional mixtures}}
{\ggai 
We next develop new corrections to  John's test for
  high-dimensional mixtures.
>From the above analysis, John's test may
still be valid  if we consider the quantity $nU-p\gamma_2/\gamma_1^2$ and apply its approximated distribution in \eqref{clt-mix-c}. However, the centralization term $p\gamma_2/\gamma_1^2$} is  unknown in practice since the MD $ G$ is unobserved.  We thus have to
replace it with some suitable statistic and find the resulting
asymptotic null distribution.
To this end, we first transform the sample $(\x_1,\ldots,\x_n)$ into a
permuted counterpart $(\check\x_1,\ldots,\check\x_n)$ as follows:
for each sample $\x_i$, we randomly permute its $p$ coordinates. 
That is 
$\check \x_i=Q_i\x_{i}$ and $(Q_i)$ stand for a sequence of
independent $p\times p$ random permutation matrices. 
{\ggai Next we calculate the $k$th moment statistic
$\check\beta_{nk}:=\tr (\check B_n^k)/p$ for $k=1,2$, where
$\check B_n=\sum_{i=1}^n\check \x_i\check\x_i'/n$ is the covariance matrix of the permuted samples,} and then let   $\check
U=\check\beta_{n2}/\check\beta_{n1}^2-1$. Notice that $\check U$ is a
substitute for $c_n \gamma_2/ \gamma_1^2$ and we have
$\check\beta_{n1}\equiv\hat\beta_{n1}$.   Finally we define a  new test statistic
\begin{equation}
  \label{eq:Tn}
  T_n=\hat\beta_{n2}-\check\beta_{n2}=\frac{1}{n^2p}\sum_{i\neq j}\left((\x_i'\x_j)^2-(\check\x_i'\check\x_j)^2\right).
\end{equation}

To examine the soundness of $T_n$, we calculate its expectation: 
\begin{eqnarray}
\E(T_n)&=&\frac{1}{n^2p}\sum_{i\neq j}\E\left((\x_i'\x_j)^2-(\check\x_i'\check\x_j)^2\right)\nonumber\\
&=&\frac{n-1}{np}\left(\tr(\Sigma_p^2)-pD_{\Sigma_p}^2-p(p-1)R_{\Sigma_p}^2\right)\nonumber\\
&=&{\ggai \frac{n-1}{np}\left(\sum_{i=1}^p\left(\sigma_{ii}-D_{\Sigma_p}\right)^2+\sum_{i\neq j}\left(\sigma_{ij}-R_{\Sigma_p}\right)^2\right)}\nonumber\\
&:=&\delta_n\geq0,\label{deltan}
\end{eqnarray}
{\ggai where $\Sigma_p=\E(w^2)\bbT_p^2:=(\sigma_{ij})$,
$D_{\Sigma_p}=\sum_{i=1}^p\sigma_{ii}/p$, and $R_{\Sigma_p}=\sum_{i\neq j}\sigma_{ij}/(p(p-1))$.}  Moreover, $\delta_n=0$ if
and only if $\Sigma_p=aI_p+b{\bf 11'}$ for some parameters $a$ and
$b$: this is the commonly called {\em compound symmetric covariance
  matrix}.  However, in this case $b> 0$ and the largest eigenvalue of
$\Sigma_p$ is $a+(p-1)b\rightarrow\infty$; it can then  be easily
recognized from sample data as the largest sample eigenvalue must be
far away from the remaining eigenvalues for large $p$. We thus exclude
this case from our alternative hypothesis.
Consequently,   $\delta_n=0$  under $H_0$ and $\delta_n>0$  under
$H_1$ {\ggai (with the compound symmetric case excluded)} so that $T_n$ is  a  potentially  reasonable test statistic.

{\ggai 
\begin{theorem}\label{test-tn}
Suppose that Assumptions (a)-(d) hold. 
\begin{itemize} \label{tm}
\item[{ 1).}] Under the null hypothesis, suppose that $\E(z_{11}^8)<\infty$, then
$$
\frac{nT_n}{\sqrt{8}\hat\gamma_{n2}}\xrightarrow{D}N(0,1),
$$
where $\hat\gamma_{n2}=(\hat\beta_{n2}-\hat\beta_{n1}^2)/c_n$.

\item[{ 2).}] Under the alternative hypothesis, if $\tr(\bbT_p^4)/p\to \tilde\gamma_2=\int t^2dH(t)<\infty$ and $n\delta_n\rightarrow\infty$ then the asymptotic power of the test tends to 1,
where $\delta_n$ is the expectation of $T_n$ defined in \eqref{deltan}.
\end{itemize} 

\end{theorem}

In the first conclusion of Theorem \ref{test-tn}, $\hat\gamma_{n2}$ is
a consistent estimator of $\gamma_2$ under $H_0$. This is from Theorem
\ref{lsd} and the recursive formulae \eqref{beta-gamma}. A
substitution of $\hat\gamma_{n2}$ is
$\check\gamma_{n2}=(\check\beta_{n2}-\check\beta_{n1}^2)/c_n$. These
two estimators are equivalent under $H_0$, but
$\E(\check\gamma_{n2}-\hat\gamma_{n2})=\delta_n/c_n>0$ under
$H_1$. Therefore, the use of $\check{\gamma}_{n2}$ 
is expected to  improve the power of the test.
\gai{Noticing that the moment condition $\E(z_{11}^8)<\infty$ is used to verify the Lindeberg's condition in Martingale CLT.}
}

\subsection{Numerical results}

We report on simulations which are carried out to evaluate the
performance of the high-dimensional John's test based on the existing null distribution in
\eqref{clt-wang}, referred as $T_J^c$, and the proposed test
$T_n$ using the asymptotic null distribution of Theorem~\ref{test-tn}.
For  comparison,  we also conduct the ideal (though
impracticable)  John's test based on
\eqref{clt-mix-c}, referred as $T_J^*$, by assuming MDs $ G$
known. Results from $T_J^*$ can be regarded as a benchmark of the
sphericity test in an ideal situation.  Throughout the experiments,  the
significance level is fixed at $\alpha=0.05$ and the number of
independent replications is 10,000.  Samples of $(z_{ij})$ are drawn
from standard normal $N(0,1)$ or scale $t$, i.e.\ $\sqrt{4/6}\cdot
t_6$.

As we discussed, the test $T_J^c$ suffers serious size distortion under mixture models, we now numerically illustrate this phenomenon. 
The model is a two-components  spherical mixture of the  form  
$
  G=0.5\delta_{1}+0.5\delta_{\sigma_2^2},
$
where the parameter 
$\sigma_2^2$ ranges from 1 to 1.6 by steps of  0.05.  Specifically, the
population covariance matrix here is 
$\Sigma_p=\frac12(1+\sigma_2^2)I_p$. The dimensional setting is $(p,n)=(200,400)$.
The exploding factor in \eqref{type-1-p} is 
$$
\frac{p}{\sqrt{n}}\left(1-\frac{  \gamma_2}{  \gamma_1^2}\right)=10\left(\frac{1-\sigma_2^2}{1+\sigma_2^2}\right)
$$
ranging from 0 to -2.13. 
Results are plotted in Figure
\ref{type-1-error}, where the circled line (in blue) shows empirical sizes and
the raw  line (in red) is the theoretical one based on \eqref{clt-mix-c}. We can see  that the empirical sizes grow from 0.05 to 1, which are perfectly fitted by the theoretical curve.

\begin{figure}[h]
\begin{minipage}[t]{0.5\linewidth}
\includegraphics[width=2.3in]{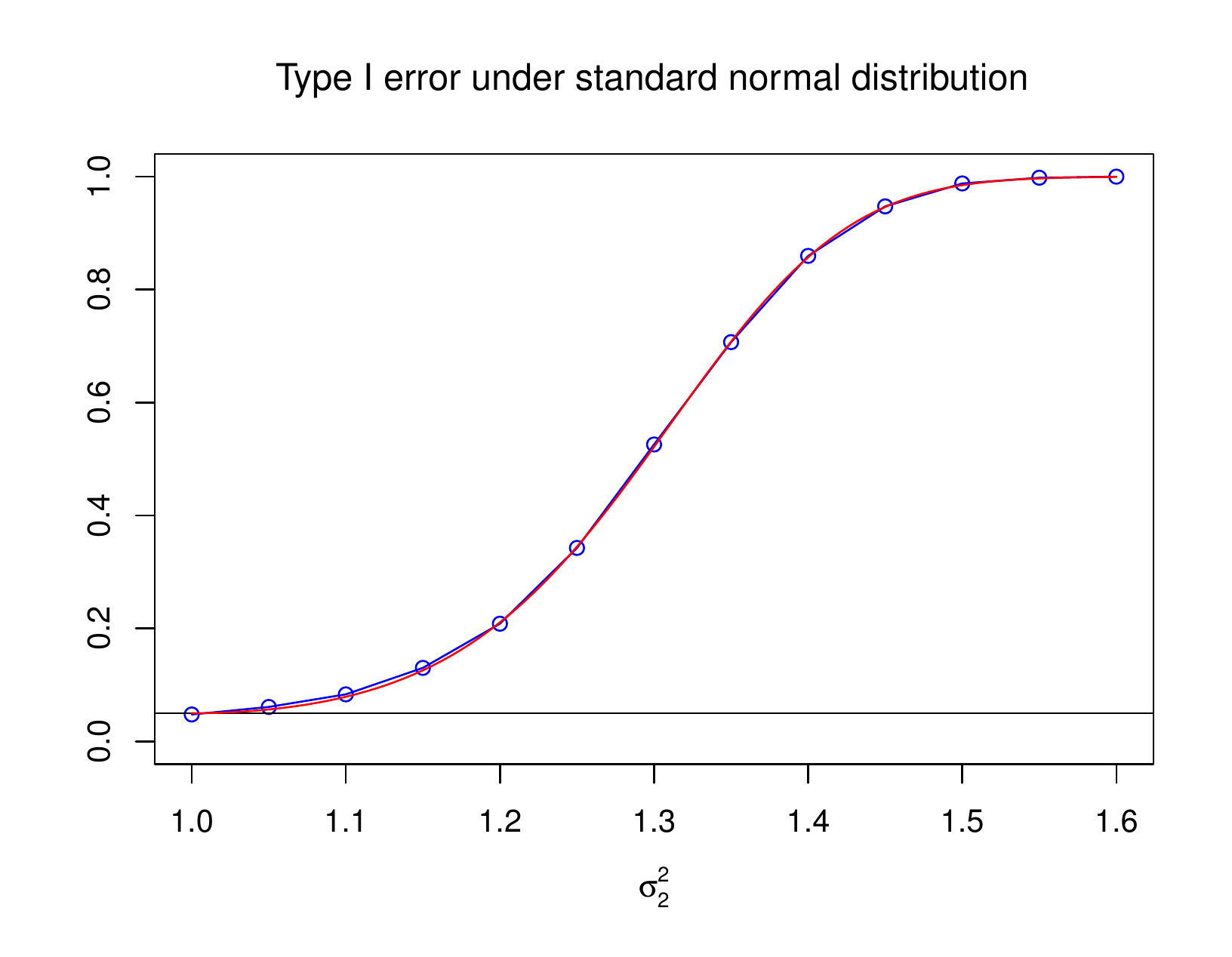} 
\end{minipage}%
\begin{minipage}[t]{0.5\linewidth}
\includegraphics[width=2.3in]{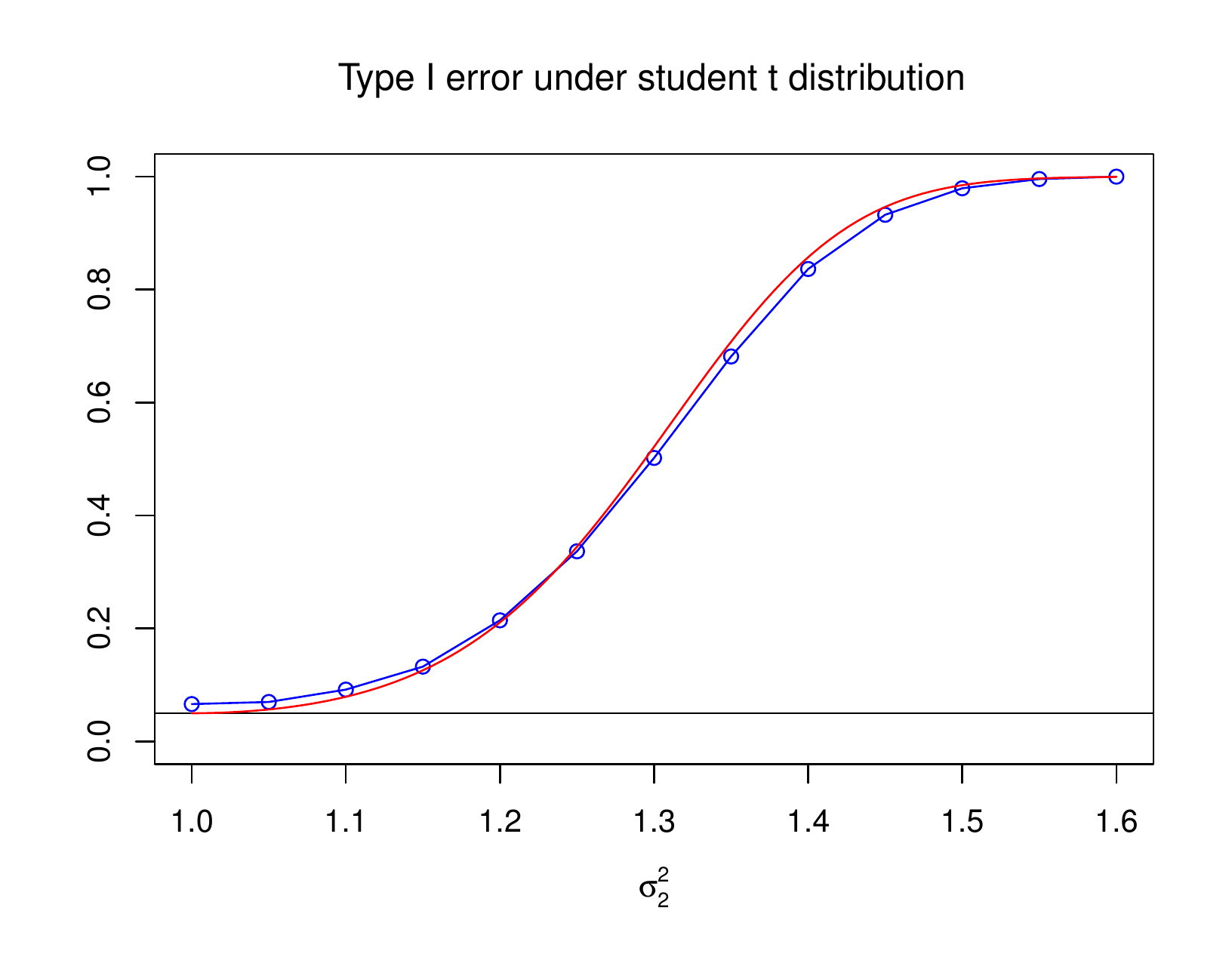}
\end{minipage}
\caption[]{Empirical sizes $T_J^c$ under populations of normal mixture (left panel)  and student-t mixture (right panel), respectively, where the blue lines marked with circles are simulated probabilities and the red lines are theoretical ones.}
  \label{type-1-error}
\end{figure}

Next we compare the performance of $T_n$ with  $T_J^*$ in terms of
both the  empirical size and power.
To study their empirical sizes, we employ three models of the MD under $H_0$,
$$
 G_1= 0.5\delta_1+0.5\delta_2,\quad G_2= 0.3\delta_1+0.4\delta_2+0.3\delta_3,\quad G_3= 0.2\delta_1+0.3\delta_2+0.3\delta_3+0.2\delta_4.
$$
The dimensional ratios are $p/n=1/2, 1,$ $2$, and the sample sizes are $n= 100,200,400$.
Results collected in Table \ref{table:s12} show that the empirical size of $T_J^*$ is around the nominal level $\alpha=0.05$ under normal mixture but contains a bias under student $t$ mixture when the dimensions are small. The bias decreases as the dimensions increase. {\ggai  In contrast, the size of $T_n$ is more favorable under both normal and student-t mixtures. As it has only slightly downward bias when the dimensions are small.}

{\ggai 
\begin{table}[tbp]
\setlength\tabcolsep{6pt}
\begin{center}
\caption{Empirical sizes of $T_J^*$ and $T_n$ (in percent) under
  normal mixture with the three MDs $ G_1$, $ G_2$,  and
  $ G_3$. The nominal significant level is $\alpha=0.05$.
  Upper block: normal mixtures. Lower block: Student-$t$ mixtures.
\label{table:s12}
}
\medskip
\begin{tabular}{rccccccccccccccccccc}
\hline
            & \multicolumn{3}{c}{$p/n=1/2$} &&\multicolumn{3}{c}{$p/n=1$} &&\multicolumn{3}{c}{$p/n=2$}\\
        $n$ & $ G_1$ &$ G_2$&$ G_3$&& $ G_1$ &$ G_2$&$ G_3$&& $ G_1$ &$ G_2$&$ G_3$\\
            \cline{2-4} \cline{6-8}\cline{10-12}
$T_J^*$ \qquad 
       $100$ &4.85&4.88&5.47&&4.91&5.24&5.15&&4.27&4.76&4.92\\
       $200$ &5.00&5.05&5.27&&4.71&5.13&5.33&&4.40&5.00&4.78\\
       $400$ &5.05&5.04&5.08&&4.72&4.92&5.12&&3.93&4.85&4.92\\[2mm]
 $T_n$ \qquad       
       $100$ &3.83&3.71&3.79&&4.31&4.61&4.28&&4.77&4.77&4.45\\
       $200$ &4.59&4.64&4.52&&4.31&4.95&4.61&&4.94&4.89&4.56\\
       $400$ &4.78&4.74&4.99&&5.38&4.70&5.01&&4.99&4.92&4.69\\[1mm]
\hline
$T_J^*$ \qquad 
       $100$ &8.41&8.33&8.94&&7.61&7.11&7.39&&6.75&6.16&6.54\\
       $200$ &7.15&7.00&7.62&&6.91&6.70&6.16&&5.67&5.82&5.27\\
       $400$ &6.46&6.53&6.19&&5.34&5.71&5.67&&5.26&5.32&5.35\\[2mm]
$T_n$ \qquad  
       $100$ &4.48&4.21&4.37&&4.83&4.67&4.46&&4.89&4.60&5.06\\
       $200$ &4.44&4.53&4.83&&4.80&5.05&4.62&&4.95&4.85&5.14\\
       $400$ &4.77&4.72&4.99&&4.44&5.01&5.35&&5.05&4.69&4.83\\
\hline
\end{tabular}
\end{center}
\end{table}
}

To compare the powers of the two tests, {\ggai we design a diagonal shape matrix $\bbT_p^2$ with its PSD
being $H=0.5\delta_{\sigma_1^2}+0.5\delta_{\sigma_2^2}$ where
$\sigma_1^2=1-x$ and $\sigma_2^2=1+x$ with $x$ ranging from $[0,0.3]$
by steps of 0.03. For this model, the factor $\delta_n$ in \eqref{deltan} is $(1-1/n)x^2\in [0,0.08955]$.}
The MD model is simply taken as $ G_3$. The dimensions are
$(p,n)=(400,200)$ and 10,000 independent replications are used (as previously).
{\ggai Figure \ref{power} illustrates that the empirical powers of $T_n$ and $T_J^*$ are both grow to 1 as the parameter $x$ gets
away from zero. Moreover, the power of $T_n$ dominates that of
$T_J^*$. This can be partially explained by the fact that $T_n$
efficiently reduces the effect caused by the fluctuation of $G_n$ around the MD $G$.}

\begin{figure}[h]
\begin{minipage}[t]{0.5\linewidth}
\includegraphics[width=2.3in]{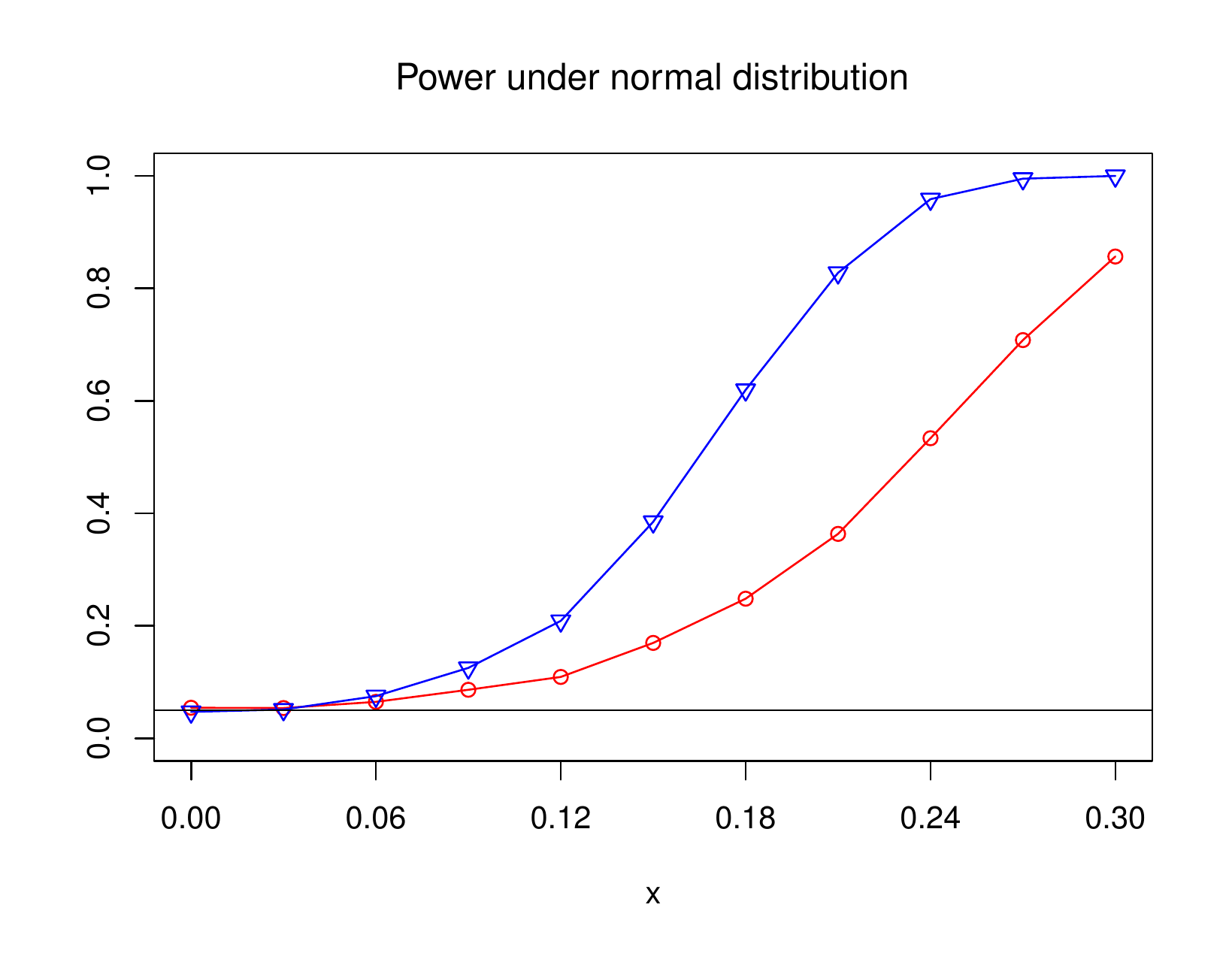} 
\end{minipage}%
\begin{minipage}[t]{0.5\linewidth}
\includegraphics[width=2.3in]{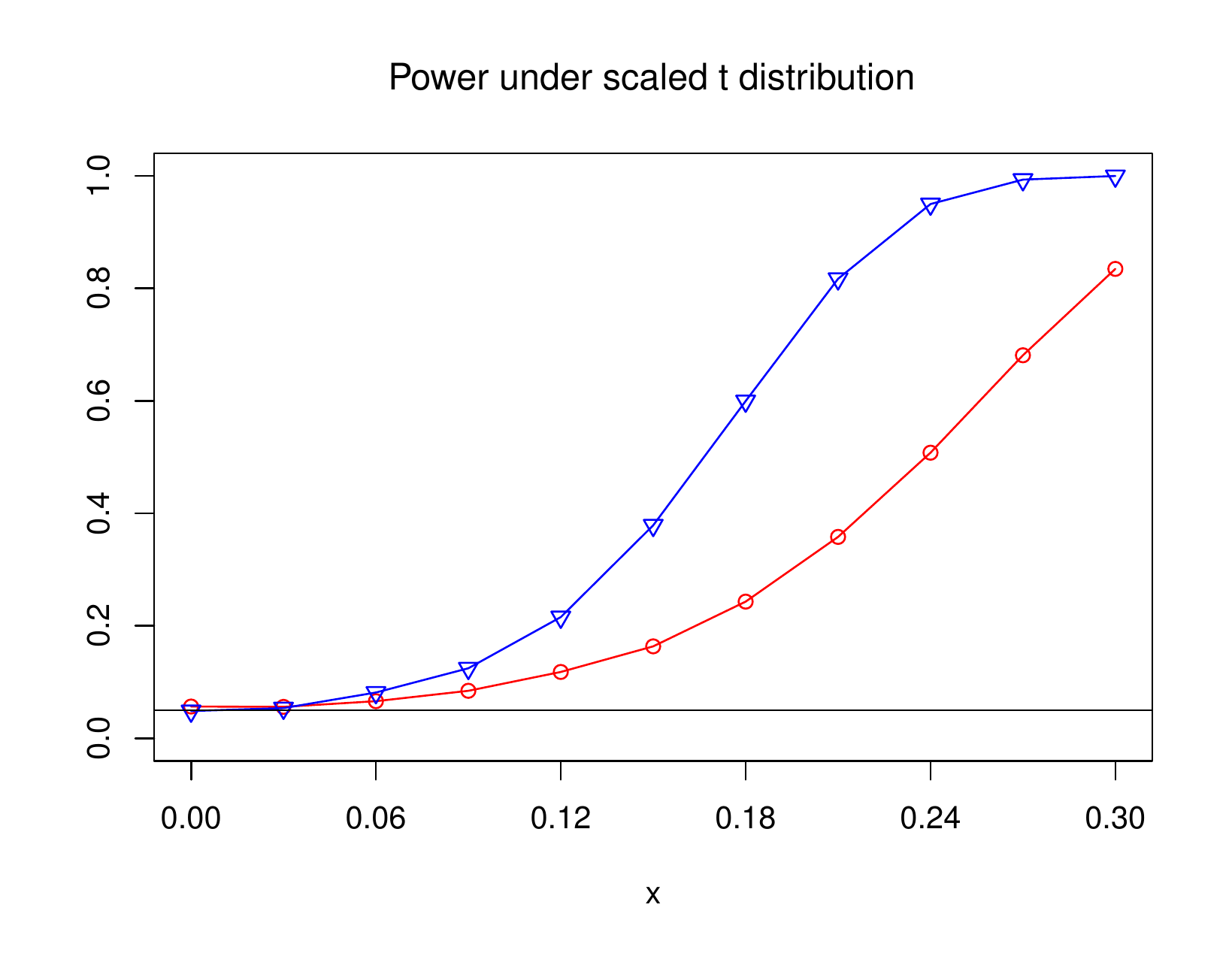}
\end{minipage}
\caption[]{Empirical powers of $T_J^*$ (red line marked with circles) and $T_n$ (blue line marked with triangles) under populations of normal mixture (left panel)  and student-$t$ mixture (right panel), respectively.}
  \label{power}
\end{figure}

\ggai 
\section{Application in model-based clustering}
\label{sec:clust}

Gaussian mixture with finite components are widely applied in model-based cluster analysis.
Considering the mixture model in \eqref{eq:fmn} with $K$ components, the $k$th component covariance can be decomposed as 
\begin{equation*}
\Sigma_k=\sigma_k^2U_k\Lambda_kU_k',
\end{equation*}
where $U_k$ is the matrix of eigenvectors representing the orientation, $\Lambda_k$ is a diagonal matrix proportional to that of the eigenvalues and representing the shape, and $\sigma_k^2$ is a scalar standing for the volume of the cluster. Based on this decomposition, \cite{Banfield93} classified the covariance structure into 14 types in the light of that whether mixture components share a common shape, volume, and/or orientation, which yields a family of parsimonious mixture models. See also \cite{C95,Bensmail96,B00,Bouveyron07}, and \cite{Fraley07}. In high-dimensional scenarios where $p\geq n$, only 6 parsimonious models are concerned:
\begin{align*}
\text{Spherical types:}&~   \Sigma_k=\sigma^2 I_p~\text{(EII)},~ \Sigma_k=\sigma_k^2I_p~\text{(VII)};\\
\text{Diagonal types:}&~ \Sigma_k=\sigma^2 \Lambda~\text{(EEI)},~ \Sigma_k=\sigma_k^2 \Lambda~\text{(VEI)},~ \Sigma_k=\sigma^2 \Lambda_k~\text{(EVI)},~ \Sigma_k=\sigma_k^2 \Lambda_k~\text{(VVI)},  
\end{align*}
which are labeled by three letters ``E", ``V", and ``I"
\citep{Banfield93,Fraley07}. \gai{The letters ``E", ``V" and ``I''
  here designate various combinations in shape, volume and orientation
  for  the component covariance matrices $\Sigma_k$'s.}

An important task in data clustering with mixture models is to properly identify the covariance structure of the mixture. 
\cite{B00} developed a so-called {\em  Integrated Classification
  Likelihood} (ICL) criterion  to select the type of covariance structure and the number of components. 
Under Gaussian assumption, \cite{Fraley07} proposed a {\em Bayesian Information Criterion} (BIC) to  deal with this problem.
Here we apply the corrected John's test $T_J^c$ and our proposed test
$T_n$ to the structure identification among EII, VII, and other
diagonal types when the dimension $p$ is larger than the sample size
$n$. For comparison, we also included the BIC and ICL criteria in our
experiments \gai{by using two ready-made functions \texttt{\small
    mclustBIC}  and  \texttt{\small mclustICL} from  the
  free R package \texttt{\small  mclust}. } 
Samples of $(z_{ij})$ are drawn from standard normal $N(0,1)$. The dimensions are fixed at $(p,n)=(400,200)$. All statistics are calculated from 10,000 independent replications.

Our first experiment is to recognize the structure between EII and VII by $T_J^c$. We take the MD model $G=0.5\delta_1+0.5\delta_{\sigma_2^2}$ as used in Section 3, where $\sigma_2^2\in [1,2]$. Thus the mixture has (at most) two components with their covariance matrices being $\Sigma_1=I_p$ and $\Sigma_2=\sigma_2^2I_p$, respectively.
Results collected in Table \ref{table:c1} show  that when $\Sigma_k$ is EII ($\sigma_2^2=1$),  the empirical size of $T_J^c$ is around the nominal level $\alpha$. For this case, BIC and ICL choose the true model with probability 1. As $\Sigma_k$ moves away from the EII structure, $T_J^c$ can detect this change with an increasing probability up to 1.
In comparison to $T_J^c$, both BIC and ICL completely fail to identify the VII structure when $\sigma_2^2$ is smaller than 1.6.  

\begin{table}
	\setlength\tabcolsep{6pt}
	\begin{center}
		\caption{Probability of rejecting the EII structure using $T_J^c$, $BIC$, and $ICL$. The nominal significant level for $T_J^c$ is $\alpha=0.1,0.05,0.01,0.005,0.001$.}
		\medskip
		
		\begin{tabular}{cccccccccccccccccccc}
			\hline
			&$\alpha=0.1$&$\alpha=0.05$&$\alpha=0.01$&$\alpha=0.005$&$\alpha=0.001$&BIC&ICL\\
			$\sigma_2^2=1.0$&0.1006& 0.0471& 0.0087& 0.0048& 0.0002&0&0\\
			$\sigma_2^2=1.2$&0.6340& 0.4964& 0.2494& 0.1824& 0.0238&0&0\\
			$\sigma_2^2=1.4$&1     & 0.9999& 0.9987& 0.9975& 0.9956&0&0\\
			$\sigma_2^2=1.6$&1     & 1     & 1     & 1     & 1     &0&0\\
			$\sigma_2^2=1.8$&1     & 1     & 1     & 1     & 1     &0.6340&0.6340\\
			$\sigma_2^2=2.0$&1     & 1     & 1     & 1     & 1     &0.7635&0.7635\\
			\hline
		\end{tabular}
		\label{table:c1}
	\end{center}
\end{table}

In the second experiment, \gai{we aim to distinguish the spherical VII
  structure from the group of non-spherical structures VEI, EVI, and VVI for the component matrices by the proposed test $T_n$.}  We employ a mixture of four components with  mixing proportions $(\alpha_1,\alpha_2,\alpha_3,\alpha_4)=(0.2,0.3,0.3,0.2)$.
The component covariance matrices are
$$
\Sigma_k=k I_p+k*diag(\underbrace{a,\ldots,a}_{[p/100]},\underbrace{0,\ldots,0}_{p-[p/100]}),\quad k=1,2,3,4,
$$
where the parameter $a\in[0,4.5]$.
When $a>0$, the covariance structure is VEI and there are only 4 entries different from the spherical basis $kI_p$ for the studied dimension $p=400$. 
Results are exhibited in Table \ref{table:c2}.
It shows that when $\Sigma_k=kI_p$ ($a=0$), the empirical size of
$T_n$ can be well controlled and close to $\alpha$.  Also when
  $1\le a\le 3$  BIC and ICL wrongly choose the spherical model with
probabilities near 1. 
This confirms a widely reported behaviour of such
  information-based criteria in high-dimensional clustering, namely as
  the dimension of the models increase very quickly with the data dimension
  (see Table~\ref{tab:mnm}), these criteria heavily drive to
  {\em over-simplistic models} such as a spherical structure.
As $\Sigma_k$ drifts away from $kI_p$, the
probability of $T_n$ rejecting the VII structure grows to 1. Compared
with BIC and ICL, 
except one case $(a=1,\alpha=0.001)$ where Type I error is kept
  extremely low, 
$T_n$ has overwhelming superiority in capturing small shifts of the covariance structure.

\begin{table}
\setlength\tabcolsep{6pt}
\begin{center}
	\caption{Probability of rejecting the VII structure using $T_n$, $BIC$, and $ICL$. The nominal significant level for $T_n$ is $\alpha=0.1,0.05,0.01,0.005,0.001$.}
	\medskip
	
	\begin{tabular}{cccccccccccccccccccc}
		\hline
		&$\alpha=0.1$&$\alpha=0.05$&$\alpha=0.01$&$\alpha=0.005$&$\alpha=0.001$&BIC&ICL\\
		$a=0  $&0.0996& 0.0507& 0.0108& 0.0063& 0.0003&0.0006&0.0006\\
		$a=1.0$&0.2408& 0.1457& 0.0431& 0.0254& 0.0012&0.0078&0.0078\\
		$a=2.0$&0.8196& 0.7184& 0.4874& 0.3940& 0.0941&0.0182&0.0182\\
		$a=3.0$&0.9990& 0.9965& 0.9872& 0.9786& 0.8596&0.0257&0.0257\\
		$a=4.0$&1     & 1     & 1     & 1     & 0.9995&0.7872&0.7876\\
		$a=4.5$&1     & 1     & 1     & 1     & 1      &0.9794&0.9794\\
		\hline
	\end{tabular}
	
	\label{table:c2}
\end{center}
\end{table}

\section{An empirical study}
\label{sec:data}

In this section, we analyze a classic microarray data set for colon cancer \citep{A99}. The preprocessed data can be found in the R package ``rda". There are 40 tumor and 22 normal colon tissue samples. The dimension of each observation is $p=2000$. Here we model these data as:
$$
\text{Tumor tissue:}~\x_1=\mu_1+w_1\bbT_p\z,\quad \text{Normal tissue:}~\x_2=\mu_2+w_2\bbT_p\z.
$$
Our first interest is to examine whether the shape matrix $\bbT_p$ is spherical. To this end, 
the unknown mean vector in each group is eliminated by subtracting their sample mean. The centralized data 
are denoted by $\y_1,\ldots,\y_n$, $n=62$, and their covariance matrix is calculated as the unbiased one $B_n^*=\sum_{i=1}^n\y_i\y_i'/(n-2)$. It turns out that the $p$-values of $T_J^c$ and $T_n$ are both smaller than $10^{-100}$. Besides, BIC and ICL also support that $\bbT_p$ is not spherical.

Next we consider principal submatrices of $\bbT_p$ and check whether
some of these  submatrices can be considered  spherical. 
Applying BIC clustering to all diagonal elements of \gai{$B_n^*$}, they are
then grouped into 6 clusters. Based on this information, we get 6
principal  submatrices of $\bbT_p$ and their corresponding sample
fragments. These matrices are denoted by $\bbT_{p_1},\ldots,
\bbT_{p_6}$, and their dimensions are $p_1=308$, $p_2=444$, $p_3=343$,
$p_4=674$, $p_5=203$, and $p_6=28$  ($\sum_{p_i}=p=2000$).
Results on identifying the structure of these matrices are presented
in Table \ref{table:r1}. It shows that  BIC and ICL suggest spherical
structure for $\bbT_{p_1},\ldots, \bbT_{p_5}$, and diagonal structure
for $\bbT_{p_6}$. On the contrary, tests of $T_J^c$ and $T_n$ reject
the spherical structure for all these submatrices with $p$-values near
0.
Given the over-simplistic nature of BIC and ICl discussed in
  Section~\ref{sec:clust}, the submatrices $\bbT_{p_1},\ldots,
  \bbT_{p_5}$ are more likely non-spherical as predicted by the test
  statistics $T_J^c$ and $T_n$.

\begin{table}[tbp]
	\setlength\tabcolsep{6pt}
	\begin{center}
		\caption{Structure identification for 6 principle submatrices.}
		\medskip
		
		\begin{tabular}{cccccccccccccccccccc}
			\hline
      Submatrices   &$\bbT_{p_1}$&$\bbT_{p_2}$&$\bbT_{p_3}$&$\bbT_{p_4}$&$\bbT_{p_5}$&$\bbT_{p_6}$\\
            BIC     &VII   &VII   &VII   &VII   &VII   &VVI\\
            ICL     &VII   &VII   &VII   &VII   &VII   &VVI\\
Standardized $T_J^c$&655.9&881.7&656.9&1030.7&345.8&48.0\\
Standardized $T_n$  &364.5&561.6&404.8&639.6 &208.3&24.7\\
			\hline
		\end{tabular}
		
		\label{table:r1}
	\end{center}
\end{table}

\normalcolor
\section{Proofs}
\label{sec:proofs}

\subsection{Proof of Theorem \ref{lsd}} \label{ssec:lsd}

The sample covariance matrix can be represented as
\begin{equation}\label{bn}
B_n=\frac{1}{n}\sum_{i=1}^n\x_i\x_i'=\frac{1}{n}\sum_{i=1}^nw_i^2\bbT_p\z_i\z_i'\bbT_p:=\frac{1}{n}\bbT_pZ_n\Sigma_GZ_n'\bbT_p,
\end{equation}
where $Z_n=(\z_1,\ldots,\z_n)$ and $\Sigma_G=diag(w_1^2,\ldots, w_n^2)$. From Assumption (c), the spectral distribution  of $\Sigma_G$ is 
\begin{equation}\label{fgn}
F^{\Sigma_G}(t)=\frac{1}{n}\sum_{i=1}^n\delta_{w_i^2}(t)\rightarrow  G(t),
\end{equation}
where the convergence holds almost surely, as $n\rightarrow\infty.$
>From  Theorem 4.1.1 in \cite{Z06} and using the independence between
$\Sigma_G$ and $Z_n$, we obtain the result of the theorem.

{\ggai 
\subsection{Proof of Theorem \ref{clt1}}\label{ssec:clt1}

When $\bbT_p$ is an identity matrix, the covariance matrix $B_n$ in \eqref{bn} reduce to $B_n=Z_n\Sigma_GZ_n'/n$.
Let $\w=(w_i)$ be the sequence of the mixing variables. Given $\w$,
the matrix $\Sigma_G$ becomes non-random and the convergence in \eqref{fgn} still holds. As $\Sigma_G$ is diagonal, when $\w$ is fixed, the assumptions of Theorem 1.4 in \cite{PZ08} holds automatically. Applying this theorem, we get the CLT of the LSS conditioning on $\w$.
As this limiting distribution is free of $\w$, Theorem \ref{clt1} is thus verified unconditionally.
}

\subsection{A key lemma} \label{ssec:lem}

  The lemma below on asymptotic fluctuations of some related Stieltjes
  transforms form the core basis for the proof of Theorem~\ref{clt2}
  In Section~\ref{ssec:Theorem3}.

Let $m_{F^{c_n, G_n}}(z)$ and $m_{F^{c_n, G}}(z)$ be the Stieltjes transforms of the LSDs $F^{c_n, G_n}$ and $F^{c_n,  G}$, respectively. Define the random process
$$
M_n(z)=\sqrt{n}\left[m_{F^{c_n, G_n}}(z)-m_{F^{c_n, G}}(z)\right],\quad z\in\mathcal C, 
$$
where the contour $\mathcal C$ is
\begin{equation}\label{contour}
\mathcal C=\{x\pm {\rm i}v_0: x\in[x_l,x_r]\}\cup\{x\pm {\rm i}v: x\in \{x_l,x_r\}, v\in[0,v_0]\},
\end{equation}
with real numbers $v_0 > 0$, $x_r>b(1+1/\sqrt{c})^2$, and $x_l<aI_{(0,1)}(1/c)(1-1/\sqrt{c})^2$.

\begin{lemma}\label{mn-clt}
Under Assumptions (a)-(d), the random process $M_n(\cdot)$ converges weakly to a two-dimensional mean-zero Gaussian process $M(\cdot)$ on $\mathcal C$, whose covariance function is given by 
\begin{eqnarray*}
{\rm Cov}\left(M(z_1),M(z_2)\right)
&=&m'(z_1)m'(z_2)\bigg(\frac{z_1-z_2}{c(m(z_2)-m(z_1))}+\frac{1}{cm(z_1)m(z_2)}\nonumber\\
&&-\frac{(1+z_1m(z_1))(1+z_2m(z_2))}{m(z_1)m(z_2)}\bigg),
\end{eqnarray*}
where $m(z)$ is the Stieltjes transform of the LSD $F^{c, G}$.
\end{lemma}

The proof of this Lemma is lengthy and technical. It is relegated to
the supplementary file (Appendix D).

\subsection{Proof of Theorem \ref{clt2}}\label{ssec:Theorem3}
For all $n$ large, with probability one,
$$
S_{F^{c_n, G_n}}\cup S_{F^{c_n, G}}\subset \left[aI_{(0,1)}(1/c)(1-\sqrt{1/c})^2,b(1+\sqrt{1/c})^2\right].
$$
Therefore, for any $f\in \{f_1,\ldots, f_k\}$, with probability one,
$$
\int f(x)d\F_{n2}(x)=-\frac{1}{2\pi\rm i}\oint_{\mathcal C} f(z)M_n(z)dz,
$$
for all $n$ large, where the contour $\mathcal C$ is defined in \eqref{contour} and takes the positive direction in the complex plane. 
>From Lemma \ref{mn-clt} and the arguments on Page 563 of \cite{BS04}, the random vector \eqref{lss} converges weakly to 
$$
\left(-\frac{1}{2\pi\rm i}\oint_{\mathcal C} f_1(z)M(z)dz,\ldots, -\frac{1}{2\pi\rm i}\oint_{\mathcal C} f_k(z)M(z)dz\right),
$$ 
which is a zero-mean Gaussian vector whose covariance function is
\begin{eqnarray*}
&&{\rm Cov}\left(-\frac{1}{2\pi\rm i}\oint_{\mathcal C} f(z)M(z)dz, -\frac{1}{2\pi\rm i}\oint_{\mathcal C} g(z)M(z)dz\right)\\
&=&\frac{-1}{4\pi^2}\oint_{\mathcal C_1}\oint_{\mathcal C_2} f(z_1)g(z_2)\Cov(M(z_1), M(z_2))dz_1dz_2\\
&=&\frac{-1}{4\pi^2}\oint_{\mathcal C_1}\oint_{\mathcal C_2} \frac{f(z_1)g(z_2) m'(z_1)m'(z_2)(z_1-z_2)}{c(m(z_2)-m(z_1))}dz_1dz_2\\
&&-\frac{1}{4\pi^2}\oint_{\mathcal C_1}\oint_{\mathcal C_2}\frac{ f(z_1)g(z_2) m'(z_1)m'(z_2)}{cm(z_1)m(z_2)}dz_1dz_2\\
&&+\frac{1}{4\pi^2}\oint_{\mathcal C_1}\oint_{\mathcal C_2}\frac{ f(z_1)g(z_2) m'(z_1)m'(z_2)(1+z_1m(z_1))(1+z_2m(z_2))}{m(z_1)m(z_2)}dz_1dz_2,
\end{eqnarray*}
where $f,g\in\{f_1,\ldots,f_k\}$ and $\mathcal C_1$, $\mathcal C_2$ are two non-overlapping contours having the same properties as $\mathcal C$.

\subsection{Proof of Proposition \ref{clt-moments}}
\label{ssec:clt-moments} 

Without loss of generality, let $\mathcal C$ be a contour as defined in \eqref{contour}, taking positive direction and satisfying 
$$\max_{t\in S_{ G}, z\in \mathcal C}|ctm(z)|<1,$$ 
where $S_{ G}$ is the support of $ G$. This can be easily done by choosing $z$ with large modulus, since $m(z)\rightarrow0$ as $|z|\rightarrow\infty$.

Denote the image of $\mathcal C$ under $m(z)$ by 
$$m(\mathcal C)=\{m(z):z\in \mathcal C\}.$$
Since $m(z)$ is a univalent analytic function on $\mathbb C\setminus (S_F\cup\{0\})$, the contour $\mathcal C$ and its image $m(\mathcal C)$ are homeomorphic, which implies $m(\mathcal C)$ is also a simple and closed contour. In addition, from the open mapping theorem and the fact $m(z)\rightarrow0$ as $|z|\rightarrow\infty$, we conclude that $m(\mathcal C)$ has negative direction and encloses zero. 

Let $P(m)=zm$ where $z=z(m)$ is a function of $m$ defined by the equation \eqref{mp3},
then $P(m)$ has Taylor expansion on $m(\mathcal C)$,
\begin{eqnarray*}
P(m)=-1+\int\frac{tm}{1+ctm}d G(t)=-1-\frac{1}{c}\sum_{k=1}^\infty   \gamma_k(-cm)^k,\label{pf}
\end{eqnarray*}
where $\gamma_k=\int t^kd G(t)$ is the $k$th moment of $ G$. Moreover, the quantities $u_{s,t}$ defined in the theorem is the coefficient of $m^t$ in the Taylor expansion of $P^s(m)$. 

Let $\mathcal C_1$ and $\mathcal C_2$ be two non-overlapping contours having the same properties as $\mathcal C$ defined above. 
>From Theorem \ref{clt}, we need to calculate the following three integrals:
\begin{eqnarray*}
I_1&=&-\frac{1}{4\pi^2}\oint_{\mathcal C_1}\oint_{\mathcal C_2}\frac{z_1^{i}z_2^{j}(z_1-z_2)m'(z_1)m'(z_2)}{c(m(z_2)-m(z_1))}dz_1dz_2,\\
I_2&=&-\frac{1}{4\pi^2}\oint_{\mathcal C_1}\oint_{\mathcal C_2}\frac{ z_1^iz_2^j m'(z_1)m'(z_2)}{cm(z_1)m(z_2)}dz_1dz_2,\\
I_3&=&\frac{1}{4\pi^2}\oint_{\mathcal C_1}\oint_{\mathcal C_2}\frac{ f(z_1)g(z_2) m'(z_1)m'(z_2)(1+z_1m(z_1))(1+z_2m(z_2))}{m(z_1)m(z_2)}dz_1dz_2.\\
\end{eqnarray*}
Notice that 
\begin{eqnarray*}
&&\frac{1}{4\pi^2}\oint_{\mathcal C_1}\oint_{\mathcal C_2}\frac{z_1^{i}z_2^{j}}{c(m(z_1)-m(z_2))}dm(z_1)dm(z_2)\\
&=&\frac{1}{4c\pi^2}\oint_{m(\mathcal C_2)}\oint_{m(\mathcal C_1)}\frac{P^{i}(m_1)P^{j}(m_2)}{m_1^{i}m_2^{j}(m_1-m_2)}dm_1dm_2\\
&=&\frac{1}{4c\pi^2}\oint_{m(\mathcal C_2)}\frac{P^{j}(m_2)}{m_2^{j}}\left(\oint_{m(\mathcal C_1)}\frac{P^{i}(m_1)}{m_1^{i}(m_1-m_2)}dm_1\right)dm_2\\
&=&-\frac{1}{2c\pi\rm i}\oint_{m(\mathcal C_2)}\frac{P^{j}(m_2)}{m_2^{j}} \sum_{l=0}^{i-1}\frac{u_{i,l}}{m_2^{i-l}}dm_2\\
&=&\frac{1}{c}\sum_{l=0}^{i-1}u_{i,l}u_{j,i+j-l-1}.
\end{eqnarray*}
Therefore, 
\begin{eqnarray*}
I_1&=&\frac{1}{c}\sum_{l=0}^iu_{i+1,l}u_{j,i+j-l}-\frac{1}{c}\sum_{l=0}^{i-1}u_{i,l}u_{j+1,i+j-l},\\
I_2&=&-\frac{1}{4c\pi^2}\oint_{m(\mathcal C_1)}\frac{ z_1^i}{m_1}dm_1\oint_{m(\mathcal C_2)}\frac{ z_2^j}{m_2}dm_2=u_{i,i}u_{j,j}/c,\\
I_3&=&\frac{1}{4\pi^2}\oint_{m(\mathcal C_1)}\frac{ P^i(m_1)(1+P(m_1))}{m_1^{i+1}}dm_1\oint_{m(\mathcal C_2)}\frac{ P^j(m_2)(1+P(m_2))}{m_2^{j+1}}dm_2\\
&=&-(u_{i,i}+u_{i+1,i})(u_{j,j}+u_{j+1,j}).
\end{eqnarray*}

\subsection{Proof of Theorem \ref{test-tn}}
{\ggai 
From the fact that $\check\gamma_{n2}\xrightarrow{a.s} \gamma_2$ under $H_0$, the first conclusion of the theorem holds if  $nT_n\xrightarrow{D}N(0,8 \gamma_2^2)$. 
We prove this convergence by the Martingale CLT.
Let $\w=(w_i)$ be the sequence of the mixing variables. We first condition on this sequence and show that the limiting results are independent of the conditioning $\w$, thus establish their validity unconditionally. 

Let $\mathbb F_0=\{\emptyset, \Omega\}$, $\mathbb F_k=\sigma\{\x_1,\ldots,\x_k\}$ the $\sigma$-field generated by $\{\x_1,\ldots,\x_k\}$, and $\E_k(\cdot)$ denote the conditional expectation with respect to $\mathbb F_k$, $k=1,\ldots,n$.
By martingale decomposition, 
\begin{eqnarray*}
	nT_n&=&n\sum_{k=1}^n(\E_k-\E_{k-1})(\hat\beta_{n2}-\check\beta_{n2})\\&=&\frac{2}{np}\sum_{k=2}^n\left[\left(\x_k'S_{k-1}\x_k-w_k^2\tr S_{k-1}\right)-\left(\check\x_k'\check S_{k-1}\check\x_k-w_k^2\tr \check S_{k-1}\right)\right]\\
	&:=&\sum_{k=2}^n (D_{nk}-\check D_{nk}),
\end{eqnarray*}
where $S_{k-1}=\sum_{i=1}^{k-1}(\x_i\x_i'-w_i^2I_p)$ and $\check S_{k-1}=\sum_{i=1}^{k-1}(\check\x_i\check\x_i'-w_i^2I_p)$.
It's clear that  $\{D_{nk}-\check D_{nk}, 1\leq k\leq n\}$ is a sequence of martingale difference with respect to $\{\mathbb F_k, 1\leq k\leq n\}$.  From the martingale CLT, say Theorem 35.12 in \cite{B95}, if
$$
\sum_{k=2}^n\E_{k-1}(D_{nk}-\check D_{nk})^2\xrightarrow{i.p.} \sigma^2 \quad \text{and}\quad \sum_{k=2}^n\E\left(D_{nk}-\check D_{nk}\right)^4\rightarrow0,
$$
then $nT_n$ converges in distribution to a normal variable $N(0, \sigma^2)$. 
Notice that $D_{nk}$ and $\check D_{nk}$ are identically distributed, we verify the above conditions by showing that 
\begin{equation}\label{mar:cons}
\sum_{k=2}^n\E_{k-1}D_{nk}^2\xrightarrow{i.p.}4 \gamma_2^2, \quad\sum_{k=2}^n\E_{k-1}D_{nk}\check D_{nk}\xrightarrow{i.p.}0, \quad \text{and}\quad \sum_{k=2}^n\E D_{nk}^4\rightarrow0,
\end{equation}
and hence $\sigma^2=8 \gamma_2^2$.
We note that the proof of the first two terms are similar, so we present only the details for the first one.

>From the expression of $D_{nk}$, we have
\begin{eqnarray*}
	\sum_{k=2}^n\E_{k-1}D_{nk}^2
	&=&\frac{4}{n^2p^2}\sum_{k=2}^n\E_{k-1}\left(\x_k'S_{k-1}\x_k-w_k^2\tr S_{k-1}\right)^2\\
	&=&\frac{4}{n^2p^2}\sum_{k=2}^nw_k^4\left(2\tr S_{k-1}^2+\Delta\tr\left(S_{k-1}\circ S_{k-1}\right)\label{mn1}\right)\\
	&=&\frac{4(2+\Delta)}{n^2p^2}\sum_{k=2}^nw_k^4\sum_{u=1}^p\left[\sum_{i=1}^{k-1}(x_{iu}^2-w_i^2)\right]^2+\frac{8}{n^2p^2}\sum_{k=2}^nw_k^4\sum_{u\neq v}^p\left[\sum_{i=1}^{k-1}x_{iu}x_{iv}\right]^2\nonumber\\
	&:=&M_{n1}+M_{n2},\nonumber
\end{eqnarray*}
where $\circ$ denotes the Hadamard product.
Elementary calculations show that
\begin{eqnarray*}
	\E M_{n1}&=&\frac{4(2+\Delta)}{n^2p^2}\sum_{k=2}^nw_k^4\sum_{u=1}^p\sum_{i=1}^{k-1}\E(x_{iu}^2-w_i^2)^2\to0,\\
	\E M_{n2}&=&\frac{8}{n^2p^2}\sum_{k=2}^nw_k^4\sum_{u\neq v}^p\sum_{i=1}^{k-1}\E x_{iu}^2x_{iv}^2=\frac{4p(p-1)}{n^2p^2}\left[\left(\sum_{k=1}^nw_k^4\right)^2-\sum_{k=1}^nw_k^8\right]\to 4 \gamma_2^2.
\end{eqnarray*}

We next deal with  the variances of $M_{n1}$ and $M_{n2}$. Notice that 
\begin{eqnarray*}
	M_{n1}
	&=&\frac{4(2+\Delta)}{n^2p^2}\sum_{k=2}^nw_k^4\sum_{u=1}^p\left(\sum_{i=1}^{k-1}(x_{iu}^2-w_i^2)^2+2\sum_{i< j}^{k-1}(x_{iu}^2-w_i^2)( x_{ju}^2-w_j^2)\right)\\
	&=&\frac{4(2+\Delta)}{n^2p^2}\left(\sum_{i=1}^{n-1}\sum_{k=i+1}^nw_k^4\sum_{u=1}^p(x_{iu}^2-w_i^2)^2+2\sum_{i< j}^{n-1}\sum_{k=j+1}^nw_k^4\sum_{u=1}^p(x_{iu}^2-w_i^2)( x_{ju}^2-w_j^2)\right),
\end{eqnarray*}
we have
\begin{eqnarray*} 
	\Var(M_{n1})
	&\leq&\frac{32(2+\Delta)^2}{n^4p^4}\left[\sum_{i=1}^{n-1}\left(\sum_{k=i+1}^nw_k^4\right)^2\sum_{u=1}^p\Var(x_{iu}^2-w_i^2)^2\right.\\
	&&\left.+4\sum_{i< j}^{n-1}\left(\sum_{k=j+1}^nw_k^4\right)^2\sum_{u=1}^p\E(x_{iu}^2-w_i^2)^2(x_{ju}^2-w_j^2)^2\right]\\
	&=&O(n^{-4}).
\end{eqnarray*}
Similar discussions on $M_{n2}$ reveal its variance is $O(n^{-2})$. Thus we get $\Var(M_{n1}+M_{n2})\rightarrow0$ and
the first condition in \eqref{mar:cons} is verified.

For the third condition in \eqref{mar:cons}, we have
\begin{eqnarray*}
	\sum_{k=2}^n\E D_{nk}^4&=&\frac{16}{n^4p^4}\sum_{k=2}^n\E\left(\x_k'S_{k-1}\x_k-w_k^2\tr S_{k-1}\right)^4\\
	&\leq &\frac{16K}{n^4p^4}\sum_{k=2}^nw_k^8\E\tr^2 \left(S^2_{k-1}\right)
	=O(n^{-1}),
\end{eqnarray*}
where the inequality is from the fact $\E(\x_k'A\x_k-w_k^2\tr(A))^4\leq w_k^8K\tr^2(A^2)$ with $K$ a constant for any non-random positive definite matrix $A$ and the final order is from elementary calculations.

Next we consider the consistency of the test. 
>From Theorem \ref{lsd}, $\hat\gamma_{n2}\xrightarrow{a.s.}\gamma_2+\gamma_1^2(\tilde\gamma_2-1)/c$. Thus, for all $n$ large,  almost surely, there is a constant $K_1$ such that $\hat\gamma_{n2}<K_1$.  
Under the alternative hypothesis, letting $\w=(w_i)$, 
\begin{eqnarray*}
	\E(T_n|\w)&\leq&\frac{K_2}{n^2p}\sum_{i\neq j}w_i^2w_j^2,\\
	\Var(T_n|\w)
	&\leq &\frac{2}{n^4p^2}\Var\left(\sum_{i\neq j}(\x_i'\x_j)^2\bigg|\w\right)+\frac{2}{n^4p^2}\Var\left(\sum_{i\neq j}(\check \x_i'\check\x_j)^2\bigg|\w\right)=O(n^{-1}),
\end{eqnarray*}
where $K_2$ is a constant and the order of $\Var(T_n|\w)$ is from  Theorem 2.2 in \cite{Srivastavaetal.(2011)}. Then, by $\Var(T_n)=\E(\Var(T_n|\w))+\Var(\E(T_n|\w))$, we get $\Var(T_n)\rightarrow 0$, which is followed by $T_n-\delta_n\xrightarrow{i.p.}0$. 
Let  $z_\alpha$ be the $(1-\alpha)$-quantile of $N(0,1)$, where $\alpha\in(0,1)$.
Finally,
\begin{eqnarray*}
	P(nT_n>\sqrt{8}\hat\gamma_2z_\alpha)
	&\geq& P(n(T_n-\delta_n/2)+n\delta_n/2>\sqrt{8}\hat\gamma_2z_\alpha, T_n>\delta_n/2,\hat\gamma_2<K_1)\\
	&\geq&P(n\delta_n/2>\sqrt{8}K_1z_\alpha, T_n>\delta_n/2) \qquad
    \longrightarrow ~~1~~,
\end{eqnarray*}
which completes the proof.
}

\section*{Acknowledgement} We are grateful to  Steve Marron
for  his suggestion of 
investigating the tricky universe of high-dimensional mixtures, \gai{and to Charles Bouveyron for discussions on the model-based cluster
analysis reported in Section~\ref{sec:clust}}.

\pagebreak\appendix 
\begin{center}
  {\large\bf  On-line supplementary material}
\end{center}
\vskip 1cm 

\section{Estimating a high-dimensional spherical mixture}
\label{sec:est}
\subsection{Estimation of a PMD}

Consider a scale mixture population with a spherical covariance
matrix, that is with $H=\delta_1$. As for the PMD $ G$,
we consider a class of discrete distributions  with finite support on $\mathbb R^+$,
\begin{equation*}
 G({\theta})=\alpha_1\delta_{\sigma^2_1}+\cdots+ \alpha_m\delta_{\sigma^2_m},\quad {\theta}\in{\Theta},
\end{equation*} 
where the order $m$ is assumed known and the parameter space $\Theta$ is
\begin{eqnarray*}
{\Theta} =\bigg\{{\theta}=(\sigma^2_1, \ldots, \sigma^2_m,\alpha_1, \ldots, \alpha_{m}): 0< \sigma^2_1<\cdots <\sigma^2_m<\infty,\ \alpha_i>0,\ \sum_{i=1}^{m}\alpha_i=1\bigg\}.
\end{eqnarray*}

The aim is to find a consistent estimator for the vector parameter $\theta$.
From \cite{BCY10} and \cite{LY14}, the parameter $\theta$ of $ G$ is uniquely determined by its moments $(\gamma_j)$, $\gamma_j=\int t^jdG(t)$, that is, the map from $\theta$ to $ \gamma_0, \gamma_1,\ldots,  \gamma_{2m-1}$,
\begin{equation*}\label{g1}
g_1: \theta\to ( \gamma_0, \gamma_1,\ldots, \gamma_{2m-1})'
\end{equation*}
is a bijection.  Moreover, the recursion  formulae in [15] (here and
below, [$x$] with  brackets  referes to Equation ($x$) of the
main paper)  shows that there is also a one to one map $g_2$ from $ \gamma_j$'s to $ \beta_j$'s,
\begin{equation*}\label{g2}
g_2: ( \gamma_0,\ldots, \gamma_{2m-1})'\rightarrow ( \beta_0,\ldots, \beta_{2m-1})'.
\end{equation*}

From the convergence of $\hat\beta_{nk}$ to $\beta_k$ and the maps $g_1$ and $g_2$ defined above,
we propose a moment estimator $\hat\theta_n$ of $\theta$, which is defined to be 
$$
\hat\theta_n=(g_2\circ g_1)^{-1}(\hat \beta_{n0},\ldots,\hat \beta_{n,2m-1}).
$$
Note that this estimator exists for all $n$ large, and thus we immediately get the following convergence theorem.

\begin{theorem}\label{theta-hat}
In addition to Assumptions (a)-(c), suppose that the true value $\theta_0$ of $\theta$ is an inner point of $\Theta$. Then
$\hat\theta_n\rightarrow \theta_0$ almost surely as $n\rightarrow\infty.$
\end{theorem}

\subsection{Numerical results}

We undertake a simulation study to assess the performance of the proposed estimator $\hat\theta_n$ of a PMD.
Two models are studied:
\begin{itemize}
\item Model 1: $ G=0.8\delta_1+0.2\delta_2$ and $c=1$.
\item Model 2: $ G=0.3\delta_1+0.4\delta_4+0.3\delta_7$ and $c=1$.
\end{itemize}
Samples of $(z_{ij})$ are drawn from $N(0,1)$ and $\sqrt{4/6}\cdot t_6$ for Model 1, and from $N(0,1)$ and $U(-\sqrt{3},\sqrt{3})$ for Model 2. The dimensions are $(p,n)$=(300,300), (500,500), (800,800), and (1200,1200). 
Statistics of the estimators from $10000$ independent replications are
collected in Table \ref{table1-2} for Model 1 and Table
\ref{table3-4} for Model 2. The results show that, in
almost all cases, both the empirical biases and the standard
deviations of all estimators reduce along with a growing dimension, which clearly demonstrates the consistency of the proposed estimator.

\begin{table}
\begin{center}
\caption{Estimates for $(\sigma^2_1,\sigma^2_2,\alpha_1)=(1,2,0.8)$ in
  Model 1 with  $p=n= 300, 500, 800, 1200$. Upper panel: normal
  samples. Lower panel: $\sqrt{4/6}t_6$ samples.}
\medskip
{\small\begin{tabular}{ccccccccccccc}
\hline
        $\theta$        &   \multicolumn{2}{c}{$n=300$} &&  \multicolumn{2}{c}{$n=500$} &&\multicolumn{2}{c}{$n=800$} &&  \multicolumn{2}{c}{$n=1200$} \\
                             &   Mean &St. D. &&   Mean &St. D. &&  Mean &St. D. &&  Mean &St. D. \\
                      \cline{2-3} \cline{5-6}\cline{8-9}\cline{11-12}
                  $\sigma^2_1$        &0.9943&0.0270&&0.9969&0.0159&&0.9982&0.0099&&0.9989&0.0065\\
                  $\sigma^2_2$        &2.0142&0.1063&&2.0067&0.0634&&2.0042&0.0396&&2.0030&0.0262\\
                  $\alpha_1$            &0.7956&0.0461&&0.7978&0.0296&&0.7989&0.0203&&0.7994&0.0151\\       
\hline
                  $\sigma^2_1$        &0.9885&0.0276&&0.9931&0.0167&&0.9955&0.0106&&0.9971&0.0076\\
                  $\sigma^2_2$        &2.0690&0.1821&&2.0396&0.1019&&2.0237&0.0767&&2.0162&0.0492\\
                  $\alpha_1$            &0.8005&0.0438&&0.8005&0.0289&&0.8004&0.0206&&0.8005&0.0156\\       

\hline
\end{tabular}}
\label{table1-2}
\end{center}
\end{table}

\begin{table}
\begin{center}
\caption{Estimates for
  $(\sigma^2_1,\sigma^2_2,\sigma^2_3,\alpha_1,\alpha_2)=(1,4,7,0.3,0.4)$
  in
  Model 2 with  $p=n= 300, 500, 800, 1200$. Upper panel: normal
  samples. Lower panel: $U(-\sqrt{3},\sqrt{3})$ samples.}
{\small\begin{tabular}{ccccccccccccc}
\hline
             $\theta$   &   \multicolumn{2}{c}{$n=300$} &&  \multicolumn{2}{c}{$n=500$} &&\multicolumn{2}{c}{$n=800$} &&  \multicolumn{2}{c}{$n=1200$} \\
                             &   Mean &St. D. &&   Mean &St. D. &&  Mean &St. D. &&  Mean &St. D. \\
                      \cline{2-3} \cline{5-6}\cline{8-9}\cline{11-12}
                  $\sigma^2_1$        &0.9749&0.2152&&1.0042&0.0838&&1.0075&0.0486&&1.0071&0.0320\\
                  $\sigma^2_2$        &4.0467&0.5396&&4.0652&0.3298&&4.0526&0.2097&&4.0408&0.1408\\
                  $\sigma^2_3$        &7.1356&0.3386&&7.0827&0.2001&&7.0503&0.1255&&7.0341&0.0832\\
                  $\alpha_1$            &0.2964&0.0555&&0.3023&0.0330&&0.3030&0.0227&&0.3028&0.0169\\
                  $\alpha_2$           &0.4225&0.0450&&0.4121&0.0314&&0.4065&0.0228&&0.4042&0.0175\\            

\hline
                  $\sigma^2_1$        &0.9473&0.3065&&0.9937&0.0878&&1.0005&0.0504&&1.0027&0.0324\\
                  $\sigma^2_2$        &3.9626&0.5750&&4.0114&0.3419&&4.0152&0.2134&&4.0171&0.1436\\
                  $\sigma^2_3$        &7.0473&0.3288&&7.0252&0.1959&&7.0140&0.1218&&7.0100&0.0822\\
                  $\alpha_1$            &0.2896&0.0607&&0.2987&0.0344&&0.3007&0.0231&&0.3011&0.0172\\
                  $\alpha_2$            &0.4156&0.0445&&0.4061&0.0312&&0.4025&0.0227&&0.4013&0.0175\\            

\hline
\end{tabular}}
\label{table3-4}
\end{center}
\end{table}

\section{Numerical calculations of an LSD}\label{sec:supp}
General forms of the LSD defined in [9]   are quite complex, so we take its simplified version as an example, which is defined in [10]  with the PSD $H=\delta_1$.
Given a model $(c,  G)$, one may find the support $S_F$ of the LSD $F^{c, G}$ with the help of the function $u=u(x)$,
$$
u(x)  =  - \frac1 {x}  +  \int\!\frac{t}{1+ctx} d G(t)~,\quad x\in A,
$$
where $A=\{x\in\mathbb R,x\neq 0, x\neq -1/(ct),\forall t\in S_{H} \}$. This function can be seen as a ``projection" of the equation
[10]  on the real line. Following \cite{SC95}, the support is $S_F=\mathbb R\setminus B$ where the set $B=\{u: du/dx>0, x\in A\}$. In addition, the support should also exclude zero when $c<1$. 

After finding the support, the LSD can be obtained by inversion of the
Stieltjes transform coupled with 
the equation [10].
Here we illustrate two examples:
\begin{itemize}
\item Model 1: $ G=0.4\delta_{0.5}+0.6\delta_{5}$ and $c=2$;
\item Model 2: $ G=0.3\delta_{0.2}+0.4\delta_{0.7}+0.3\delta_{1}$ and $c=10$.
\end{itemize}
In Model 1, the mixture is a combination of two distributions with a proportion of 2:3 and the corresponding covariance matrices are $0.5I_p$ and $5I_p$, respectively. 
It turns out that the support $S_{F}$ is consist of a mass point at zero and two continuous intervals $[0.1450, 1.5618]$ and $[2.3027, 24.1683]$.
In Model 2, the mixture is made up of three distributions with a proportion 3:4:3. The support of $F^{c, G}$ is $S_F=\{0\}\cup[1.2223, 2.5178]\cup [4.2013, 14.5272]$. 

%
\begin{figure}[hp]
\begin{minipage}[t]{0.5\linewidth}
\includegraphics[width=2.3in,height=1.8in]{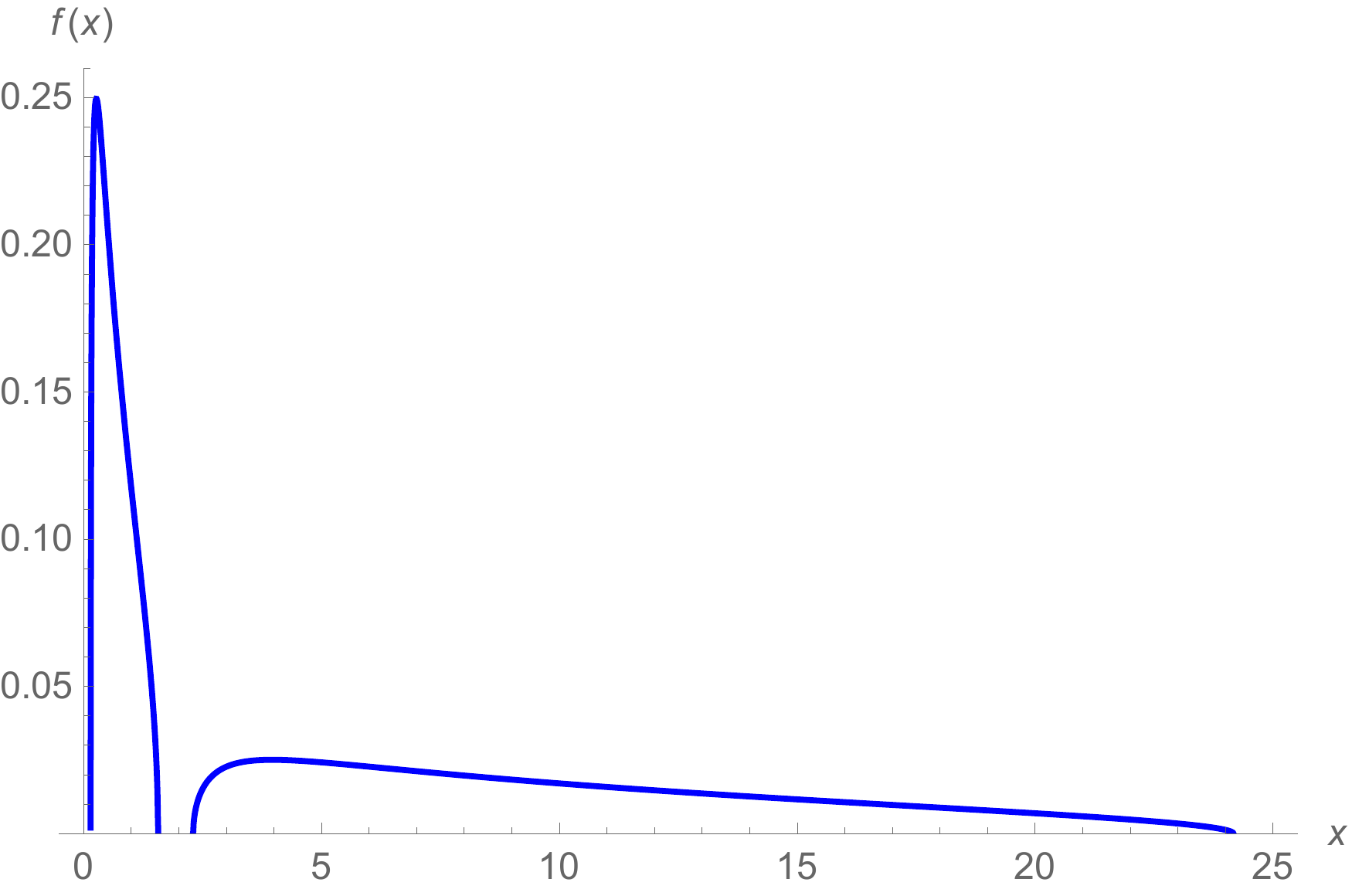} 
\end{minipage}%
\begin{minipage}[t]{0.5\linewidth}
\includegraphics[width=2.3in,height=1.8in]{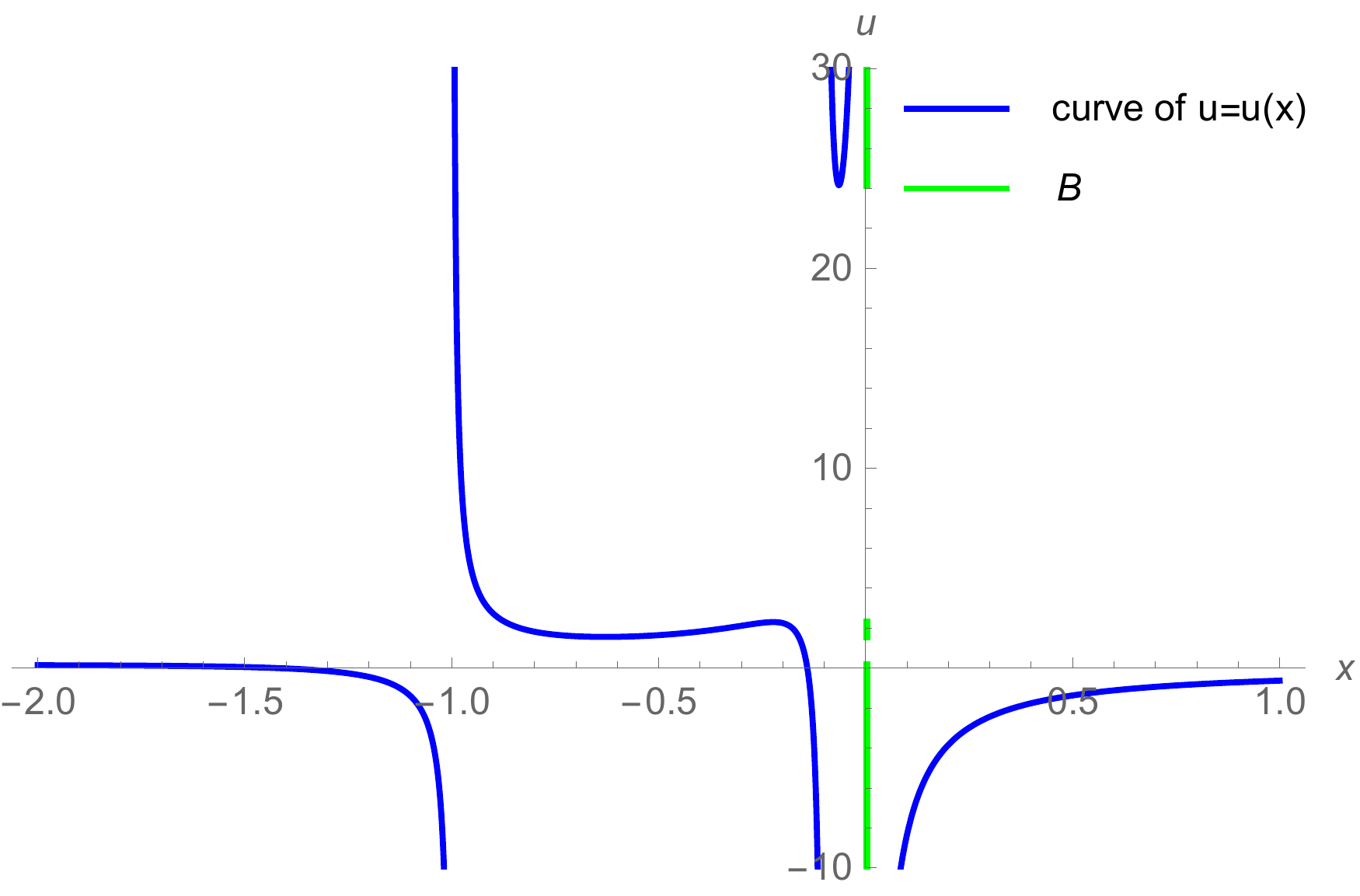}
\end{minipage}
\caption[]{Density curve of the LSD $F^{c, G}$ and the graph of $u=u(x)$ for Model 1.}
\label{fig1}
\end{figure}

\begin{figure}[hp]
\begin{minipage}[t]{0.5\linewidth}
\includegraphics[width=2.3in,height=1.8in]{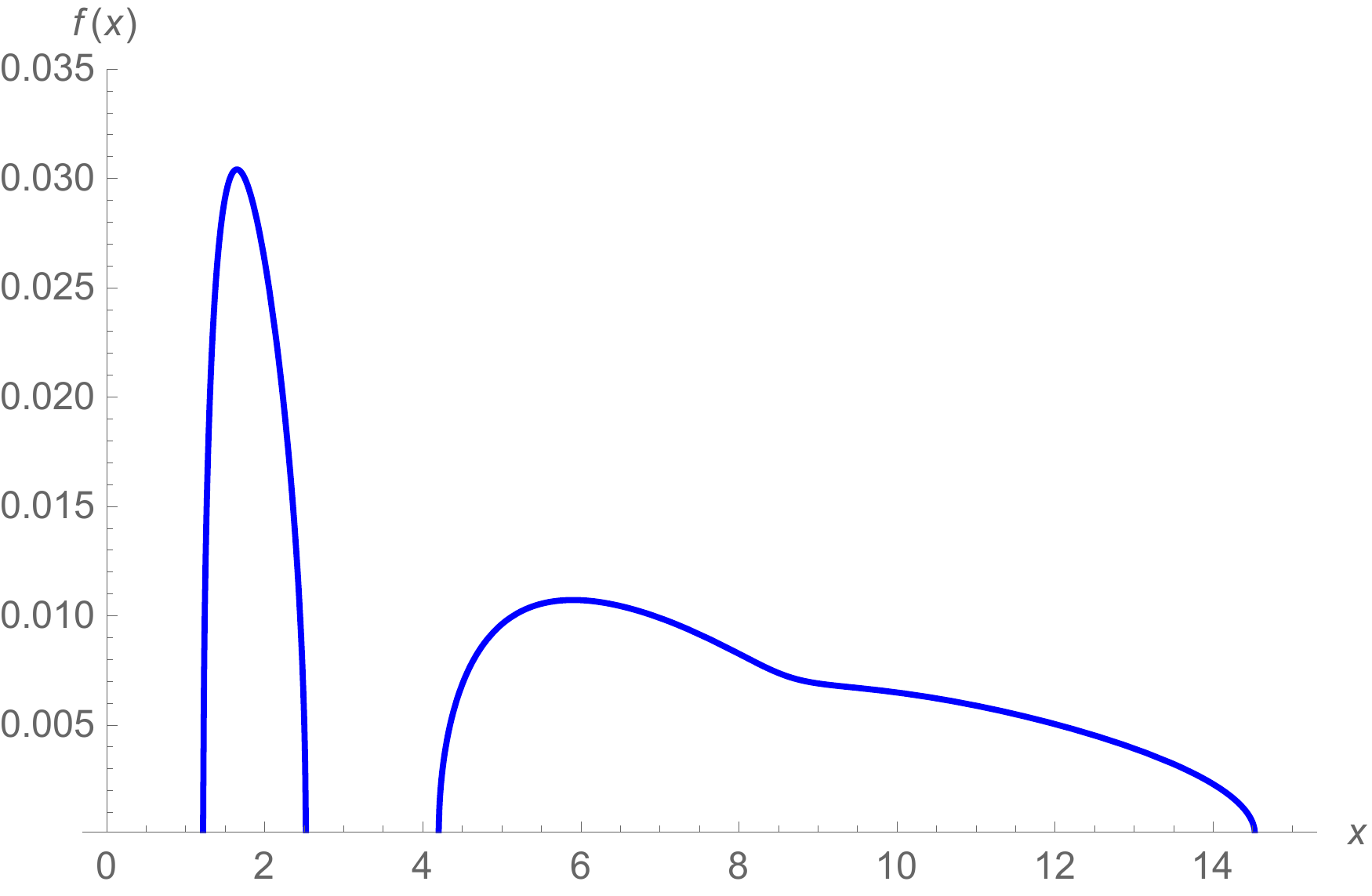}
\end{minipage}%
\begin{minipage}[t]{0.5\linewidth}
\includegraphics[width=2.3in,height=1.8in]{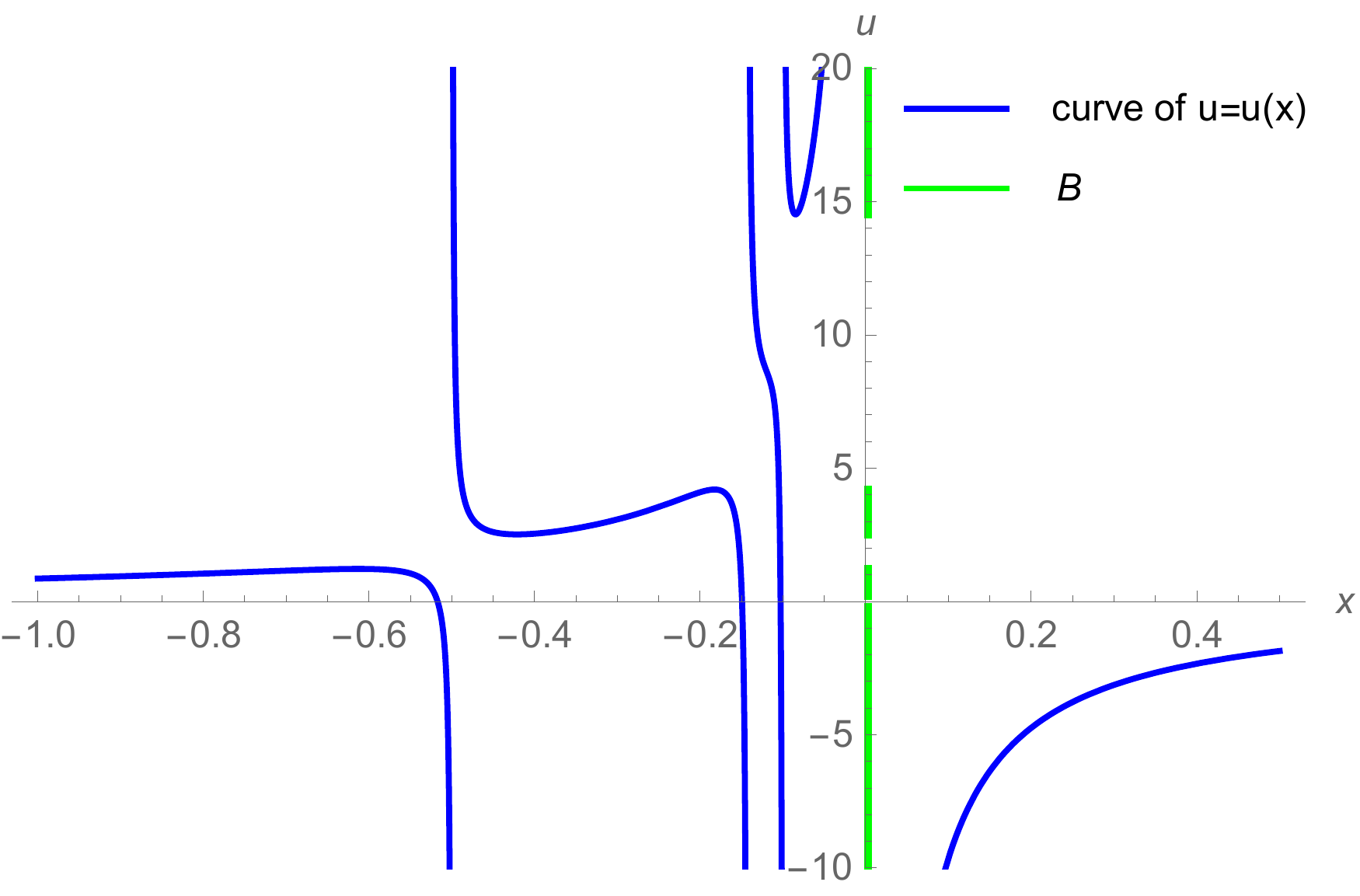}
\end{minipage}
\caption[]{Density curve of the LSD $F^{c, G}$ and the graph of $u=u(x)$ for Model 2.}
\label{fig2}
\end{figure}

One may see from Figures \ref{fig1} and \ref{fig2} that the support $S_F$ of $F^{c, G}$ is a combination of several disjoint intervals. This phenomenon of separation is not new and has been observed in traditional generalized MP distributions \citep{SC95}. It is reported that, for a discrete PMD concentrated in $m$ mass points, the number of disjoint intervals contained in the support of the corresponding LSD grows to $m$ as the dimensional ratio $c$ becomes small. However, the conclusion for the mixture model is just opposite, that is, the number of the disjoint intervals is equal to that of the components in the mixture if the ratio $c$ is large enough. We explain this by considering a mixture model of two component.
 
Let $ G$ be a discrete PMD of order 2, i.e.,
\begin{equation}\label{mix-2}
 G=\alpha_1\delta_{\sigma_1^2}+\alpha_2\delta_{\sigma_2^2},
\end{equation}
where $\alpha_1+\alpha_2=1$, $0<\alpha_1<1$, and $\sigma_1^2\neq\sigma_2^2$. In this case, there are one or two continuous intervals in the support $S_F$ depending on the value of the dimensional ratio $c$ for any fixed $ G$.

\begin{proposition}\label{supp-2}
Suppose that Assumptions (a)-(c) hold. 
For the mixture model \eqref{mix-2},  the support $S_F^*:=S_F\setminus\{0\}$ has the form
 \begin{eqnarray}
 S_F^*=
 \begin{cases}
 [s_1,s_2]\cup [s_3, s_4]&c> c_0,\\
 [s_1,s_4]&c\leq c_0,
 \end{cases}
 \end{eqnarray}
 where $c_0=\int (tx^*)^2/(1+tx^*)^2d G(t)$
with $ x^*$ the only real root of $\int t^2/(1+tx)^3d G(t)=0$ and $0\leq s_1<s_2< s_3<s_4$ are real numbers given in the proof.
\end{proposition}

\begin{proof}

For the PMD in \eqref{mix-2} and $c\neq 1$, the equation $u'(x)=0$ is quartic and thus has two or four real roots which correspond to the boundary points of $S_F^*$.  
Let
\begin{equation}\label{dz-dm}
f(x)=c-\alpha_1\left(\frac{c\sigma_1^2x}{1+c\sigma_1^2x}\right)^2-\alpha_2\left(\frac{c\sigma_2^2x}{1+c\sigma_2^2x}\right)^2,
\end{equation}
then $f(x)=0$ shares the same roots with $u'(x)=0$. Notice that the equation $f'(x)=0$ can be reduced to 
$$
\frac{\alpha_1\sigma_1^4}{(1+c\sigma_1^2x)^3}+\frac{\alpha_2\sigma_2^4}{(1+c\sigma_2^2x)^3}=0,
$$
which is a cubic equation and has only one real root $x_0=x^*/c$. Note that this root is a minimum point of $f(x)$. Therefore, the function $f(x)$ has four real zeros $x_1<x_2< x_3<x_4$ if $f(x_0)>0$, three zeros $x_1<x_0<x_4$ if $f(x_0)=0$, and two zeros $x_1<x_4$ if $f(x_0)<0$. From \cite{SC95} and the fact $f(x_0)=c-c_0$, we get
 \begin{eqnarray*}
 S_F^*=
 \begin{cases}
 [u(x_1),u(x_2)]\cup [u(x_3), u(x_4)]&c> c_0,\\
 [u(x_1),u(x_4)]&c\leq c_0.
 \end{cases}
 \end{eqnarray*}
 For the case $c=1$, $f(x)=0$ is a cubic function and thus has one or three real roots, denoted by $x_2$ and $x_2\leq x_3<x_4$ respectively. Following similar arguments,
 the support $S_F^*$ is
 \begin{eqnarray*}
 S_F^*=
 \begin{cases}
 [0,u(x_2)]\cup [u(x_3), u(x_4)]&c> c_0,\\
 [0,u(x_4)]&c\leq c_0.
 \end{cases}
 \end{eqnarray*}
\end{proof}

Figure \ref{fig3} shows the evolution of the support $S_F^*$ with respect to the ratio $c$ under four models, where their parameters are $(\alpha_1,\alpha_2)=(0.9,0.1)$, $(0.99,0.01)$ and $(\sigma_1^2,\sigma_2^2)=(1,5)$, $(1,10)$, respectively. The  shadowed area exhibits the support $S_F^*$, from which we see that the support is a single interval (blue color) when $c\leq c_0$ and is a union of two separate intervals (red color) when $c>c_0$.   
When $c$ tends to zero, the support shrinkages to the point $E(w^2)=\int td G(t)$. Notice that in the classical low dimensional setting where $p$ is fixed while $n$ grows to infinity, the sample covariance matrix converges almost surely to it population counterpart $\E(w^2)\cdot I_p$ so that all eigenvalues converge to $\E(w^2)$. Therefore, the above high-dimensional case with $c$ small just mimics this low-dimensional setting.
In addition, comparing the four models, the critical value $c_0$ for the separation of $S_F^*$ becomes small when $\alpha_1$ and/or $\sigma_2^2$ increase. At last, corresponding to these supports, we present also some density curves of the LSD in Figure \ref{fig4}. 

\begin{figure}[hp]
\centering
\begin{minipage}[t]{0.5\linewidth}
\includegraphics[width=2.3in,height=1.6in]{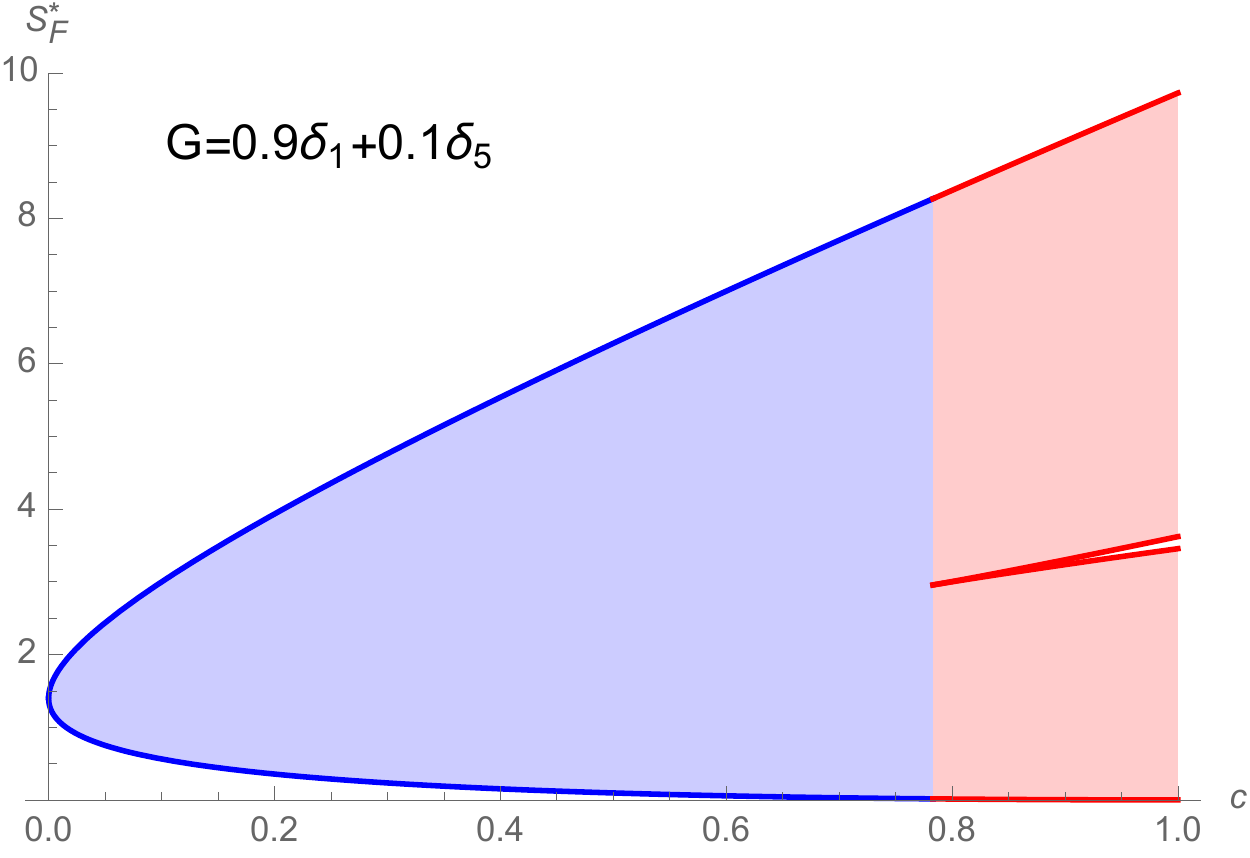} 
\end{minipage}%
\begin{minipage}[t]{0.5\linewidth}
\includegraphics[width=2.3in,height=1.6in]{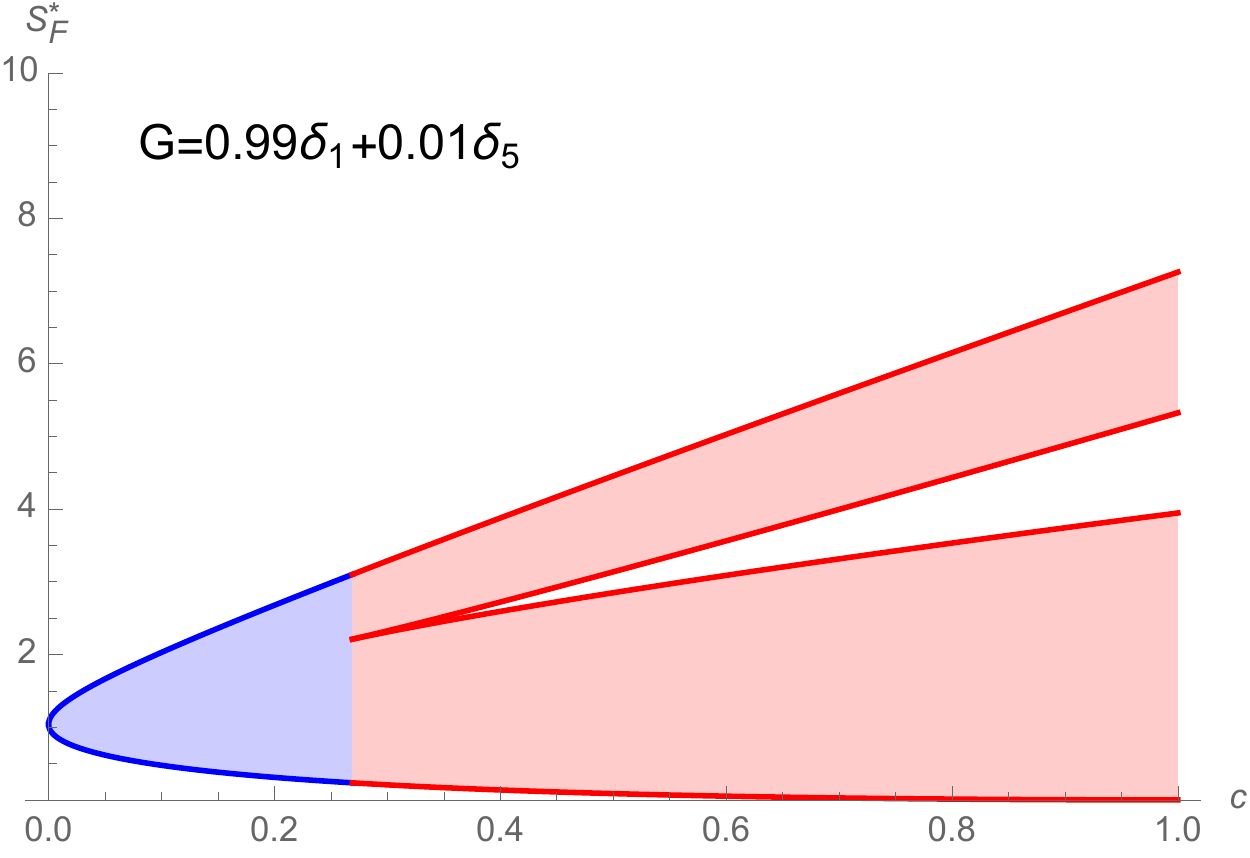}
\end{minipage}
\begin{minipage}[t]{0.5\linewidth}
\includegraphics[width=2.3in,height=1.6in]{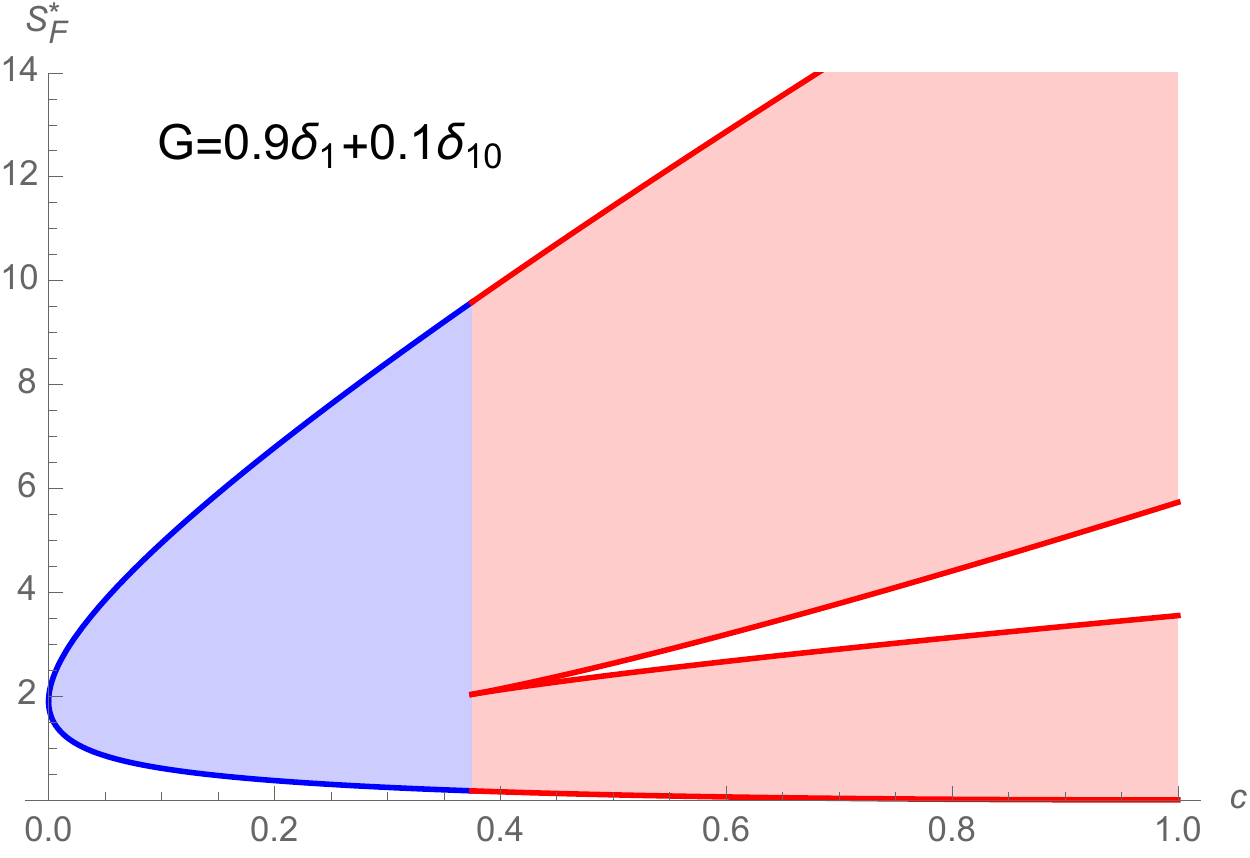} 
\end{minipage}%
\begin{minipage}[t]{0.5\linewidth}
\includegraphics[width=2.3in,height=1.6in]{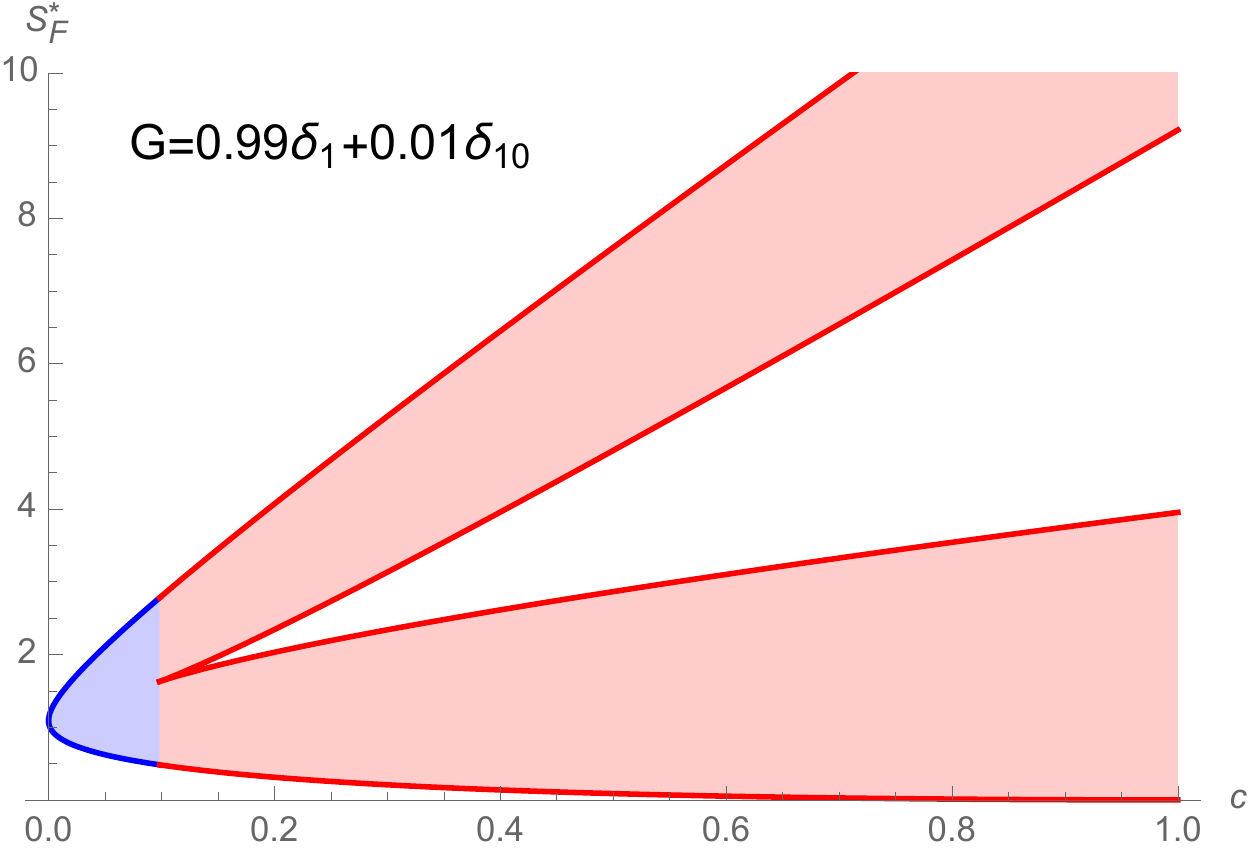}
\end{minipage}
\caption{Evolution of the support $S_F$ as the increase of the dimensional ratio $c$.}
\label{fig3}
\end{figure}


%
\begin{figure}[ph!]
\centering
\begin{minipage}[t]{0.5\linewidth}
\includegraphics[width=2.3in,height=1.6in]{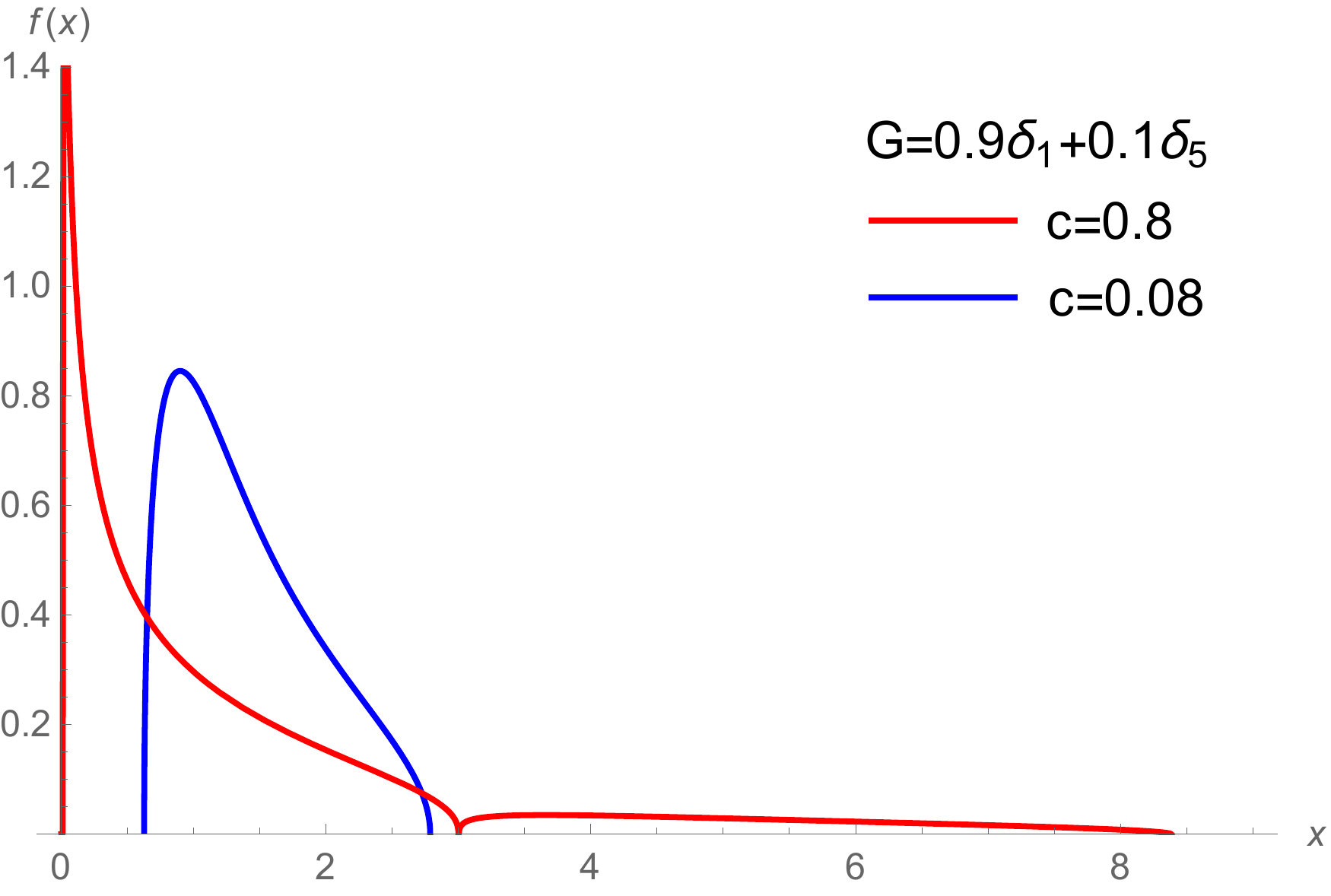} 
\end{minipage}%
\begin{minipage}[t]{0.5\linewidth}
\includegraphics[width=2.3in,height=1.6in]{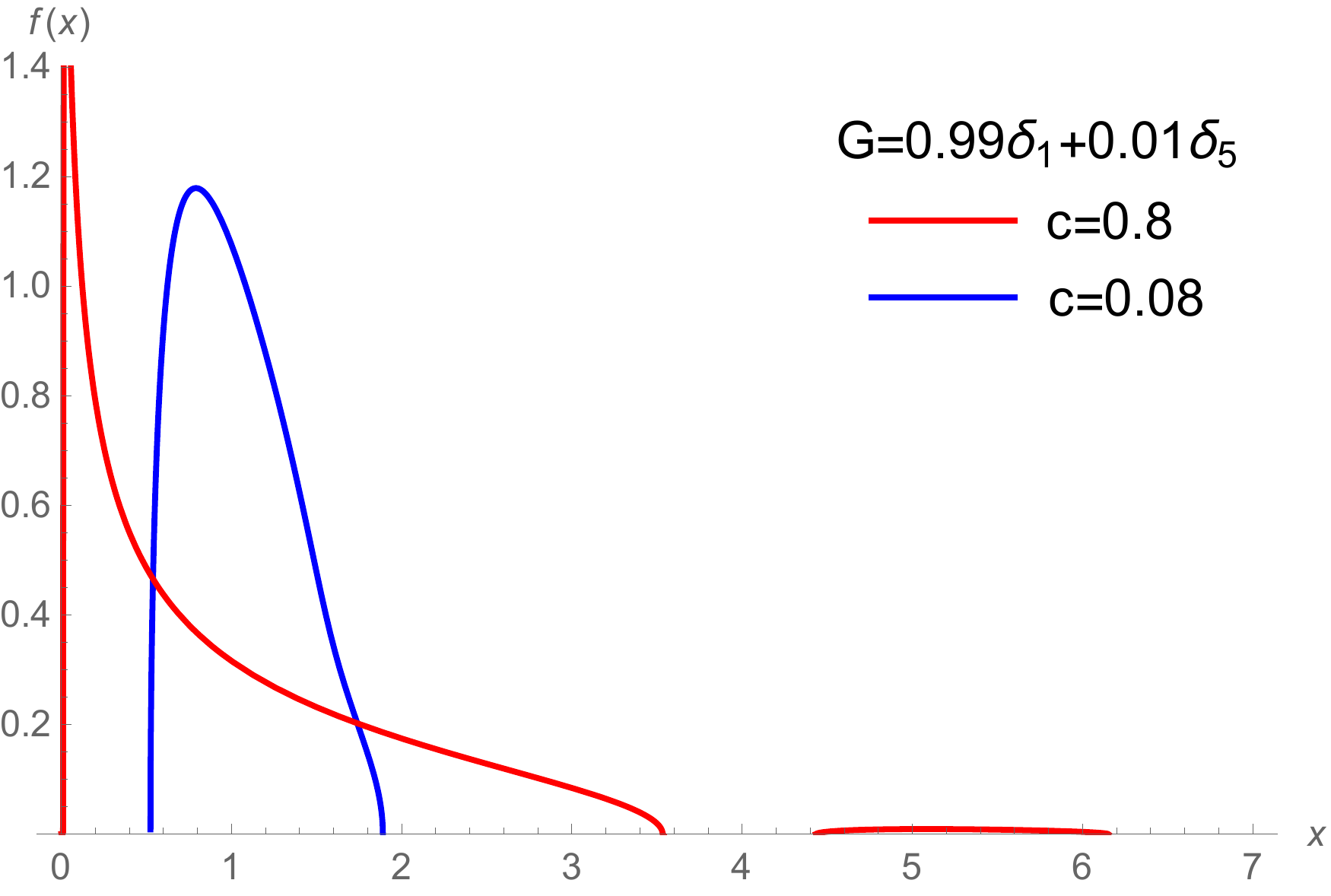}
\end{minipage}
\begin{minipage}[t]{0.5\linewidth}
\includegraphics[width=2.3in,height=1.6in]{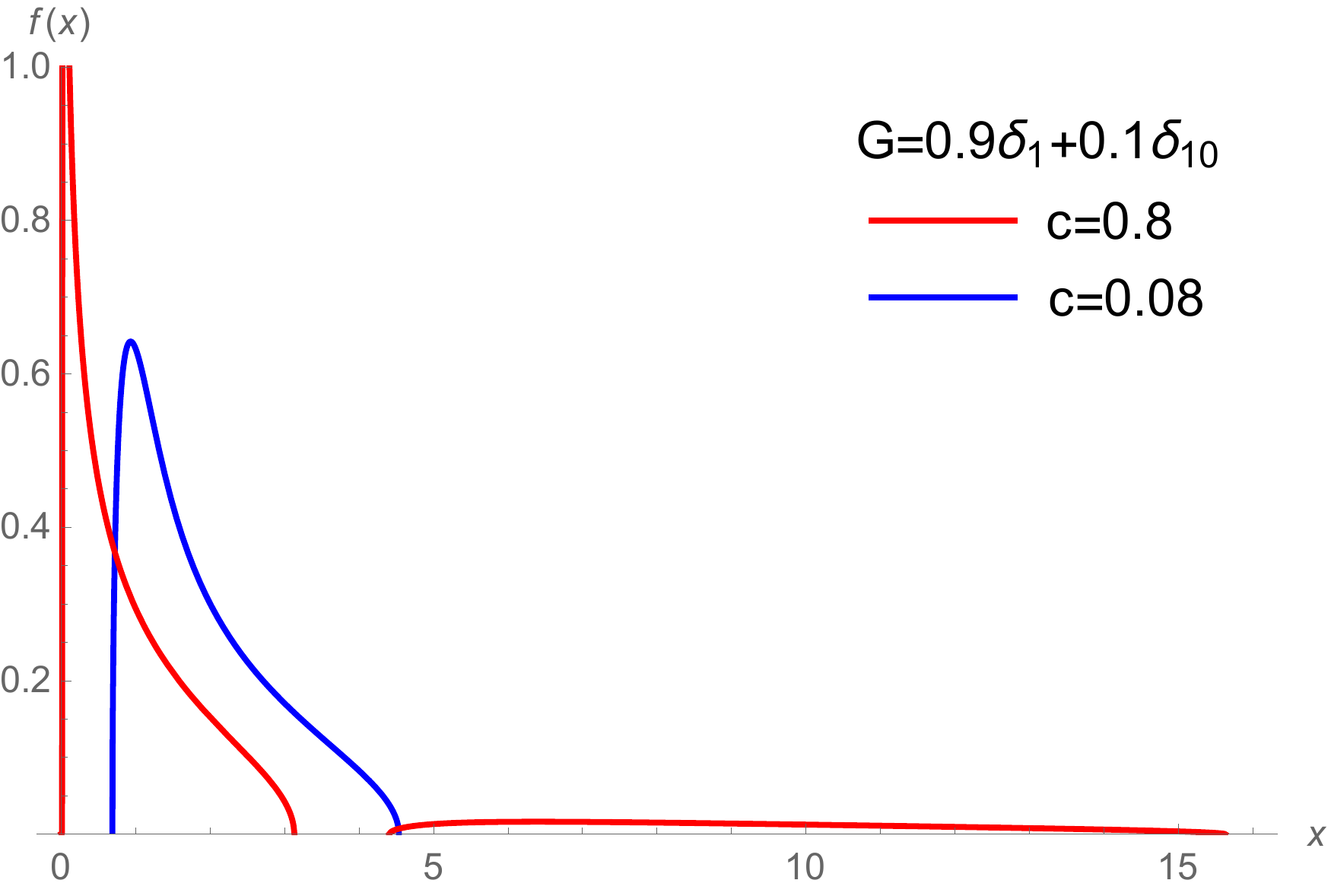} 
\end{minipage}%
\begin{minipage}[t]{0.5\linewidth}
\includegraphics[width=2.3in,height=1.6in]{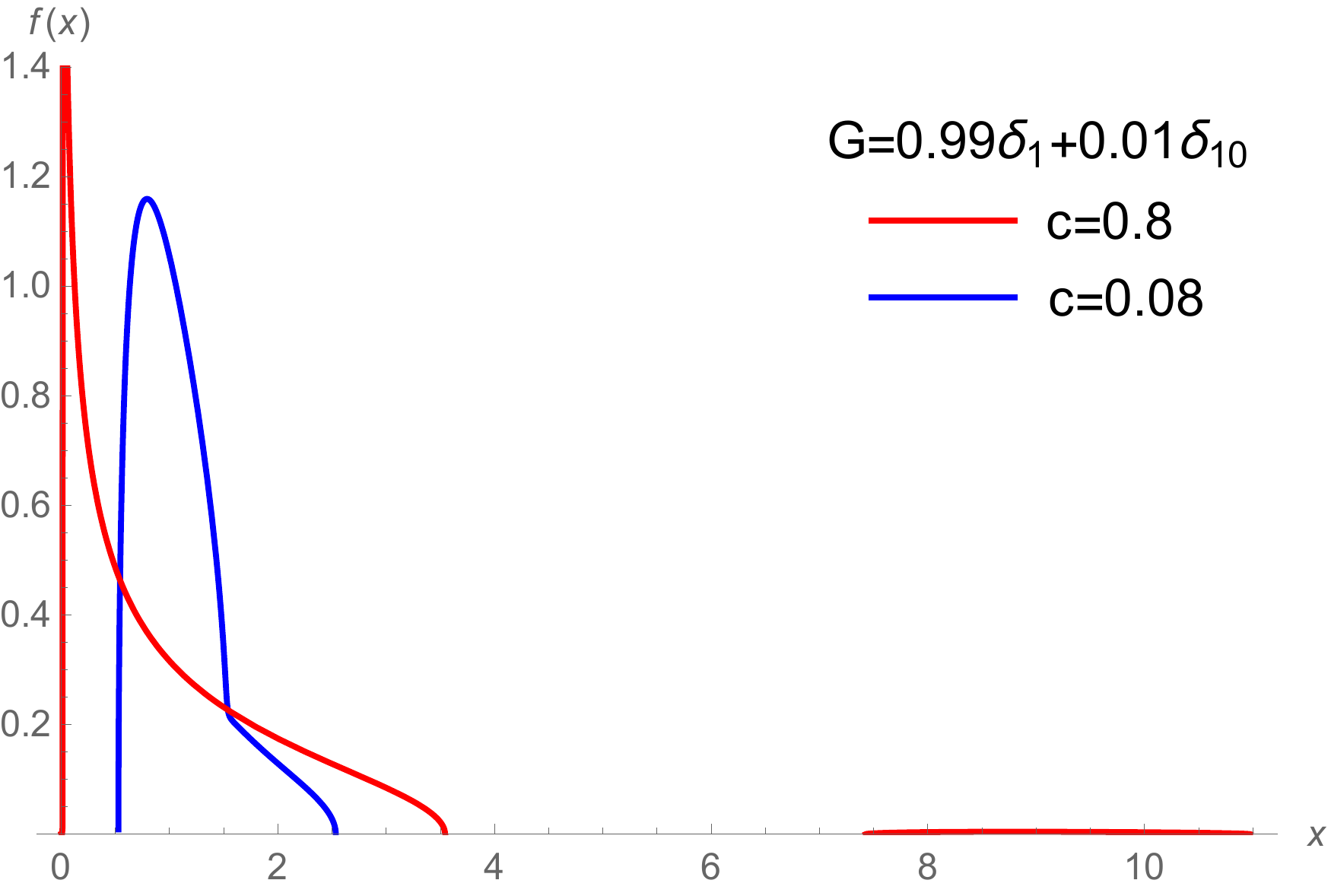}
\end{minipage}
\caption{Density curves of the LSD associated with different combinations of $ G$ and $c$.}
\label{fig4}
\end{figure}

\clearpage

\section{QQ-plots  for the simulaion experiment of Section 2.3}

We refer to Section 2.3 of the main paper for the asymptotic
distribution of the the first two moments of the sample eigenvalues. 
Here we report numerical results from a detailed simulation
experiment.

We adopt a PMD $ G=0.4\delta_1+0.6\delta_3$ and a ratio $c=0.5$. For this model, $v_2=5.8(1+\Delta)$, $\psi_{111}=11.6(2+\Delta)$, $\psi_{211}=0.96$, $\psi_{122}=1364.03 + 614.736 \Delta$, and $\psi_{222}=39.3216$.
Samples of $(z_{ij})$ are drawn from standard normal $N(0,1)$, scaled
$t$, i.e. $\sqrt{4/6}\cdot t_6$, standardized $\chi^2$,
i.e. $\sqrt{1/6}\cdot(\chi^2_3-3)$, and uniform distribution
$U(-\sqrt{3},\sqrt{3})$, where $\Delta=0,3,4,-1.2$,
respectively. Notice that the last three distributions have heavy
tail, skewed and heavy tail, and null  tail, respectively. The dimensions are fixed at $(p,n)=(200,400)$ and the number of independent replications is $10000.$

We exhibit QQ-plots of moment statistics normalized using the euqation
(18) of the main paper 
with respect to 
standard normal  $N(0,1)$ 
in Figure~\ref{QQplots} under the four distributions  for the base
variables $(z_{ij})$.
It shows that the empirical distributions of the statistics match the
standard normal very well.

\begin{figure}[hpt]
\centering
\includegraphics[width=4in]{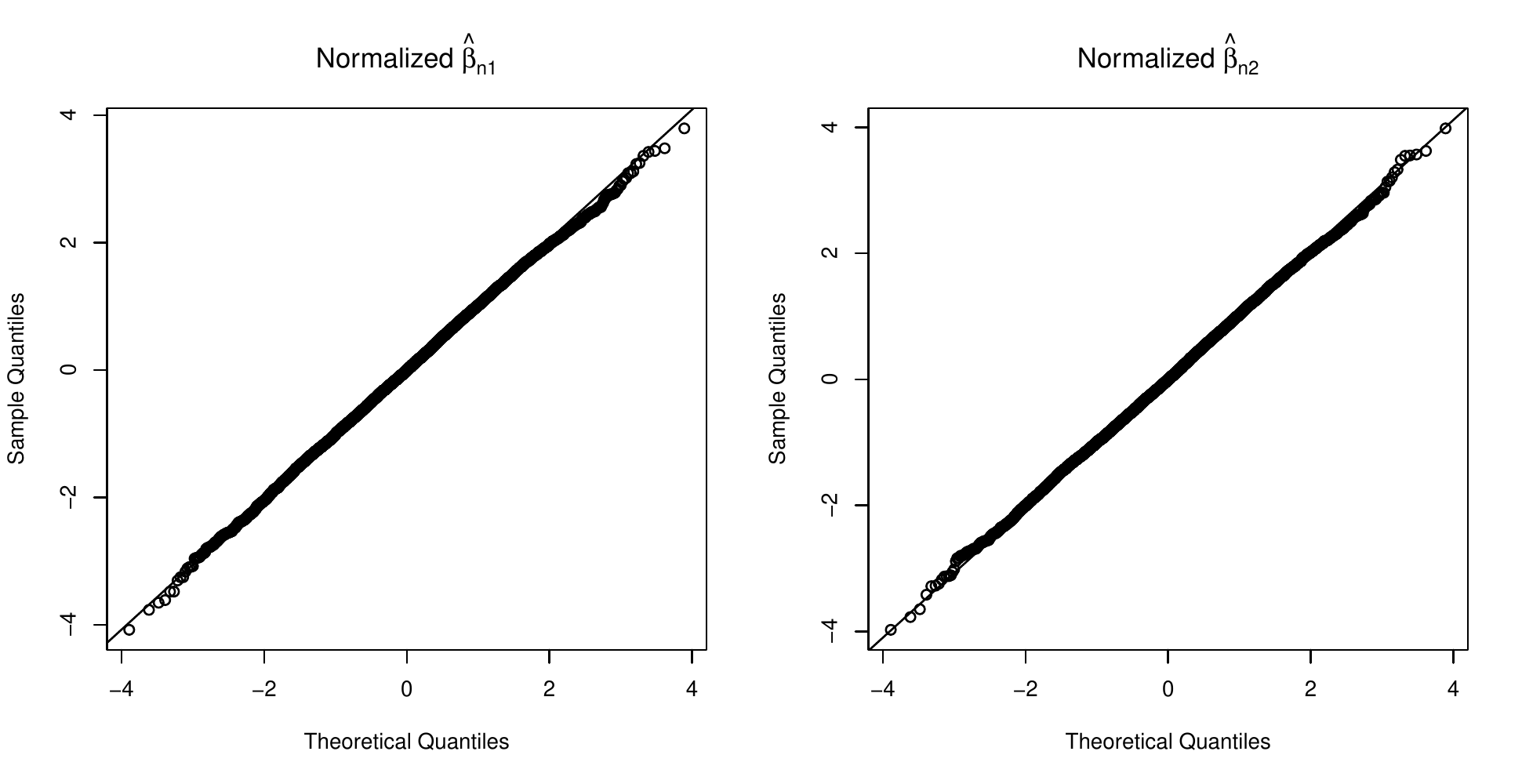}\\
\includegraphics[width=4in]{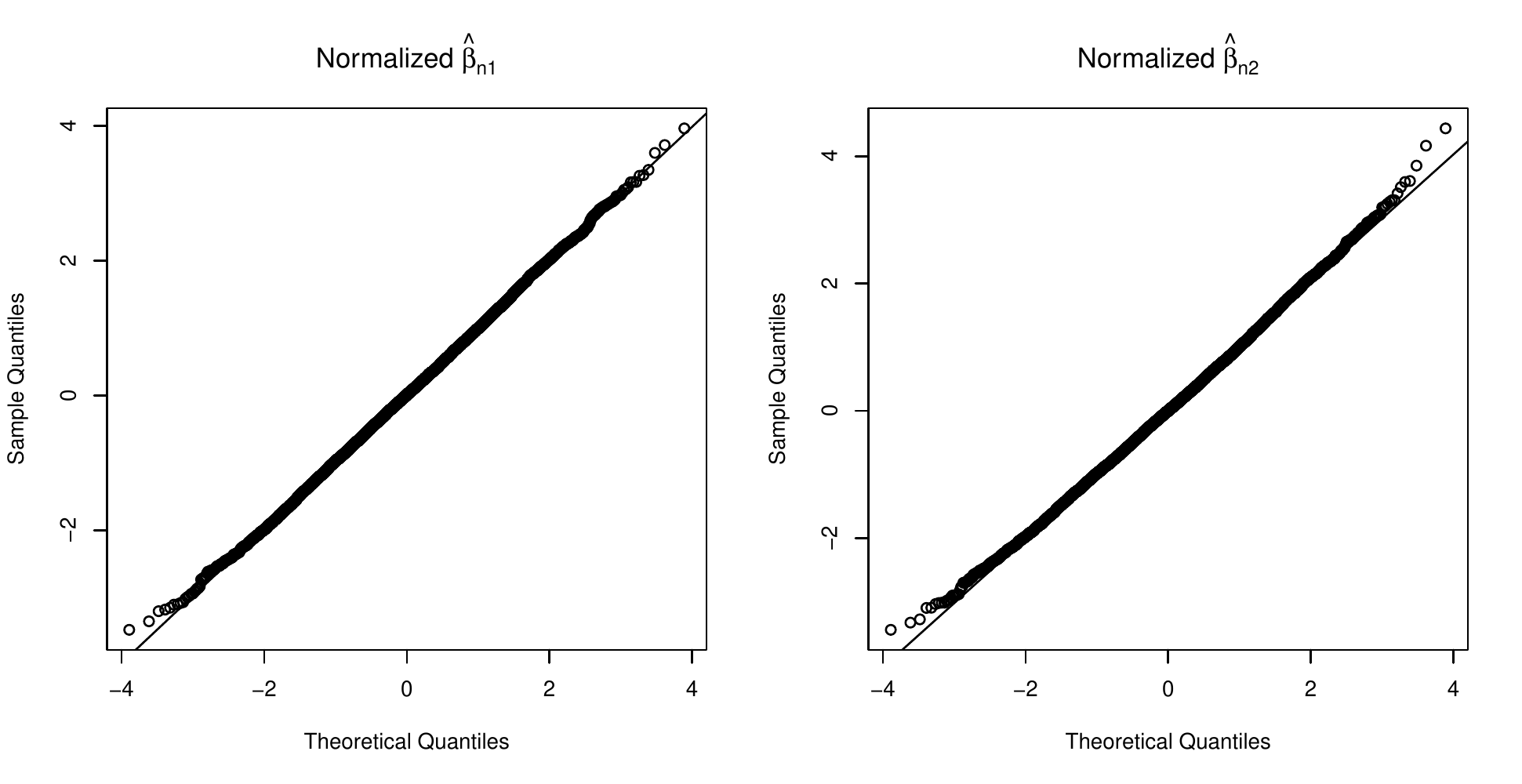} \\
\includegraphics[width=4in]{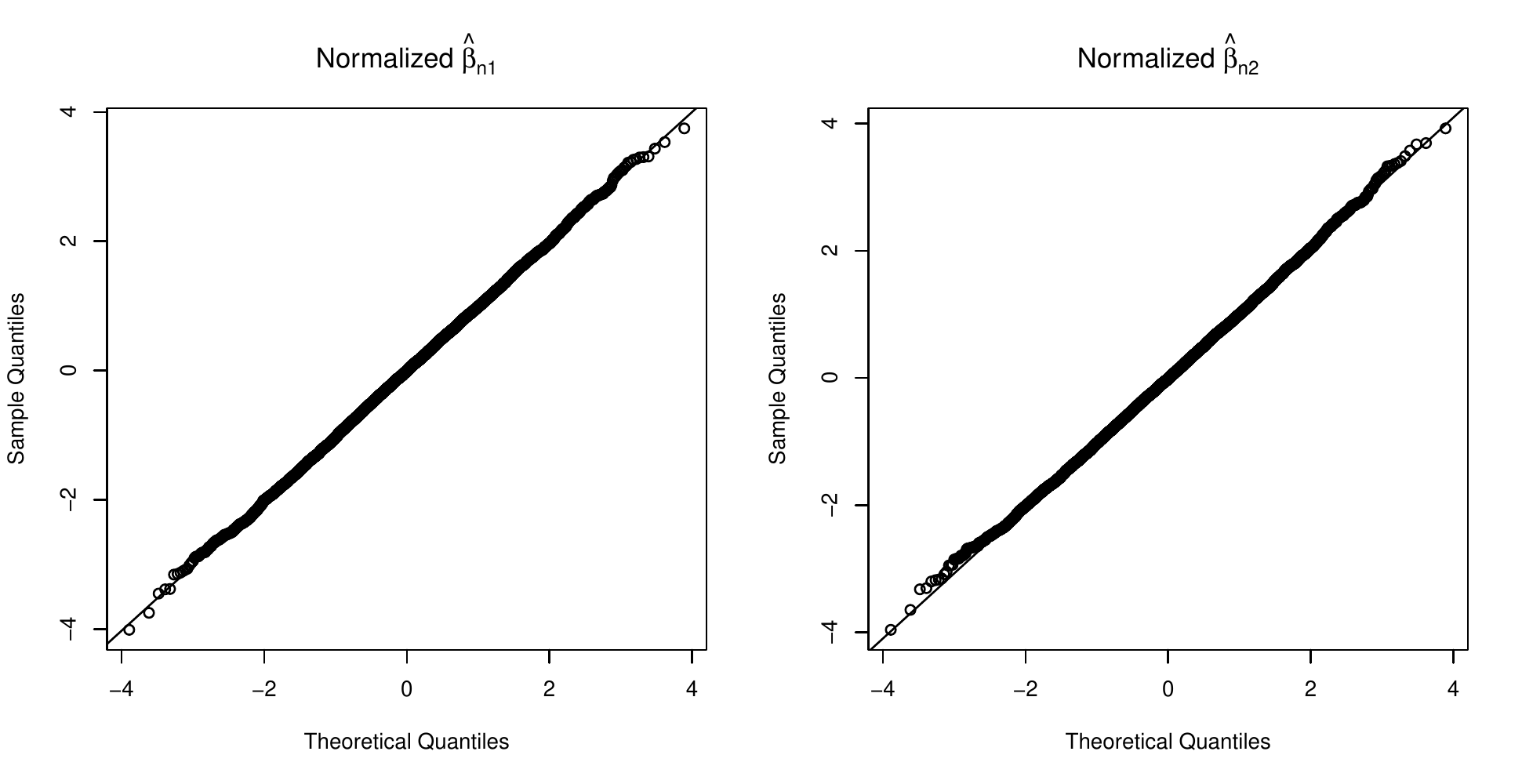} \\ 
\includegraphics[width=4in]{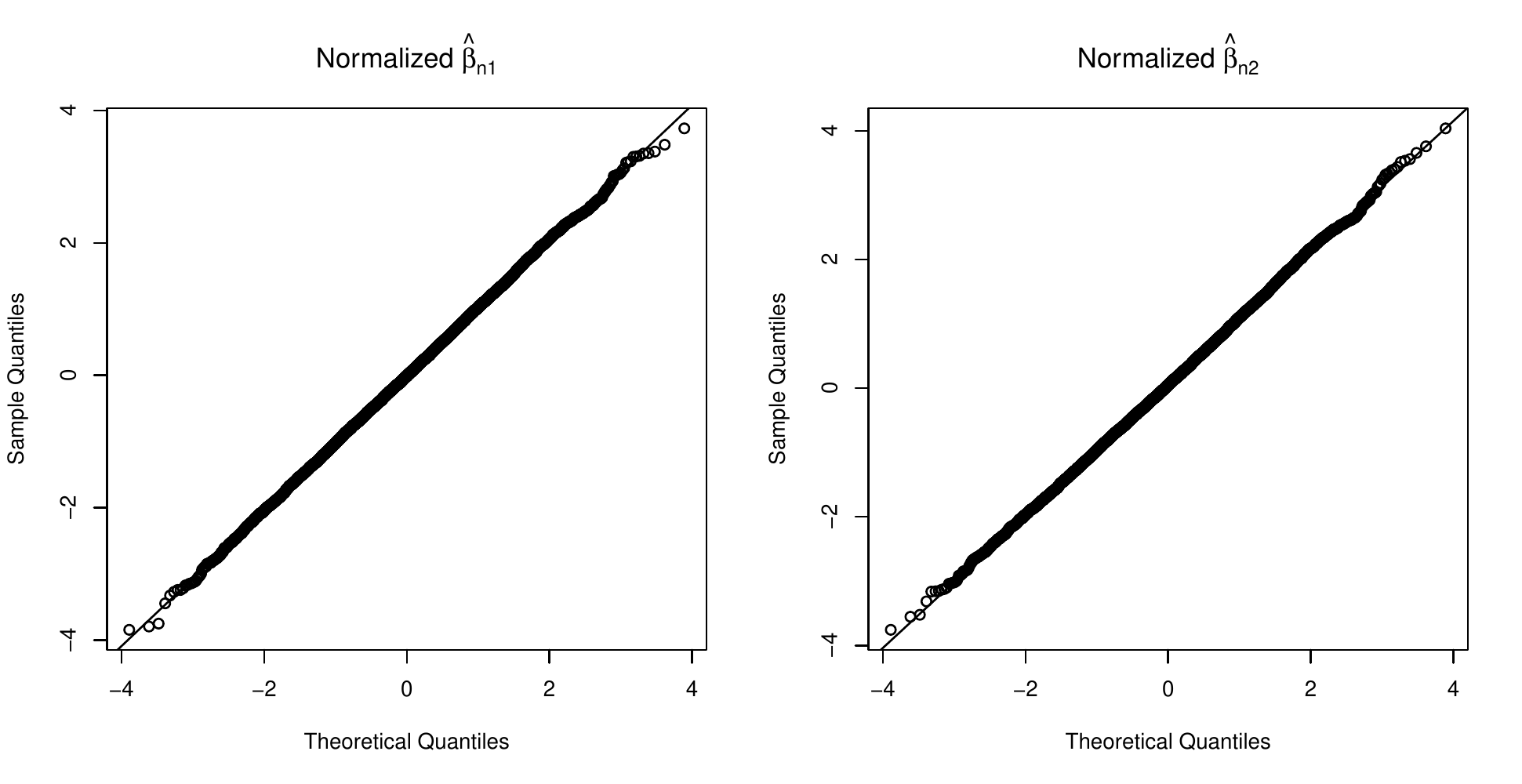} 
\caption{QQ-plots of normalized $\hat \beta_{n1}$ and $\hat \beta_{n2}$
  with respect to standard normal distribution
 under normal, Student-$t$, chi-square and uniform  population (top to
 bottom).\label{QQplots}}
\end{figure}

\section{Proof of  Lemma 1 of the main paper}

We follow the strategy developed in \cite{BS04}.
The convergence of $M_n(z)$ can be obtained by showing the following two facts:
\begin{itemize}
\item[] Fact 1: Finite dimensional convergence of $M_n(z)$ in distribution;
\item[] Fact 2: Tightness of $M_n(z)$ on $\mathcal C_n$.
\end{itemize}
\subsubsection{Finite dimensional convergence of $M_n(z)$ in distribution}
In this part we will show that for any positive integer $r$ and real constants $\alpha_1,\ldots,\alpha_r$, the sum
$$
\sum_{i=1}^r\alpha_i M_n(z_i)
$$
will converge in distribution to a Gaussian random variable.

Denote $m_{nn}=m_{F^{ c_n, G_n}} $ and $m_n=m_{F^{ c_n, G}} $, then these two Stieltjes transform satisfy
\begin{eqnarray*} 
  z  =  - \frac1 {m_{nn}}  +  \int\!\frac{t}{1+c_ntm_{nn} } d G_n(t),\quad z  =  - \frac1 {m_n}  +  \int\!\frac{t}{1+c_ntm_n} d G(t),  
\end{eqnarray*}
respectively. Taking the difference of the two identities yields
\begin{eqnarray*} 
\frac{m_n-m_{nn}}{m_n m_{nn}} &=& \int\!\frac{t}{1+c_ntm_{nn} } d G_n(t)-\int\!\frac{t}{1+c_ntm_n} d G(t),\\
&=&\int\!\frac{c_nt^2(m_n-m_{nn})}{(1+c_ntm_{nn})(1+c_ntm_n) } d G_n(t)+\int\!\frac{t[d G_n(t)-d G(t)]}{1+c_ntm_n}.
\end{eqnarray*}
Therefore, we get
\begin{eqnarray*} 
M_{n}(z)=\sqrt{n}(m_{nn}(z)-m_n(z))= \beta_n(z)\sqrt{n}\int\!\frac{t[d G_n(t)-d G(t)]}{1+c_ntm_n(z)},
\end{eqnarray*}
where $ \beta_n^{-1}(z)=\int c_nt^2/[(1+c_ntm_{nn}(z))(1+c_ntm_n(z))] d G_n(t)-1/(m_n(z)m_{nn}(z)).$
From \cite{SC95}, for any $z\in \mathcal C$, $1/|1+ctm(z)|$ is uniformly bounded in $t\in S_{ G}$. Notice that 
$$m_{nn}(z)\xrightarrow{a.s.}m(z),\quad m_n(z)\xrightarrow{a.s.} m(z), \quad  G_n(t)\xrightarrow{a.s.}  G(t),\quad c_n\rightarrow c,$$ 
then, for all $n$ large, almost surely,  the quantities $1/|1+c_ntm_{nn}(z)|$ and $1/|1+c_ntm_n(z)|$ are both uniformly bounded in $t\in S_{ G}$. From this and Lebesgue's dominated convergence theorem, 
$$
 \beta_n(z)\xrightarrow{a.s.}\left(\int\!\frac{ct^2}{(1+ctm(z))^2 } d G(t)-\frac{1}{m^2(z)}\right)^{-1}=-m'(z),
$$
as $n\rightarrow\infty.$  


Let $g(x,z)=x/(1+cxm(z))$, from the convergence of $c_n$, $m_n(z)$, and $ \beta_n(z)$, the linear combination 
$\sum_{i=1}^r\alpha_iM_n(z_i)$ has the same limiting distribution as 
\begin{eqnarray*}
&&-\sqrt{n}\sum_{i=1}^r\alpha_im'(z_i)\int g(t,z_i)\left(d G_n(t)-d G(t)\right)\\
&=&-\frac{1}{\sqrt{n}}\sum_{j=1}^n\sum_{i=1}^r\alpha_im'(z_i)\left(g(w_j^2,z_i)-Eg(w_j^2,z_i)\right),
\end{eqnarray*}
which is a  sum of centralized i.i.d.\ random variables with finite variance and thus converges in distribution to a zero-mean normal variable.
Moreover, for $1\leq i\neq j \leq r$,
\begin{eqnarray*}
&&\Cov\left[M_n(z_i), M_n(z_j)\right]\\
&=&nm'(z_i)m'(z_j)\Cov\left[\int g(t,z_i)d G_n(t), \int g(t,z_j) d G_n(t)\right]+o(1)\\
&=&m'(z_i)m'(z_j)\frac{1}{n}\sum_{k=1}^n\Cov\left[ g(w_k^2,z_i), g(w_k^2,z_j)\right]+o(1)\\
&=&m'(z_i)m'(z_j)\left(\E(g(w^2,z_i)g(w^2,z_j)-\E g(w^2,z_i)\E g(w^2,z_j))\right)+o(1)\\
&\to&m'(z_i)m'(z_j)\left(\int g(t,z_i)g(t,z_j)d G(t)-\int g(t,z_i)d G(t)\int g(t,z_j)d G(t)\right)\\
&=&m'(z_i)m'(z_j)\left(\frac{z_i+1/m(z_i)-z_j-1/m(z_j)}{c(m(z_j)-m(z_i))}-\frac{(1+z_im(z_i))(1+z_jm(z_j))}{m(z_i)m(z_j)}\right),
\end{eqnarray*}
as $n\to\infty$, where the last equality is obtained from the equation 
(10) of the main paper.

\subsubsection{Tightness of $M_n(z)$}

The tightness of $M_n(z)$ on $\mathcal C_n$ can be established by verifying the moment condition (12.51) of \cite{B68}, i.e.,
\begin{equation}\label{tightness}
\sup_{n,z_1,z_2\in \mathcal C_n}\frac{\E|M_n(z_1)-M_n(z_2)|^2}{|z_1-z_2|^2}<\infty.
\end{equation}

Taking the partial derivative of $g(x,z)$ with respect to $z$, we get
$$
g'_z(x,z):=\frac{\partial{g(x,z)}}{\partial z}=-\frac{cx^2m'(z)}{(1+cxm(z))^2}<K,
$$
where $K$ is an upper bound of $g'_z(x,z)$ on $S_{ G}\times \mathcal C$. From this, for any $z_1,z_2\in\mathcal C$ and $x\in S_{ G}$, there is a constant $\xi$ such that  
$$
|g(x,z_1)-g(x,z_2)|\leq |g_z'(x,\xi)||z_1-z_2|\leq K|z_1-z_2|,
$$
where $\xi=z_1+\theta(z_2-z_1)$ and $\theta\in (0,1)$.
Let $\tilde g(w^2,z)=g(w^2,z)-\E(g(w^2,z))$, we have then
\begin{eqnarray*}
\E\bigg|\frac{M_n(z_1)-M_n(z_2)}{z_1-z_2}\bigg|^2
&=&\frac{1}{n|z_1-z_2|^2}\E\bigg|\sum_{j=1}^n\tilde g(w_j^2,z_1)- \tilde g(w_j^2,z_2)\bigg|^2+o(1)\\
&\rightarrow&\frac{\E|\tilde g(w_1^2,z_1)- \tilde g(w_1^2,z_2)|^2}{|z_1-z_2|^2}\\
&\leq&\E\bigg|\frac{g(w_1^2,z_1)-g(w_1^2,z_2)}{z_1-z_2}\bigg|^2\leq K^2,\\
\end{eqnarray*}
which confirms the inequality in \eqref{tightness}.

\end{document}